\documentclass[12pt,journal,onecolumn,draftcls]{IEEEtran}
\usepackage{epsfig}
\usepackage{times}
\usepackage{float}
\usepackage{afterpage}
\usepackage{amsmath}
\usepackage{amstext}
\usepackage{amssymb,bm}
\usepackage{latexsym}
\usepackage{color}
\usepackage{graphicx}
\usepackage{amsmath}
\usepackage{amsthm}
\usepackage{graphicx}
\usepackage[center]{caption}
\usepackage{pstricks}
\usepackage{caption}
\usepackage{subcaption}
\usepackage{booktabs}
\usepackage{multicol}
\usepackage{lipsum}% dummy text

\allowdisplaybreaks

\def \eu{\mathrm{e}}

\newcommand{\snr}{{\mathsf{S}}}
\newcommand{\inr}{{\mathsf{I}}}
\newcommand{\mud}{{\mathsf{I_d}}}
\newcommand{\mug}{{\mathsf{I_g}}}
\newcommand{\gap}{{\mathsf{gap}}}
\newcommand{\pam}{{\mathsf{PAM}}}
\newcommand{\points}{{\mathsf{N_d}}}
\newcommand{\supp}{{\mathsf{supp}}}
\newcommand{\co}{{\mathsf{co}}}
\newcommand{\loglog }{O\left(\log \ln \left(\min(\snr,\inr) \right)\right)}
\newcommand{\loglogG}{O\left(\log \left( \frac{\ln \min(\snr,\snr^\alpha) }{\gamma} \right)\right)}

\newtheorem{thm}{Theorem}%[section]

\newtheorem{lem}[thm]{Lemma}
\newtheorem{prop}[thm]{Proposition}
\newtheorem{rem}{Remark}

\begin{document}

\title{Interference as Noise: Friend or Foe?
%Treating Interference as Noise without Time Sharing is Optimal to within an additive Gap for the Symmetric Two-User Gaussian Interference Channel
}
\author{
\IEEEauthorblockN{Alex Dytso, Daniela Tuninetti, and Natasha Devroye\\}
%\authorrefmark{1}
%\medskip
\IEEEauthorblockA{%
%Department of Electrical and Computer Engineering\\
University of Illinois at Chicago, Chicago IL 60607, USA,\\
Email: {\tt odytso2, danielat, devroye @ uic.edu}}%
\thanks{%
The work of the authors was partially funded by NSF under awards 1017436 and 1422511;
the contents of this article are solely the responsibility of the authors and do not necessarily represent the official views of the NSF.
The work was presented in part at ISIT 2014, ITA 2015, and ISIT 2015.}%
}
\maketitle

\begin{abstract}
This paper shows that for the %symmetric
two-user Gaussian Interference Channel (G-IC) %with additive white  noise 
Treating Interference as Noise without Time Sharing (TINnoTS) achieves the closure of the capacity region to within either a constant gap, or to within a gap of the order $\loglog$ where $\snr$ is the largest Signal to Noise Ratio (SNR) on the direct links and $\inr$ is  the largest Interference to Noise Ratio (INR) on the cross links. As a consequence, {TINnoTS} is optimal from a generalized Degrees of Freedom (gDoF) perspective for all channel gains except for a subset of zero measure. %Generalizations to some asymmetric settings are also provided.
{TINnoTS} with Gaussian inputs is known to be optimal to within 1/2~bit for a subset of the weak interference regime. 
Rather surprisingly, this paper shows that {TINnoTS} is gDoG optimal in {\em all} parameter regimes,
even in the strong and very strong interference regimes where joint decoding of Gaussian inputs is optimal.

For approximate optimality of {TINnoTS} in all parameter regimes it is critical to use non-Gaussian inputs.
%The main contribution of %as follows. Since Gaussian inputs are not well suited to be treated as noise, this paper proposes to 
This work thus proposes to use {\it mixed inputs} as channel inputs for the G-IC,
where a mixed input is the sum of a discrete and a Gaussian random variable.
%, where the discrete and continuous part of a mixed input mimic the common and private message, respectively, in the Han-Kobayashi achievable scheme.
Interestingly, compared to the Han-Kobayashi achievable scheme,
the discrete part of a mixed input is shown to effectively behave as a common message in the sense that, 
although treated as noise, its effect on the achievable rate region is as if it were jointly decoded
together with the desired messages at a non-intended receiver.
The practical implication is that a discrete interfering input is a ``friend'', 
while an Gaussian interfering input is in general a ``foe''.
The paper also discusses other practical implications of the prosed {TINnoTS} scheme with mixed inputs.

Since {TINnoTS} requires neither explicit joint decoding nor time sharing, 
the results of this paper are applicable to a variety of %channels of 
{\it oblivions} or {\it asynchronous} channels, 
such as the block asynchronous G-IC (which is not an information stable channel) and the G-IC with partial codebook knowledge at one or more receivers.

% MAYBE THE FOLLOWING IS TOO MUCH DETAILED FOR AN ABSTRACT
%The use of mixed inputs poses a number of technical challenges that this paper addresses:
%(1) a closed-form lower bound on the rates achievable with mixed inputs is obtained, and
%(2) bounds on the cardinality and the minimum distance of {\it sum-set} constellations,
%obtained as the sum of two discrete random variables, are derived.
%These bounds are then optimized
%%as inputs and optimizing over 
%with respect to the number of points of the discrete part of the mixed inputs, 
%as well as with respect to the power split among the discrete and the continuous part of the mixed inputs.
%The optimality to within a gap of {TINnoTS} is proved by showing 
%that each point on the closure of the Etkin-Tse-Wang outer bound to the capacity region of the
%symmetric G-IC to can be attained with pair of mixed inputs.
\end{abstract}

\section{Introduction}
\label{sec:Intro}
Consider the two-user memoryless real-valued additive white Gaussian noise interference channel (G-IC) with input-output relationship
\begin{subequations}
\begin{align}
Y_1^n&=h_{11}X_1^n+h_{12}X_2^n+Z_1^n, \\ %\sqrt{\snr}X_1^n+\sqrt{\inr}X_2^n+Z_1^n\\
Y_2^n&=h_{21}X_1^n+h_{22}X_2^n+Z_2^n,    %\sqrt{\inr}X_1^n+\sqrt{\snr}X_2^n+Z_2^n
\end{align}
\label{eq:block awgn ic}
where
$X_j^n := (X_{j1}, \cdots X_{jn})$ and $Y_j^n := (Y_{j1}, \cdots Y_{jn})$ are the length-$n$ vector input and output, respectively, for user $j\in[1:2]$, the noise vector $Z_j^n$ has i.i.d. zero-mean unit-variance Gaussian components,
the input $X_j^n$ is subject to a per-block power constraint $\frac{1}{n}\sum_{i=1}^n X_{ji}^2 \leq 1$, %, for $j\in[1:2]$.
and the channel gains $(h_{11},h_{12},h_{21},h_{22})$ are fixed and known to all nodes.
The input $X_j^n, \ j\in[1:2],$ carries the independent message $W_j$ that is uniformly distributed on $[1:2^{n R_j}]$, where $R_j$ is the rate and $n$ the block-length. Receiver $j\in[1:2]$ wishes to recover $W_j$ from the channel output $Y_j^n$ with arbitrarily small probability of error. Achievable rates and capacity region are defined in the usual way~\cite{elgamalkimbook}. We shall denote the capacity region by $\mathcal{C}$. 
\end{subequations}

\begin{subequations}
For simplicity we will focus primarily on the {\it symmetric} G-IC defined by
\begin{align}
&|h_{11}|^2=|h_{22}|^2=\snr \geq 0, \\
&|h_{12}|^2=|h_{21}|^2=\inr \geq 0,
\end{align}
and we will discuss how the results for the symmetric G-IC extend to the general asymmetric setting.
\label{eq:awgn sym}
\end{subequations}

The general {\it discrete memoryless} IC was introduced in~\cite{Ahlswede:1973} where it was shown  that the capacity region of an {\em information stable IC}~\cite{dobrushin1963general} is given by
%\begin{subequations}
\begin{align}
%C(\snr,\inr) =\left \{\begin{array}{ll}
%C := 
%\max \limits_{n\in\mathbb{N}, P_{X_1^n,X_2^n}=P_{X_1^n}P_{X_2^n}} \frac{1}{n}\Big(I(X_1^n;Y_1^n) +I(X_2^n;Y_2^n)\Big). 
%\\
%\\
%\text {s.t. }   \frac{1}{n}\sum_{i=1}^n X_{1i}^2 \leq 1,   \frac{1}{n}\sum_{i=1}^n X_{2i}^2 \leq 1.
%\end{array} \right.
\mathcal{C}
&=\lim_{n \to \infty} 
\co\left(\bigcup_{P_{X_1^nX_2^n}=P_{X_1^n}P_{X_2^n}}
\left\{ %(R_1,R_2): 
\begin{array}{l} 
0\leq R_1 \leq \frac{1}{n} I(X_1^n;Y_1^n)\\
0\leq R_2 \leq \frac{1}{n} I(X_2^n;Y_2^n)\\
\end{array} \right \} \right),
\label{eq:capacity info stable IC}
%\\
%\mathcal{C}_{\Sigma}&=\max \left \{ R_1+R_2: (R_1,R_1) \in \mathcal{C} \right \}. \label{eq:opt:SumCap}
\end{align}
%\end{subequations}
where $\co$ denotes the convex closure operation. For the G-IC in~\eqref{eq:block awgn ic}, the maximization in~\eqref{eq:capacity info stable IC} is further restricted to inputs satisfying the power constraint.

An inner bound to the capacity region in~\eqref{eq:capacity info stable IC} %and sum-capacity
can be obtained by considering i.i.d. inputs in~\eqref{eq:capacity info stable IC} thus giving
%\begin{subequations}
\begin{align}
\mathcal{R}_{\text{in}}^\text{TIN+TS} %(\snr,\inr)
&=
\co\left(\bigcup_{P_{X_1X_2}=P_{X_1}P_{X_2}}
\left\{ %(R_1,R_2): 
\begin{array}{l} 
0\leq R_1 \leq I(X_1;Y_1)\\
0\leq R_2 \leq I(X_2;Y_2)\\
\end{array} \right \} \right),
\label{eq:RL:TINwithT}
%\mathcal{R}_{\text{L},\Sigma}(\snr,\inr)&=
%\left \{\begin{array}{ll}
%\max \limits_{P_{X_1,X_2}=P_{X_1}P_{X_2}} \Big(I(X_1;Y_1) +I(X_2;Y_2)\Big),
% \\
%\\
%\text {s.t. }   \frac{1}{n}\sum_{i=1}^n X_{1i}^2 \leq 1,   \frac{1}{n}\sum_{i=1}^n X_{2i}^2 \leq 1,
%\end{array} \right.
\end{align}
%\end{subequations}
where the superscript ``TIN+TS'' reminds the reader that the region is achieved by Treating Interference as Noise (TIN)\footnote{We use the terminology ``treating interference as noise'' to denote the rates obtained when evaluating expressions for the interference channel of the form 
\[ 
\mbox{Desired rate} \leq I(\mbox{desired input}; \mbox{output}),
\] 
without any other rate expressions, mutual information terms,  or explicit rate splits. When evaluated with independent and identically distributed (i.i.d.) Gaussian inputs in the interference channel in \eqref{eq:block awgn ic}, these rate expressions look like those in which the interference is indeed treated as noise, i.e., 
\[ 
0\leq R_i \leq \frac{1}{2}\log\left(1+\frac{\snr}{1+\inr}\right), \ i\in[1:2],
\]
where the `effective noise' (at the denominator within the log) looks like the true noise power plus all the interferer's power. Whether this expression has the same ``treating interference as noise'' interpretation when using non-Gaussian inputs is open to interpretation, and is one of the focusses of this work. We will however continue to use this terminology.
} 
and with Time Sharing (TS), where TS is enabled by the convex hull operation~\cite{elgamalkimbook}. 
By further removing the convex hull operation in~\eqref{eq:RL:TINwithT} we arrive at 
%\begin{subequations}
\begin{align}
\mathcal{R}_{\text{in}}^\text{TINnoTS} %(\snr,\inr)
&=
\bigcup_{P_{X_1X_2}=P_{X_1}P_{X_2}}
\left\{ %(R_1,R_2): 
\begin{array}{l} 
0\leq R_1 \leq I(X_1;Y_1)\\
0\leq R_2 \leq I(X_2;Y_2)\\
\end{array} \right \}.
\label{eq:RL:TINnoTS}
\end{align}
%\end{subequations}
The region in~\eqref{eq:RL:TINnoTS} does not allow the users to time-share.
For the G-IC the maximization in~\eqref{eq:RL:TINwithT} and~\eqref{eq:RL:TINnoTS} is further restricted to inputs satisfying the power constraint.

Obviously 
\begin{align*}
\mathcal{R}_{\text{in}}^\text{TINnoTS} \subseteq \mathcal{R}_{\text{in}}^\text{TIN+TS}\subseteq \mathcal{C}.
\end{align*}
The question of interest in this paper is how $\mathcal{R}_{\text{in}}^\text{TINnoTS}$ fares compared to $\mathcal{C}$.
Note that there are many advantages in using {TINnoTS} in practice. For example, {TINnoTS} 
does not require codeword synchronization, as for example for joint decoding or interference cancellation,
and does not require much coordination between users, thereby reducing communications overhead. % On the other hand, {TINnoTS} 
The goal of this paper is to show that despite its simplicity, {TINnoTS} approximately archives the capacity $\mathcal{C}$.
%is not too bad in general.

Next, we review past work relevant to our investigation. We refer the interested reader to~\cite{elgamalkimbook} for a comprehensive literature survey on general discrete memoryless  ICs.

\subsection{Past Work}

In general, little is known about the optimizing input distribution in~\eqref{eq:capacity info stable IC} for the G-IC (or in~\eqref{eq:RL:TINwithT} and in~\eqref{eq:RL:TINnoTS}) and only some special cases have been solved.  In~\cite{ICsumCapacityKramer} it was showen that i.i.d. Gaussian inputs maximize the sum-capacity in~\eqref{eq:capacity info stable IC} for $\sqrt{\frac{\inr}{\snr}}(1+\inr)\leq \frac{1}{2}$ in the symmetric case. 
In contrast, the authors of~\cite{LimCharMemChan} showed that in general multivariate Gaussian inputs do not exhaust regions of the form in~\eqref{eq:capacity info stable IC}. 
The difficulty arises from the competitive nature of the problem~\cite{CoordinateSystemAbbe}: for example, say $X_2$ is i.i.d. Gaussian, taking $X_1$ to be Gaussian increases $I(X_1;Y_1)$ but simultaneously decreases $I(X_2;Y_2)$, as Gaussians are known to be the ``best inputs'' for Gaussian point-to-point power-constrained channels, but are also the ``worst noise'' (or interference, if it is treated as noise) for a Gaussian input. 

Recently in~\cite{DytsoITA2014,DytsoCodebookJournal}, for the G-IC with one oblivious receiver, we showed that a properly chosen discrete input has a somewhat different behavior than a Gaussian input: a discrete $X_2$ may yield a ``good'' $I(X_1;Y_1)$ while keeping $I(X_2;Y_2)$ relatively unchanged compared to a Gaussian input, thus substantially improving the rates compared to Gaussian inputs in the same achievable region expression. Moreover, in~\cite{DytsoITA2014,DytsoCodebookJournal} we showed that treating interference as noise at the oblivious receiver and joint decoding at the other receiver is to within an additive gap of 3.34~bits of the capacity. 
In this work we seek to {\it analytically} evaluate the lower bound in~\eqref{eq:RL:TINnoTS}  for a special class of {\em mixed inputs} (a superposition of a Gaussian and a discrete random variable) by generalizing the approach of~\cite{DytsoITA2014,DytsoCodebookJournal} and show that using {TINnoTS} at both receivers is to within an additive gap of the capacity.  In a way this work follows the philosophy of~\cite{inPraiseOfBadCodes}: the main idea is to use sub-optimal point-to-point codes in which the reduction in achievable rate for the intended receiver is more than compensated by the decrease in the interference created at the other receiver, which results in an overall rate region improvement.
In a conference version of this paper~\cite{dytsoTINconf} we demonstrated that {TINnoTS} is gDoF optimal and can achieve to within an additive gap  the symmetric sum-capacity of the classical G-IC. In~\cite[Theorem 3.]{Shuo-ISIT2015-TIN}, the authors showed that the sum-capacity result of~\cite{dytsoTINconf} can be achieved by an input with purely discrete marginals, i.e., the Gaussian part of our mixed inputs can be replaced by a discrete random variable.

Recently there has been lots of interest in characterizing when TIN, with or without TS, is approximately optimal. For example, in~\cite{TINweakKuser} 
``It is shown that in the $K$-user interference channel, if for each user the desired signal strength is no less than the sum of the strengths of the strongest interference from this user and the strongest interference to this user (all values in dB scale), then the simple scheme of using point to point Gaussian codebooks with appropriate power levels at each transmitter and treating interference as noise at every receiver (in short, TIN scheme) achieves all points in the capacity region to within a constant gap. The generalized degrees of freedom (gDoF) region under this condition is a polyhedron, which is shown to be fully achieved by the same scheme, without the need for time-sharing.''
In this paper we aim to show that one can always use {TINnoTS} and be optimal to within an additive gap in all parameter regimes, and not just in the very weak interference regime identified in~\cite{TINweakKuser}. The key is to use more ``friendly'' codebooks than Gaussian codebooks.
We note that for an input constrained additive-noise channel where the noise distribution is arbitrary, Gaussian inputs are known to be optimal to within 0.265~bits~\cite{ZamirGaussianIsNotToBad}; what our work shows is that the same is not true in general in a multiuser competitive scenario.

We are not the first to consider discrete inputs for Gaussian noise channels. Shannon himself pointed out
the asymptotic optimality of a Pulse Amplitude Modulation (PAM) input for the point-to-point power-constrained Gaussian noise channel~\cite[ref.121]{ShamaiShannonLecture}. 
Shannon's argument was solidified in~\cite{PAMozarow} where firm lower bounds on the achievable rate with a PAM input were derived and used to show their optimality to within 0.41~bits~\cite[eq.(9)]{PAMozarow}.
In~\cite{CoordinateSystemAbbe,totAsynchIC} the authors demonstrated the existence of input distributions that outperform i.i.d. Gaussian inputs in $\mathcal{R}_{\text{in}}^\text{TINnoTS}$ for certain asynchronous G-IC.
Both~\cite{CoordinateSystemAbbe,totAsynchIC} used local perturbations of an i.i.d. Gaussian input:~\cite[Lemma 3]{totAsynchIC} considered a fourth order approximation of mutual information, while~\cite[Theorem 4]{CoordinateSystemAbbe} used perturbation in the direction of Hermite polynomials of order larger than three. In both cases the input distribution is assumed to have a density.
%, though~\cite[Fig. 1]{totAsynchIC} numerically shows the performance of a ternary PAM input as well. 
For the cases reported in~\cite{totAsynchIC, CoordinateSystemAbbe}, the improvement over i.i.d. Gaussian inputs shows in the decimal digits of the achievable rates; it is hence not clear that  
% perturbed continuous Gaussian inputs as in~\cite{totAsynchIC, CoordinateSystemAbbe}
these classes of inputs can actually provide substantial rate gains compared to Gaussian inputs. 
In~\cite{MoshksarJournal}, the authors showed that 
%a novel approach, of  
using a discrete input for one user and a Gaussian input for the other user
 %{\green REMOVE in the  asynchronous IC} setting    have been shown to 
outperforms Gaussian inputs at both users in the {TINnoTS} region;
%~\eqref{eq:RL:TINnoTS} 
this gain however was neither shown to be unbounded nor to achieve the capacity to within a gap, as we will do here.
Moreover, the approach of~\cite{MoshksarJournal} was based on bounding the achievable mutual information by using Fano's inequality, similarly to~\cite[Part a)]{PAMozarow}; the resulting bounds however will not be tight enough for the purposes of deriving gap results. In this work we generalize the bound due to Ozarow-Wyner in~\cite[Part b)]{PAMozarow}, which turns out to sharper than~\cite[Part a)]{PAMozarow}. % bounds due to Fano's inequality. 

We remark that the optimality of {TINnoTS} for all channel parameters for the G-IC was pointed out in~\cite[Remark 6.12]{elgamalkimbook}. 
The proof follows since {TINnoTS} is always optimal for the Linear Deterministic Approximation (LDA) of the G-IC at high-SNR~\cite{avestimehr2011wireless}.
Moreover, a scheme for the LDA can be translated into a scheme for the real-valued G-IC that is optimal to within at most 18.6~bits~\cite[Theorem~2]{bresler_tse}.
This line of reasoning based on a universal gap between the LDA and the G-IC, thus giving a constant gap result, does not provide a concrete practical construction of an approximately optimal scheme. The idea of `lifting' an LDA optimal scheme to the G-IC has been used in~\cite{Shuo-ISIT2015-TIN} where a $O(\log\log(\snr))$, rather than a constant, gap result was proved for the symmetric sum-capacity. 
Our proof here extends our original approach in~\cite{dytsoTINconf} %is constructive and 
and provides, in closed form, the optimal number of points in the discrete part of the mixed inputs, as well as of the optimal power split among the discrete and continuous parts of the mixed inputs. Moreover, our derived gap is in general smaller than 18.6~bits (this is so because the log-log function grows very slowly in its argument).
%the one found in~\cite{bresler_tse}.

We conclude this overview of relevant past work by pointing out that in practice it is well known that a non-Gaussian interference should not be treated as a Gaussian noise. The optimal detector for an additive non-Gaussian noise channel may however be far more complex than a classical minimum-distance decoder. Nonetheless, since the performance increase can be substantial for a moderate complexity increase, Network-Assisted Interference Cancellation and Suppression (NAICS) receivers, which account for the discrete and coded nature of the interference, were adopted in the Long Term Evolution (LTE) Advanced Release 12~\cite{docomo2012requirements,etsi2012evolved,mediatek2012broadcom}. 
The boost in performance of NAICS-type detectors may be understood as follows.
As we pointed out in~\cite{DytsoCodebookJournal}, with TIN the mapping of the codewords to the messages is unknown but the codeword symbols may be known through soft symbol-by-symbol estimation as remarked in~\cite{sand_decentr_proces}, where the authors write ``We indeed see that BPSK signaling outperforms Gaussian signaling. This is because demodulation is some form of primitive decoding, which is not possible for the Gaussian signaling.'' 
This interpretation is supportd somehow by~\cite[eq.(16)]{BelfioreLing:IZS12:flatnessfactor}, where
the authors showed that with one-dimensional lattices it holds that
$\arg\max_{n_1\in\mathbb{Z}}\varrho(n_1)
=\arg\min_{n_1\in\mathbb{Z},n_2\in\mathbb{Z}} |y-h_1 n_1-h_2 n_2|^2$, where $\varrho(n_1) := \sum_{n_2\in\mathbb{Z}}e^{-|y-h_1 n_1-h_2 n_2|^2}$ is the metric to be optimized in decoding lattice $n_1\in\mathbb{Z}$ while treating lattice $n_2\in\mathbb{Z}$ as noise on a multi-user Gaussian noise channel with channel gains $(h_1,h_2)$; note that $|y-h_1 n_1-h_2 n_2|^2$ is the matric to be optimized on the same channel when the two lattices are jointly decoded. An interesting question is whether the same holds for ``truncated latices,'' such as PAM constellations used here, which appears reasonable when any pair $(n_1,n_2)$  results in a distinct linear combination $h_1 n_1+h_2 n_2$.

\subsection{Contributions and Paper Outline}

The main contributions of the paper are as follows: % and organization 
\begin{enumerate}

\item In Section~\ref{sec:Ozarow}, Proposition~\ref{prop:lowbound OW generalized} presents a generalization of a lower bound from~\cite{PAMozarow} on the mutual information attained by a discrete input on a point-to-point additive noise channel and compares its performance with other lower bounds available in the litterature

\item In Section~\ref{sec:carAndMinD}, Proposition~\ref{prop:combNOOUTAGE} and Proposition~\ref{prop:card:dmin:measurebound} present new bounds on the cardinality and minimum distance of {\it sum-sets} formed by two discrete constellations.
Proposition~\ref{cor: mini dist 0} shows that the set of channel gains for which
% and show that almost surely 
the cardinality of a sum-set is not equal to the product of the cardinalities of the constituent sets
%this does not hold 
has zero measure.

\item Section~\ref{sec:ptpExample} provides examples of how we intend to use the developed tools. 
First, we show that discrete inputs are approximately optimal for the point-to-point power-constrained Gaussian channel. 
Second, we show that a discrete additive state, unknown to both the transmitter and the receiver, degrades performance of a point-to-point power-constrained Gaussian channel by at most a constant gap compared to the case where the state is known at all terminals.

\item In Section~\ref{sec:ach}, Proposition~\ref{prop:ach-with-mixedinput} presents an inner bound obtained by evaluating the {TINnoTS} region %$\mathcal{R}_{\text{in}}^\text{TINnoTS}$ 
with our proposed mixed inputs, whose performance will then be compared to the outer bound in Proposition~\ref{prop:ETWouter}.

\item \label{item:sym}
Section~\ref{sec:cap within gap sym} focuses on the symmetric G-IC.
Theorem~\ref{thm:Cap:gap sym} shows that {TINnoTS} with mixed inputs is to within $O(1)$, or $\loglog$ except for a set of Lebesgue measure $\gamma$ for any $\gamma \in (0,1]$, of the outer bound  in Proposition~\ref{prop:ETWouter}.
%given in~\cite{etkin_tse_wang}.
%Important practical implications of this result are that:
From this result we infer that:
\begin{enumerate}
\item the discrete part of the mixed input behaves as a ``common message'' whose contribution can be removed from the channel output of the non-intended receiver, even though explicit joint decoding of the interference is not employed in {TINnoTS},
\item the continuous part of the mixed input behaves as a ``private message'' whose power should be chosen such that it is either received below the noise floor of the non-intended receiver~\cite{etkin_tse_wang}, or to have a rate that is approximately half  the target rate,
and
\item time-sharing may be mimicked by varying the number of points in the  discrete part of the mixed inputs.
\end{enumerate}

\item \label{item:some asym}
In Section~\ref{sec:cap within gap asymmetric} we extend the gap result of Theorem~\ref{thm:Cap:gap sym}
%some of the gap results in Section~\ref{sec:cap within gap sym} 
to some general asymmetric G-IC's. The channel parameter regime covered in Theorem~\ref{thm:Cap:gap some asymmetric} is such that
%the rate region does not comprise rate bounds of the form $2R_{1+R_2$ or $R_1+2R_2$, i.e., such 
bounds of the form $2R_1+R_2$ or $R_1+2R_2$ are not active in the outer bound in Proposition~\ref{prop:ETWouter}.
The excluded regime, roughly speaking, is such that $\min(|h_{11}|^2_{\text{dB}},|h_{22}|^2_{\text{dB}}) < |h_{12}|^2_{\text{dB}}+|h_{21}|^2_{\text{dB}} < |h_{11}|^2_{\text{dB}}+|h_{22}|^2_{\text{dB}}$, i.e., the sum of the crosslink gains is upper bounded by the sum of the direct link gains and lower bounded by the minimum of the direct link gains, all quantities expressed in dB scale. Numerical experiments suggest that the insights gained in the symmetric case (see above item~\ref{item:sym}) hold for the asymmetric case as well and that the proposed {TINnoTS} with mixed inputs is approximately optimal for the general asymmetric G-IC.

\item In Section~\ref{sec:gDoF:cap}, Theorem~\ref{thm:gDoF:cap} shows that {TINnoTS} with mixed inputs is gDoF optimal almost everywhere (a.e.), that is, for all channel gains except for an {\it outage set} of zero measure. 

\item In Section~\ref{sec: async and obl IC} shows that our approximate optimality results hold for a variety of channels, such as for example the {\it block-asynchronous} G-IC and the {\it codebook oblivious} G-IC, thereby demonstrating that lack of codeword synchronism or of codebook knowledge at the receivers results in penalty of at most $O(1)$, or $\loglog$, compared to the classical G-IC.

\item In Section~\ref{sec:practical} we discuss som practical implications of our {TINnoTS} with mixed inputs achievability scheme, such as
\begin{itemize}
\item
in Section~\ref{sec: simple receiver} we discuss an approximate MAP decoder for the very strong interference regime that is very simple to implement with {TINnoTS},
\item 
in Section~\ref{sec: Actual vs. Analytic Gap} we show through numerical evaluations that our %constant and $O(\log\ln \min(\snr,\inr))$ 
gap results are very conservative and that in practice the achievable rates are much closer to capacity than predicted by our analytical results, 
\item 
in Section~\ref{sec: Mixed (Gaussian+Discrete) vs. Discrete (Discrete+Discrete) Inputs} we show that a gap result can be obtained 
%that a cost of an extra additive gap we could have used 
by using as inputs purely discrete random variables, i.e., to within an additive gap the Gaussian part of the mixed inputs can be replaced by another PAM input.
\end{itemize} 

\end{enumerate}

Section~\ref{sec:concl} concludes the paper. Some proofs can be found in the Appendix.

\subsection{Notation}
\label{sec:notation}
Throughout the paper we adopt the following notation convention: 
\begin{itemize}
\item
Lower case variables are instances of upper case random variables which take on values in calligraphic alphabets.

\item
$\log(\cdot)$ denotes logarithms in base 2 and $\ln(\cdot)$ in base $\eu$. 

\item
$[n_1:n_2]$ is the set of integers from $n_1$ to $n_2 \geq n_1$.

\item
$Y^{j}$ is a vector of length $j$ with components $(Y_1,\ldots,Y_j)$.

\item
If $A$ is a r.v. we denote its support by $\supp(A)$.  

\item
The symbol $|\cdot|$ may denote different things: 
$|\mathcal{A}|$ is the cardinality of the set $\mathcal{A}$, 
$|X|$ is the cardinality of $\supp(X)$ of the r.v. $X$ , or
$|x|$ is the absolute value of the real-valued $x$.

\item
For $x\in\mathbb{R}$ we let
$\left \lfloor x \right\rfloor$ denote the largest integer not greater than $x$.

\item
For $x\in\mathbb{R}$ we let $[x]^{+} :=\max(x,0)$ and $\log^{+}(x) :=[\log(x)]^+$.

\item
$d_{\min  \left( \mathcal{S} \right)} := \min_{i \neq j :  s_i,s_j \in \mathcal{S}} |s_i-s_j|$ denotes the minimum distance among the points in the set $\mathcal{S}$.
With some abuse of notation we also use $d_{\min  \left(X \right)}$ to denote $d_{\min(\supp(X))}$ for a r.v. $X$. 

\item
Let $f(x),g(x)$ be two real-valued functions.
We use the Landau notation %$O(\cdot)$ refers to $f(x),g(x) \in \mathbb{R}$ such that 
$f(x)=O(g(x))$ to mean that for {\it some}  $c>0$ there exists an $x_0$ such that $f(x)\leq c \, g(x)$ for all $x \geq x_0$.
%and $f(x)=o(g(x))$ to mean that for {\it every} $c>0$ there exists an $x_0$ such that $f(x)<   c g(x)$ for all $x \geq x_0$. 

\item
Operator $\co(\cdot)$ will refer to convex hull operation. 

\item
$X \sim \mathcal{N}(\mu,\sigma^2)$ denotes the density of a real-valued Gaussian r.v. $X$ with mean $\mu$ and variance $\sigma^2$.

\item
$X \sim \pam\left(N,d_{\min(X)} \right)$ denotes the uniform probability mass function over a zero-mean PAM constellation with $|\supp(X)|=N$ points, minimum distance $d_{\min(X)}$, and therefore average energy $\mathbb{E}[X^2] = d_{\min  \left(X \right)}^2\frac{N^2-1}{12}$.

\item 
$m( \mathcal{S} )$ denotes Lebesgue measure of the set $\mathcal{S}$.

\item

We let 
\begin{align}
\mug(x) &:=\frac{1}{2}\log(1+x),
\label{eq:def:mug}
\\
\mud(X) &:=
\left[H(X)
- \frac{1}{2}\log\left(\frac{2\pi\eu}{12}\right) 
- \frac{1}{2}\log\left(1+\frac{12}{d_{\min(X)}^2}\right)\right]^+,% 
\label{eq:def:mud}
\\
\points(x)&:=\left \lfloor \sqrt{1+x} \right\rfloor,
\label{eq:def:pti}
\end{align}
where the subscript $\mathsf{d}$ reminds the reader that discrete inputs are involved, while $\mathsf{g}$ that Gaussian inputs are involved. Here $H(X)$ is the entropy of the discrete random variable $X$, while $h(X)$ is the differential entropy of the absolutely continuous random variable $X$.

\end{itemize}

\section{Main Tools}
\label{sec:mainTools}

In this Section we present the main tools to evaluate the {TINnoTS} lower bound in~\eqref{eq:RL:TINnoTS} under mixed inputs.

\subsection{Generalized Ozarow-Wyner Bound} 
\label{sec:Ozarow}

At the core of our proofs is the following lower bound on the rate achieved by a discrete input on a point-to-point additive noise channel.  The important point here is to derive firm bounds that are valid for {\it any} discrete constellation at {\it any} SNR, as opposed to bounds that are either optimized for a fixed SNR, or hold asymptotically in the low or high SNR regimes.

\begin{prop}[Ozarow-Wyner-B bound]
\label{prop:lowbound OW generalized}
Let $X_D$ be a discrete random variable %, whose support has size $N$ and whose 
with minimum distance $d_{\min(X_D)} > 0$.
%, and average energy $\mathcal{E}_{X_D}=\mathbb{E}[X_D^2]$. 
Let $Z$ be a zero-mean unit-variance random variable independent of $X_D$ (not necessarily Gaussian). 
Then
\begin{subequations}
\begin{align}
   &\mud(X_D) :=\left[H(X_D) - {\gap}_{\eqref{eq:OW gen partB}}\right]^+ \leq I(X_D;  X_D+Z) \leq H(X_D), 
   \label{eq:OW gen I partB}
%\end{align}
%where %$H(X_D)$ is the entropy of $X_D$ and 
%\begin{align}
\\&{\gap}_{\eqref{eq:OW gen partB}} := \frac{1}{2}\log\left(\frac{2\pi\eu}{12}\right) + \frac{1}{2}\log\left(1+\frac{12}{d_{\min(X_D)}^2}\right). % DT REMOVED \frac{\mathcal{E}_{X_D}}{\mathcal{E}_{X_D}+1}
   \label{eq:OW gen gap partB}
\end{align}
For $|\supp(X_D)|=N$, the mutual information bounds 
%\nbl{ND: you mean gap rather than bounds?}
in~\eqref{eq:OW gen I partB} are the largest for a PAM constellation 
(since PAM satisfies with equality the general inequality $H(X_D)\leq\log(N)$).
%a constellation that maximizes $\mud(X_D)$, that is, 
\label{eq:OW gen partB}
\end{subequations}
\end{prop}
\begin{IEEEproof}
The upper bound in~\eqref{eq:OW gen I partB} is trivial.
For the lower bound, let $\widetilde{X}:=X_D + U$ with $U$ uniformly distributed on $[-d_{\min(X_D)}/2,+d_{\min(X_D)}/2]$ and independent of $X_D$ and $Z$, and let $Y:=X_D+Z$; from~\cite[eq(15)]{PAMozarow} we know that 
\begin{align}
I(X_D; Y) \geq I(\widetilde{X}; Y) = h(\widetilde{X}) - h(\widetilde{X}|Y).
\end{align}
By removing the assumption that $X_D$ is a PAM we write~\cite[eq(16)]{PAMozarow} as
\begin{align}
h(\widetilde{X}) = H(X_D) + \log(d_{\min(X_D)}).
\label{eq:h(tildeX)}
\end{align}
Next, the derivation of~\cite[eq(19)]{PAMozarow} holds under the assumptions of the proposition, i.e., no need to assume a PAM input or a Gaussian noise; thus, by using $s^2 = \mathbb{E}[(\widetilde{X} - k Y)^2]$ with $k = \frac{\mathbb{E}[\widetilde{X}Y]}{\mathbb{E}[Y^2]}$ we write~\cite[eq(19)]{PAMozarow} as
\begin{align}
h(\widetilde{X}|Y) \leq \frac{1}{2}\log\left[2\pi\eu\left(\frac{d_{\min(X_D)}^2}{12}+\frac{\mathbb{E}[X_D^2]}{\mathbb{E}[X_D^2]+1}\right)\right].
\label{eq:h(tildeX|Y)with E/(E+1)}
\end{align}
Combing this, 
by the non-negativity of mutual information, 
and since $\frac{\mathbb{E}[X_D^2]}{\mathbb{E}[X_D^2]+1} \leq 1$, 
the lower bound in~\eqref{eq:OW gen I partB}
with the gap expression in~\eqref{eq:OW gen gap partB} follows immediately.
%\cite[Part b): eq(16) and eq(14)]{PAMozarow}
\end{IEEEproof}

\begin{rem}
The proof of Proposition~\ref{prop:lowbound OW generalized}
holds for any continuous $U$ such that $\supp(U) \subseteq [-d_{\min(X_D)}/2,+d_{\min(X_D)}/2]$.
In this case
$\log(d_{\min(X_D)})$ must be replaced by $h(U)$ in~\eqref{eq:h(tildeX)}, and
$\frac{d_{\min(X_D)}^2}{12}$ must be replaced by the variance of $U$ in~\eqref{eq:h(tildeX|Y)with E/(E+1)}.
However, for this more general case, it may not be easy to analytically express the entropy as a function of the variance,
and to relate them to the bound on the size of the support of the distribution given by $d_{\min(X_D)}$.
\end{rem}

\begin{rem}
\label{rem:lowbound OW generalized AWGN}
If in Proposition~\ref{prop:lowbound OW generalized} we set
$Z=Z_G \sim \mathcal{N}(0,1)$, then we can tighten the upper bound in~\eqref{eq:OW gen I partB} to 
\begin{align}
   \mud(X_D) \leq 
   I(X_D;  X_D+Z_G) \leq \min\left(H(X_D),\mug(\mathbb{E}[X_D^2])\right),
   \label{eq:OW gen I AWGN}
\end{align}
since a Gaussian input is capacity achieving for the power-constrained point-to-point Gaussian noise channel. 
\end{rem}

We next compare the Ozarow-Wyner-B lower bound in Proposition~\ref{prop:lowbound OW generalized} to bounds available in the literature.

\paragraph*{Ozarow-Wyner-A, or Fano-based, bound}
Proposition~\ref{prop:lowbound OW generalized} generalizes the approach of~\cite[Part b)]{PAMozarow}.
Had we generalized~\cite[Part a)]{PAMozarow}, we would have obtained the following lower bound valid for Gaussian noise only

\begin{subequations}
\begin{align}
  &\left[H(X_D) - {\gap}_{\eqref{eq:OW gen partA}}\right]^+ \leq I(X_D;  X_D+Z_G),
   \label{eq:OW gen I partA}
\\&{\gap}_{\eqref{eq:OW gen partA}} := \left.\xi \log\frac{1}{\xi} + (1-\xi) \log\frac{1}{1-\xi}+\xi\log(N-1)\right.,
   \label{eq:OW gen gap partA}
\\&{\xi := 2 Q\left(\frac{ d_{\min(X_D)}}{2}\right)}, % \leq 2 \eu^{-d_{\min(X_D)}^2/8}
   \label{eq:OW gen Ps partA}
\end{align}
where $\xi$ is the union-of-events upper bound on the probability of symbol error for a minimum-distance symbol-by-symbol detector in Gaussian noise from Fano's inequality.  
\label{eq:OW gen partA}
\end{subequations}
We note that a similar Fano-based bounding technique was also used in~\cite[Theorem 3]{MoshksarJournal}.

In the following we are interested in showing that certain upper and lower bounds are to within a constant gap of one another, regardless of the channel parameters.
For bounds as in~\eqref{eq:OW gen partB}, the quantity ``${\gap}$'' upper bounds the difference between the upper and lower bounds. 
The gap in~\eqref{eq:OW gen partA} (that generalizes~\cite[Part a)]{PAMozarow} to any discrete input on the Gaussian noise channel) is bounded if the term $\xi\log(N-1)$ is bounded; by using the Chernoff's bound for the Q-function, i.e., $Q(x) \leq \frac{1}{2}\eu^{-x^2/2}$ and by imposing $\xi\log(N-1) \leq 1$ ,we get
\begin{align*}
 \text{bounded gap in~\eqref{eq:OW gen partA}}
 &\Longleftrightarrow
   \log(N-1) \leq  \eu^{d_{\min(X_D)}^2/8}
\\&\Longleftrightarrow
  d_{\min(X_D)}^2 \geq 8\ln(\log(N-1)),
\end{align*}
in other words, the minimum distance squared must be of the order of $\ln(\log(N))$ for the gap in~\eqref{eq:OW gen partA} to be bounded.

On the other hand, the gap in~\eqref{eq:OW gen partB} (that generalizes~\cite[Part b)]{PAMozarow} to any discrete input on  any additive noise channel) is bounded as long as the minimum distance is lower bounded by a constant; for example 
\begin{align*}
 \text{bounded gap in~\eqref{eq:OW gen partB}, say} \
 &{\gap}_{\eqref{eq:OW gen gap partB}} \leq \frac{1}{2}\log\left(6\pi\eu\right) \approx 2.047~\text{bits}
 \\&\Longleftrightarrow
 d_{\min(X_D)}\geq 2,
\end{align*}
that is, the minimum distance does not need to grow in a particular way with the number of points of the constellation, but it is required to be bounded by a constant from below.
%\end{rem}

%\begin{rem}
\paragraph*{DTD-ITA'14 bound}
In a conference version of this work~\cite{DytsoITA2014}, we derived the following lower bound for the mutual information with a discrete input on a Gaussian noise channel.
As before, let the noise $Z_G\sim \mathcal{N}(0,1)$ be independent of the discrete input $X_D$, and let $\Pr[X_D = s_j] = p_j > 0, \ j\in[1:N]$ such that $\sum_{j\in[1:N]}p_j=1$. We have -- the proof can be found in Appendix~\ref{app:ADlowerbound}:
\begin{subequations}
\begin{align}
&
\left[\log(N)-{\gap}_{\eqref{eq:AD gen}}\right]^+ 
\label{eq: ITA bound simplif}
\\& \leq \left[ -\log  \left( {\sum_{(i,j)\in[1:N]^2} p_i p_j \frac{1}{\sqrt{4\pi}} \eu^{-\frac{(s_i-s_j)^2}{4}}}\right)- \frac{1}{2}\log \left( 2\pi \eu \right) \right]^{+} \leq I(X_D;  X_D+Z_G),
\label{eq:ITA with Shaping}
\\
&
{\gap}_{\eqref{eq:AD gen}} := \frac{1}{2}\log\left(\frac{\eu}{2}\right)+\log\left(1+(N-1)\eu^{-d_{\min(X_D)}^2/4} \right).
\label{eq:AD gen gap}
\end{align}
The advantage of the bound in~\eqref{eq: ITA bound simplif} (referred to in the following as `simple DTD-ITA'14 bound') is its simplicity: it only depends on the constellation through the number of points and the minimum distance.
The bound in~\eqref{eq:ITA with Shaping} (referred to in the following as `full DTD-ITA'14 bound') is in general tighter than the one in~\eqref{eq: ITA bound simplif} but requires the knowledge of the whole ``distance spectrum'' (all pair-wise distances among constellation points) as well as the ``shaping'' of the constellation (the a priori probability of each constellation point), which does not make it amenable for closed form analytical computations in general.  
\label{eq:AD gen}
\end{subequations}

%The gap in~\eqref{eq:AD gen gap} is bounded if, 
Again aiming at a bounded gap, we have
\begin{align*}
 \text{bounded gap in~\eqref{eq:AD gen}}
  &\Longleftrightarrow
 (N-1) \eu^{-\frac{ d_{\min(X_D)}^2}{4}} \leq 1 
\\&\Longleftrightarrow d_{\min(X_D)}^2 \geq 4\ln(N-1),
%\label{eq:dmin as log(N) in AD}
\end{align*}
in other words, the minimum distance squared must be of the order of $\log(N)$ for the gap in~\eqref{eq:AD gen gap} to be bounded.
Because of this `strong' requirement on the the minimum distance, in~\cite{dytsoTINconf} we could show that a mixed input achieves the capacity region of the classical G-IC to within an additive gap of the order of $\loglog$, rather than a constant gap; but it was nonetheless sufficient to show that {TINnoTS} with mixed inputs achieves the sum gDoF of the classical G-IC for all channel gains up to a set of zero measure.

\paragraph*{Numerical Comparisons}
We conclude this subsection by numerically comparing the lower bounds in~\eqref{eq:OW gen partB},~\eqref{eq:OW gen partA} and~\eqref{eq:AD gen} for the Gaussian noise channel with a PAM input, which is asymptotically capacity achieving at high SNR~\cite{PAMozarow}.

In Fig.~\ref{fig: bounds on Id} we plot bounds on $I(X_D;  \sqrt{\snr}\ X_D+Z_G)$ vs. $\snr$ in dB;
here $\snr$ represents the SNR at the receiver,
$Z_G\sim \mathcal{N}(0,1)$ is the noise, and
$X_D \sim \pam\left( N, \sqrt{\frac{12}{N^2-1}} \right)$ is the input with 
$N = \points(\snr)= \left \lfloor \sqrt{1+\snr} \right\rfloor \approx \snr^{\frac{1}{2}}$.
In Fig.~\ref{fig:AchiebablePTP} we plot the rate bounds while in Fig.~\ref{fig:gapsToPTPcap} the gap to capacity, i.e., the difference between the channel capacity and the different lower bounds. 
In both figures we show:
\begin{enumerate}
\item 
The black curve is the channel capacity $\mug(\snr)$.
\item 
The blue curve is the Ozarow-Wyner-B bound in~\eqref{eq:OW gen I partB}. %with $\epsilon=0$. 
From Fig.~\ref{fig:gapsToPTPcap} this bound is asymptotically (for $\snr\geq 30$dB) to within   
$0.754$~bits of capacity, which is much better than the analytic worst case gap of $\frac{1}{2} \log(6 \pi \eu) = 2.8395$~bits shown before. 

\item 
The magenta curve is the Ozarow-Wyner-A bound in~\eqref{eq:OW gen I partA}. %with $\epsilon=0$. 
This bound is to within  $O(\log(\snr))$ of capacity (i.e., straight line as a function of $\snr|_\text{dB}$).  
\item 
The cyan curve is the simple DTD-ITA'14 bound in~\eqref{eq: ITA bound simplif}.
Here we used $N = \points(\snr^{1-\epsilon})\approx \snr^{\frac{1-\epsilon}{2}}$
%, where the parameter $\epsilon$ is needed to control minimum distance / get a bounded gap.
with $\epsilon=\max\left(0,\frac{\log(\frac{1}{6}\ln(\snr))}{\log(\snr)} \right)$. 
This choice of $\epsilon$ was derived in~\cite[Theorem 3]{DytsoITA2014} in order to have a $O(\log \log (\snr))$ gap to capacity. 
Had we chosen $\epsilon=0$ then we could only achieve a `gap' of $O(\log(\snr))$. 
Similarly, for the Ozarow-Wyner-A, had we choose the same $\epsilon=\max\left(0,\frac{\log(\frac{1}{6}\ln(\snr))}{\log(\snr)} \right)$  a similar $O(\log \log (\snr))$ gap would have been observed. 
\item  
The green curve is the full DTD-ITA'14 bound in~\eqref{eq:ITA with Shaping}, % with $\epsilon=0$, 
which from Fig.~\ref{fig:gapsToPTPcap} achieves asymptotically (for $\snr\geq 30$dB) to within   
$0.36$~bits of capacity. 
\end{enumerate}

The quantity $\frac{1}{2}\log\left(\frac{\pi\eu}{6}\right)$  is also shown  for reference in Fig.~\ref{fig:gapsToPTPcap}; 
this is the ``shaping loss'' for a one-dimensional infinite lattice 
and is the limiting gap if the number of points $N$ grows faster %(in the exponent) 
than $\snr^{1/2}$.
The ``zig-zag'' behavior of the curves at low SNR is due to the floor operation in $N = \left \lfloor \sqrt{1+\snr} \right\rfloor$.

We observe that the relative ranking among the bounds at low SNR (roughly less than 27~dB) is different than at high SNR. 
In particular we observe a qualitatively different behavior at high SNR: 
the Ozarow-Wyner-B bound in~\eqref{eq:OW gen I partB} (blue curve) and 
the full DTD-ITA'14 bound in~\eqref{eq:ITA with Shaping} (green curve) result in a constant gap, while 
the Ozarow-Wyner-A bound in~\eqref{eq:OW gen I partA} (magenta curve) and 
the simple DTD-ITA'14 bound in~\eqref{eq: ITA bound simplif} (cyan curve) result in a gap that grows with SNR;
this is in agreement with the previous discussion that points out that for a constant gap in the latter two cases the number of points $N$ must grow slower %%(in the exponent)
than $\snr^{1/2}$.
The smallest gap at high SNR for $N \approxeq \snr^{1/2}$ is given by our full DTD-ITA'14 bound in~\eqref{eq:ITA with Shaping} (green curve);
as pointed out earlier, this bound is unfortunately not amenable for closed form analytical evaluations,  so in the following we shall use the Ozarow-Wyner-B bound in~\eqref{eq:OW gen I partB} (blue curve) from Proposition~\ref{prop:lowbound OW generalized} whose simplicity comes at the cost of a larger gap.

\begin{figure}
        \centering
        \begin{subfigure}[a]{0.5\textwidth}
                \includegraphics[width=8.5cm]{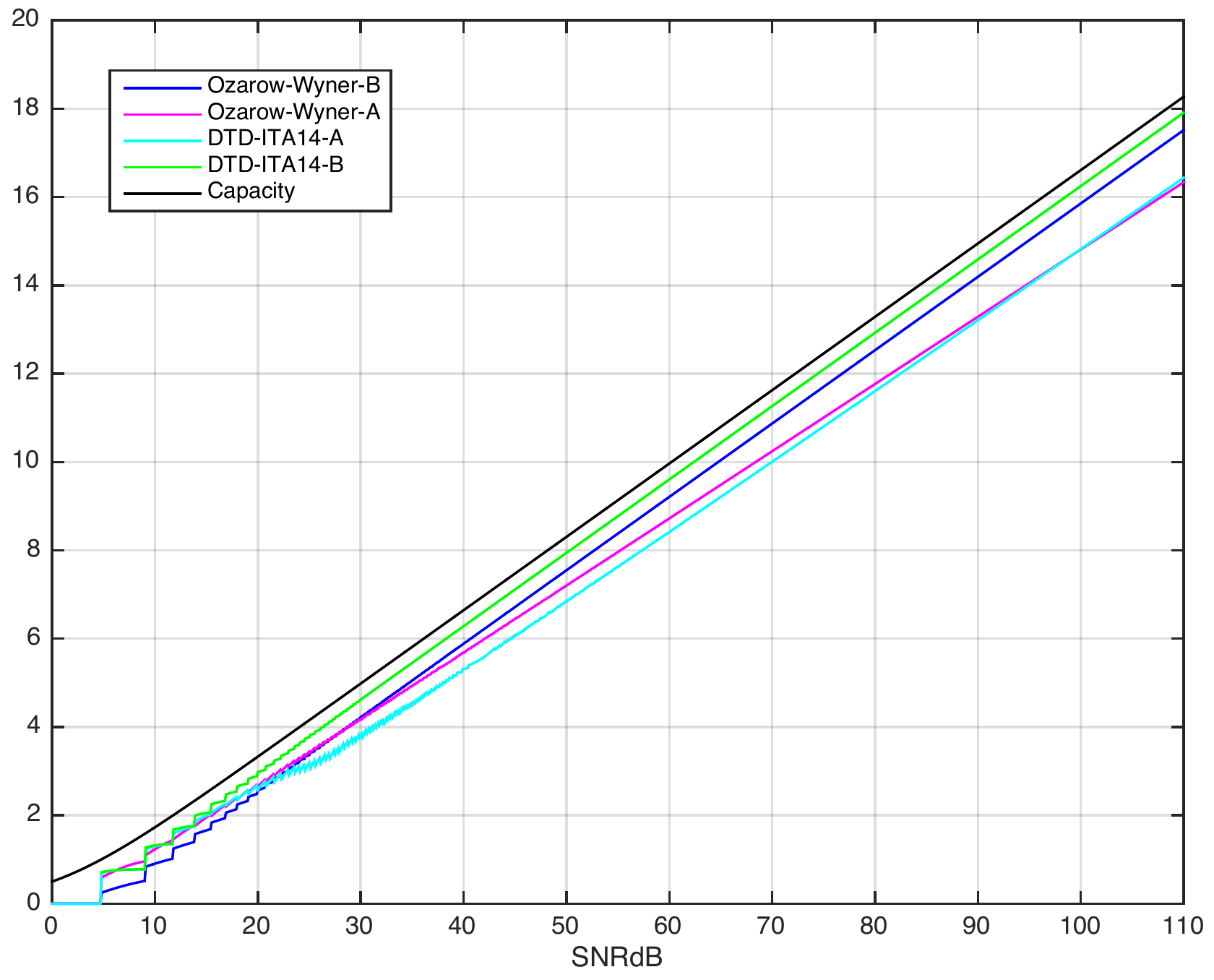}
                \caption{Bounds on $I(X_D; \sqrt{\snr} X_D +Z_G)$ vs $\snr$.}
                \label{fig:AchiebablePTP}
        \end{subfigure}%
        ~ %add desired spacing between images, e. g. ~, \quad, \qquad, \hfill etc.
        \begin{subfigure}[a]{0.5\textwidth}
                \includegraphics[width=8.5cm]{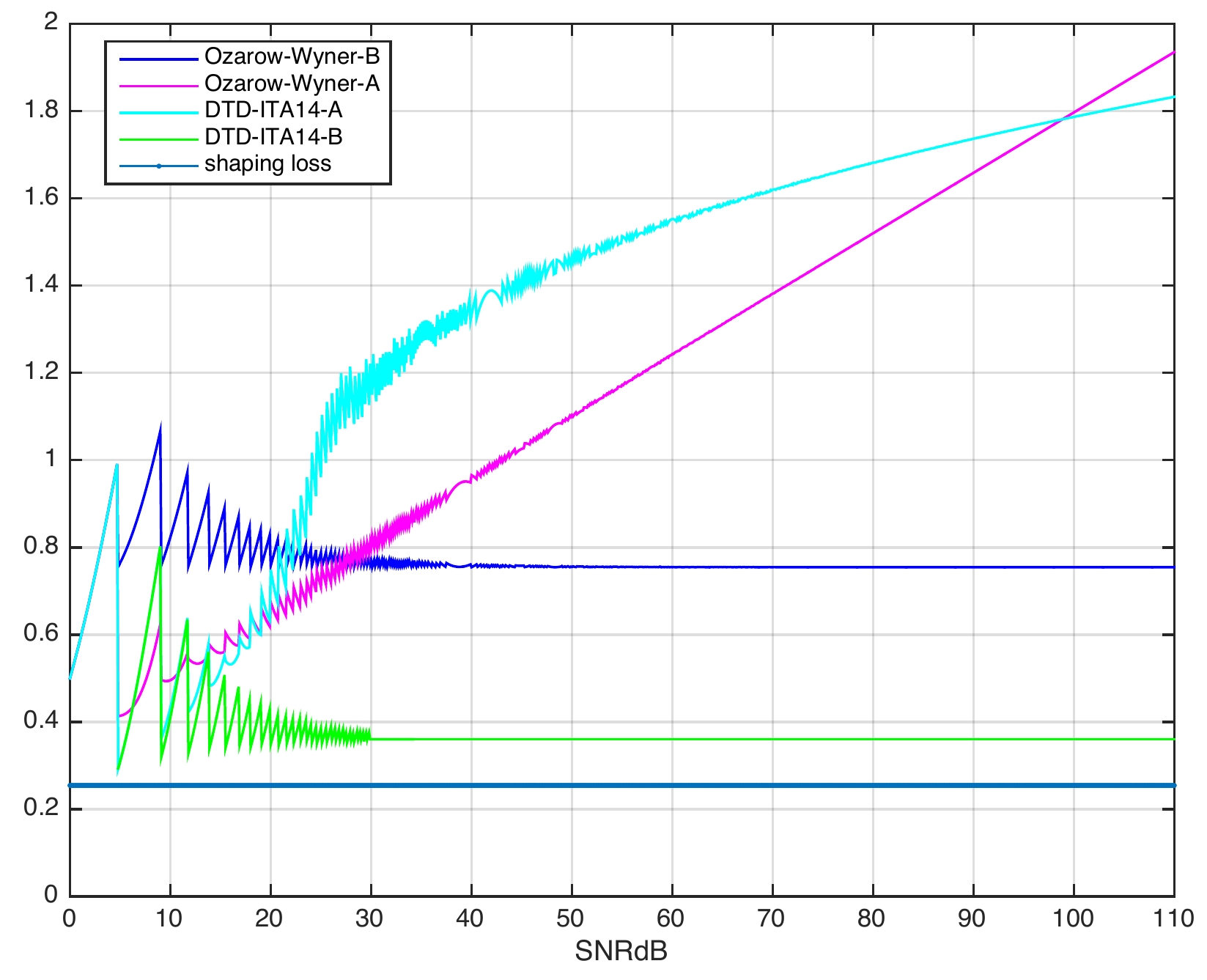}
                \caption{Gap to capacity vs $\snr$.}
                \label{fig:gapsToPTPcap}
        \end{subfigure}
        ~ %add desired spacing between images, e. g. ~, \quad, \qquad, \hfill etc.
        \caption{Comparison of different bounds for a PAM input on a Gaussian noise channel.}
        \label{fig: bounds on Id}
\end{figure}

%\end{rem}

\subsection{Cardinality and Minimum Distance Bounds for Sum-Sets}
\label{sec:carAndMinD}

In multi-user settings, we may wish to select one user's input as Gaussian, another as discrete, %both discrete, both Gaussian,
or both mixtures of discrete and Gaussian. To handle such scenarios, we need bounds on the cardinality and minimum distance of sums of discrete constellations. If $X$ and $Y$ are two sets, we denote the {\it sum-set} as 
\[
X+Y : = \{x+y| x\in {X}, y\in {Y}\}. 
\]
Tight bounds on the cardinality and the minimum distance of $X+Y$, for general $X$ and $Y$, are an open problem in the area of additive combinatorics and number theory~\cite{addCombTao}.

The following set of sufficient conditions for the sum-set obtained with two PAM constellations 
(actually the probability with which each point is used does not matter as long as it is strictly positive)
will play an important role in evaluating our inner bound.

\begin{prop} 
\label{prop:combNOOUTAGE}
Let $(h_x,h_y)\in \mathbb{R}^2$ be two constants.
Let $X \sim \pam(|X|,d_{\min(X)})$ and $Y \sim \pam(|Y|,d_{\min(Y)})$.
Then 
\begin{align}
|h_xX+h_yY|&=|X||Y|,
\\
d_{\min\left(h_xX+h_yY\right)}&=\min \left(|h_x|d_{\min(X)},|h_y|d_{\min(Y)} \right),
\end{align}
under the following conditions
\begin{subequations}
\begin{align}
\text{either}    \ \ {|Y||h_y|d_{\min(Y)}} \leq |h_x|d_{\min(X)},
\label{eq:thm:comb condition1}\\
\text{or} \ \ {|X||h_x|d_{\min(X)}} \leq |h_y|d_{\min(Y)}.
\label{eq:thm:comb condition2}
\end{align}
\label{eq:thm:comb conditions}
\end{subequations}
\end{prop}
\begin{IEEEproof}
The condition in~\eqref{eq:thm:comb conditions} is such that one PAM constellation is completely contained within two points of the other PAM constellation, see Fig.~\ref{fig:noOverLap} for a visual illustration.
\end{IEEEproof}

\begin{figure}
\center
\includegraphics[width=9cm]{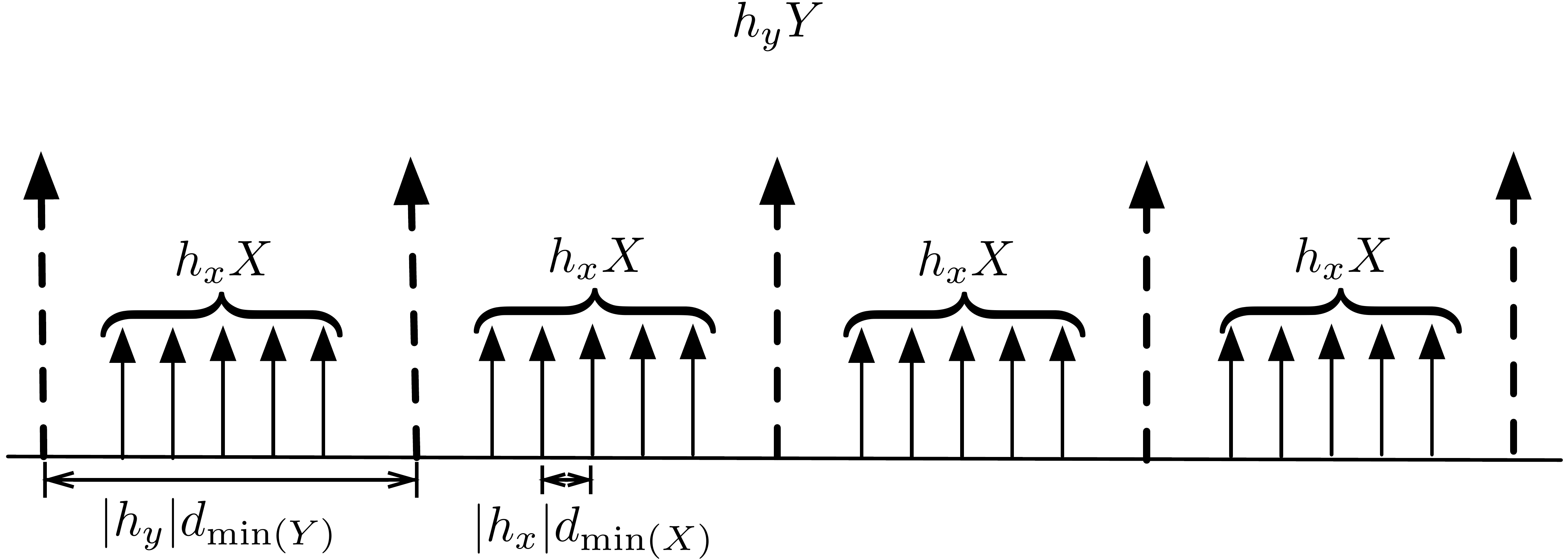}
\caption{Structure of the sum-set under the conditions in Proposition~\ref{prop:combNOOUTAGE}.}
\label{fig:noOverLap}
\end{figure}

\begin{figure}
\centering
\includegraphics[width=8cm]{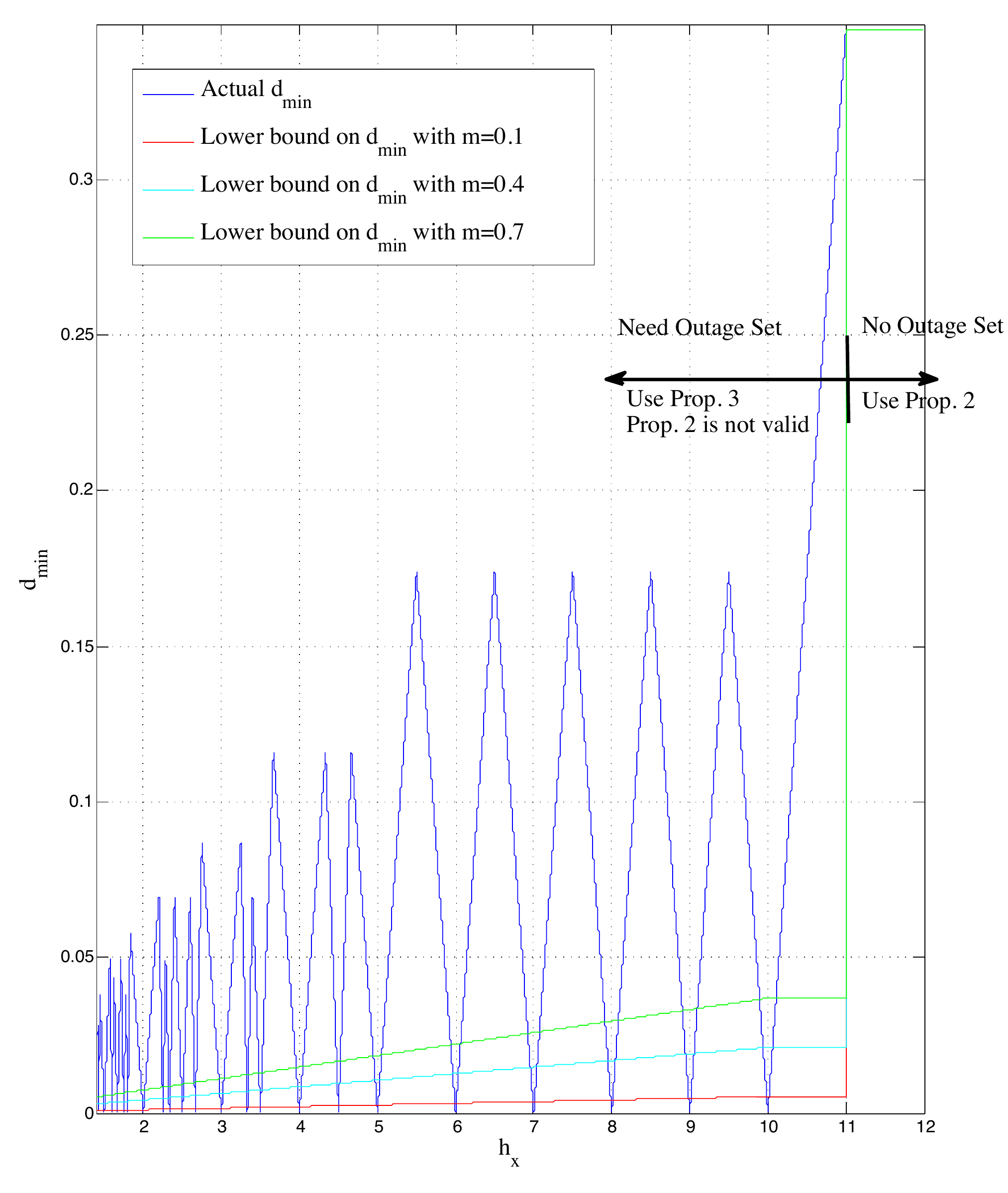}  %width=9cm,height=9cm
\caption{Minimum distance (blue line) for the sum-set $h_xX+h_yY$ as a function of $h_x$ for fixed $h_y=1$ and for $X \sim Y \sim \pam \left(10,1\right)$.
 On the right of the vertical green line Proposition~\ref{prop:combNOOUTAGE} is valid.
On the left of the vertical green line Proposition~\ref{prop:card:dmin:measurebound} must be used;
in this case, the minimum distance lower bound in~\eqref{eq:defdiminblabliblabla} holds for set of $h_x$'s for which the blue line is above the red / cyan / green line, where the red, cyan and green lines represent a different value for the measure of the outage set. }
\label{fig:minimumDistanceBehaviour}
\end{figure}

We will refer to the condition in~\eqref{eq:thm:comb conditions} as the {\it non-overlap} condition. Unfortunately, Proposition~\ref{prop:combNOOUTAGE} is not sufficient for our purposes because it restricts the set of channel parameters for which we can compute the minimum distance to those cases where the non-overlap condition holds. When the non-overlap condition in~\eqref{eq:thm:comb conditions} is not satisfied, the minimum distance is very sensitive to the fractional values of $h_x$ and $h_y$.
Fig.~\ref{fig:minimumDistanceBehaviour} shows, in solid blue line, the minimum distance for the sum-set $h_xX+h_yY$ 
%each constellation \dt{is a PAM} with $|X|=|Y|=10$ points \dt{and minimum distance equal to one}, 
as a function of $h_x$ for fixed $h_y=1$ and where $X$ and $Y$ are the same $\pam(10,1)$ constellation.
It can be observed that there are channel gains for which the minimum distance is zero;
those occur on the left of the vertical green line, which separates the values of $h_x$ for which Proposition~\ref{prop:combNOOUTAGE} is valid (right side) for those where it is not (left side).
To bound the cardinality and the minimum distance when the condition in~\eqref{eq:thm:comb conditions} is not satisfied  we have the following  lower bound.
\begin{prop}
\label{prop:card:dmin:measurebound}
Let $X \sim \pam(|X|,d_{\min(X)})$ and $Y \sim \pam(|Y|,d_{\min(Y)})$.
Then  for  $(h_x,h_y) \in\mathbb{R}^2$ 
\begin{align}
&|h_xX+h_yX| = |X| |Y| \ \text{almost everywhere (a.e.)},
\label{eq:dedminoutage}
\end{align}
and for any $\gamma>0$ there exists a set $E \subseteq \mathbb{R}^2$ such that for all $(h_x,h_y)\in E$
\begin{subequations}
\begin{align}
d_{\min \left( h_xX+h_yY\right)} &\geq \kappa_{\gamma,|X|,|Y|} \cdot
\min\left(|h_x|d_{\min(X)}, |h_y|d_{\min(Y)},\Upsilon_{|h_x|,|h_y|,|X|,|Y|}\right),
\label{eq:defdiminblabliblabla}
\\
\kappa_{\gamma,|X|,|Y|} &:= \frac{\gamma/2}{1+\ln(\max(|X|,|Y|))},
\label{eq:defkappa}
\\
\Upsilon_{|h_x|,|h_y|,|X|,|Y|} &:= \max \left(\frac{|h_x|d_{{ \min}(X)}}{  |Y|},\frac{|h_y|d_{{ \min}(Y)}}{  |X|} \right),
\label{eq:defUpsilon}
\end{align}
where the Lebesgue measure of the complement of the set $E$ (referred to as the {\it outage set}) satisfies $m(E^c)\leq \gamma$.
\end{subequations}
\end{prop}
\begin{proof}
The proof can be found in Appendix~\ref{app:Proof of Minimum Distance measure bound}.\footnote{%
%\begin{rem}
In our conference paper~\cite{dytsoTINconf}, the minimum distance bound in~\cite[eq.(8)]{dytsoTINconf} 
was missing the term $\Upsilon_{|h_x|,|h_y|,|X|,|Y|}$ in~\eqref{eq:defUpsilon}.
However, this did not impact the claimed gDoF results. %of~\cite{dytsoTINconf}. 
%\end{rem}
}
\end{proof}

The reason we need to introduce an outage set in Proposition~\ref{prop:card:dmin:measurebound} is that there are  values of $(h_x,h_y)$ for which the minimum distance is zero, as it can be seen from Fig.~\ref{fig:minimumDistanceBehaviour}.
In computing the gap later on, we want to exclude the set of channel gains for which the minimum distance is too close to zero; 
the measure of this set can be controlled through the parameter $\gamma$. 
The green, cyan, and red lines in Fig.~\ref{fig:minimumDistanceBehaviour} represent lower bounds on the minimum distance that are valid everywhere except for a set of measure no greater than  $\gamma=0.1, 0.3$ and $0.7$, respectively.
It is important to notice that the set of channel gains for which the minimum distance is exactly zero satisfies: 

\begin{prop}
\label{cor: mini dist 0}
 Under the same assumptions of Proposition~\ref{prop:card:dmin:measurebound}, the set of $(h_x,h_y) \in \mathbb{R}^2$ such that $d_{\min \left( h_xX+h_yY\right)}=0$ has Lebesque measure zero for any pair of countable sets $X$ and $Y$.
\end{prop}
\begin{proof}
The proof follows by observing that the set of channel gains for which $d_{\min \left( h_xX+h_yY\right)}=0$ and $ \left| h_xX+h_yY\right| \neq |X||Y|$  are equivalent and given by eq.\eqref{eq: Bad set} in Appendix~\ref{app:Proof of Minimum Distance measure bound}. The rest of the proof is similar to that of Proposition~\ref{prop:card:dmin:measurebound}.
\end{proof}

\begin{rem}
%We would like to point out that 
Different minimum distance bounds for sum-sets based on Diophantine approximations were used in~\cite{MatahariRealInteferenceAlignment}.
For example, consider the sum-set $h_1X+h_2X$, i.e., both transmitters use the same PAM constellation $X$, where $h_1^2=h_S^2\snr$ and  $h_1^2=h_I^2\snr^\alpha$ for some fixed $(h_S,h_I)\in\mathbb{R}^2$ and $\alpha>0$.
The authors of~\cite{MatahariRealInteferenceAlignment} focused on the degrees of freedom  (DoF) for the case when $\alpha=1$;
in this case the minimum distance can be lower bounded as follows 
\begin{subequations}
\begin{align}
d_{\min(h_1X+h_2X)}&=\min_{ x_{1i} , x_{2i}\in X } |h_1x_{1i}-h_2x_{2i}|
\notag\\
&= \min_{ z_{1i}, z_{2i} \in [-\frac{N}{2}:\frac{N}{2}], } | h_S\sqrt{\snr} d_{\min(X)} \ z_{1i}-h_I\sqrt{\snr}d_{\min(X)} \ z_{2i}| 
\notag\\
&= \sqrt{\snr} \ d_{\min(X)} \min_{ z_{1i}, z_{2i} \in [-\frac{N}{2}:\frac{N}{2}], }  | {h_S} z_{1i}-{h_I}z_{2i}|
\label{eq:a bound based in Diophantine approximations: factorization} 
\\
&\geq \kappa_\epsilon \frac{2^{\epsilon}}{N^{\epsilon}} \ \sqrt{\snr} \ d_{\min(X)},
\label{eq:a bound based in Diophantine approximations} 
\end{align}
where the inequality in~\eqref{eq:a bound based in Diophantine approximations} comes from Diophantine approximation results, specifically from the Khintchine-Groshev theorem, and says that for  almost all real numbers $({h_S},{h_I})$ and for any $\epsilon>0$ there exists a constant $\kappa_\epsilon>0$, whose analytical expression is not known, such that the bound in \eqref{eq:a bound based in Diophantine approximations} holds.
\end{subequations}

Unfortunately, bounds such as~\eqref{eq:a bound based in Diophantine approximations} are only well suited for the derivation of DoF (i.e $\alpha=1$ but not for gDoF (i.e. $\alpha \neq 1$), which is of interest here. The fundamental problem is that  for $\alpha \neq 1$, the factorization in \eqref{eq:a bound based in Diophantine approximations: factorization} is no longer possible and $\kappa_\epsilon$ may end up being a function of $\snr$ and $\alpha$.
Moreover, the fact that we have auxiliary constants $\epsilon$ and $\kappa_\epsilon$ %in the minimum distance lower bound 
in~\eqref{eq:a bound based in Diophantine approximations}, and where $\kappa_\epsilon$ is essentially not known in closed form,
makes derivation of closed form gap results very difficult. 
\end{rem}

\subsection{Examples} 
\label{sec:ptpExample}

In this Section we %show the usefulness of
give an example of how we intend to use discrete inputs in the {TINnoTS} region in~\eqref{eq:RL:TINnoTS} for the G-IC by considering the familiar point-to-point power-constrained additive white Gaussian noise channel. The goal is to derive some properties / results for a simple setting that we shall use often in the subsequent sections.
Specifically, we aim to show that
the unit-energy discrete input $X_D$ with a properly chosen number of points $N=|\supp(X_D)|$ as a function of $\snr$ achieves, roughly speaking ($\approx$)
\begin{align}
I(X_D;\sqrt{\snr} X_D+Z_G)     & \approx \log(N), \ \ Z_G\sim \mathcal{N}(0,1), \label{eq:discrete:input}
\\
I(X_G; \sqrt{\snr} X_G+X_D+Z_G) & \approx \mug(\snr), \ \ X_G\sim \mathcal{N}(0,1), \label{eq:discrete:interf}
\end{align}
that is, the discrete input $X_D$ is a ``good'' input and a ``good'' interference. To put it more clearly, when we use a discrete constellation  as input, as in~\eqref{eq:discrete:input}, the mutual information is roughly equal to the entropy of the constellation, which is highly desirable. On the other hand, if the interference, unknown to transmitter and receiver, is from a discrete constellation  as in~\eqref{eq:discrete:interf}, the mutual information is roughly as if there was no interference, which is again  highly desirable.  In contrast, a Gaussian input instead of $X_D$ would be the ``best'' input for~\eqref{eq:discrete:input} but the ``worst'' interference/noise in~\eqref{eq:discrete:interf}. We next formalize the approximate statements in~\eqref{eq:discrete:input} and~\eqref{eq:discrete:interf}.

\paragraph*{Gaussian Channel}
Consider the point-to-point power-constrained Gaussian noise channel
\begin{subequations}
\begin{align}
  &Y=\sqrt{\snr}\ X+Z_G,
\\&\mathbb{E}[X^2]\leq 1, \ Z_G\sim \mathcal{N}(0,1),
\end{align}
\label{ptp:InputOutput}
\end{subequations}
where $X$ is the information carrying signal, independent of the noise $Z_G$.
The capacity of this channel, as a function of the SNR $\snr$, is $C\left(\snr\right)=\mug\left(\snr\right)$ and is achieved by $X  \sim \mathcal{N}(0,1)$ for every $\snr$. 
Consider now the input $X=X_D \sim \pam\left( N, \sqrt{\frac{12}{N^2-1}} \right)$ on the channel in~\eqref{ptp:InputOutput}.
By Proposition~\ref{prop:lowbound OW generalized} and Remark~\ref{rem:lowbound OW generalized AWGN} 
\begin{align}
  %\mud(X_D)= 
  \left[\log(N) - \frac{1}{2}\log\left(\frac{2\pi\eu}{12}\right) 
  - \frac{1}{2}\log\left( 1+\frac{N^2-1}{\snr} \right)\right]^+ %\frac{\snr}{\snr+1}
  \leq I(X_D;  \sqrt{\snr} X_D+Z_G) \leq \mug(\snr).
\label{eq:1stexamplebounds}
\end{align}
By observing the bounds in~\eqref{eq:1stexamplebounds}, we see that for a PAM input to be optimal to within a constant gap we need that $\log(N) \approx \mug(\snr)$ and that $\frac{N^2-1}{\snr}$ is upper bounded by a constant. By choosing $N = \left \lfloor \sqrt{1+\snr} \right\rfloor =: \points(\snr)$ it is easy to see that a PAM input can achieve the capacity $\mug(\snr)$ to within 
$\frac{1}{2}\log\left(\frac{2\pi\eu}{3}\right)\approx 1.25$~bits,
where the maximum gap is for $\snr = 3 -\epsilon$
for some $0<\epsilon \ll 1$. %, in which case $N = \points(\snr)=???$.

Note that, had we kept the term $\frac{\mathbb{E}[X_D^2]}{\mathbb{E}[X_D^2]+1}$ in~\eqref{eq:h(tildeX|Y)with E/(E+1)},
the bound in~\eqref{eq:1stexamplebounds} would have had $\frac{N^2-1}{\snr+1}$ in place of $\frac{N^2-1}{\snr}$
and would have resulted in a gap of at most $\frac{1}{2}\log\left(\frac{\pi\eu}{2}\right)\approx 1.047$~bits.
As always, bounds which allow for expressions that are easier to manipulate analytically come at the expense of a larger gap.

\paragraph*{Gaussian Channel with States}
The above example showed that a discrete input with $\log(N) \approx \mug(\snr)$ %and bounded $\frac{N^2-1}{\snr+1}$
is a ``good'' input in the sense alluded to by~\eqref{eq:discrete:input}. We now show that a discrete interference is a ``good'' interference in the sense alluded to by~\eqref{eq:discrete:interf}. We study an extension of the channel in~\eqref{ptp:InputOutput} by considering an {\em additive state $T$ available neither at the encoder nor at the decoder}. The input-output relationship is
\begin{subequations}
\begin{align}
& Y=\sqrt{\snr} \ X+h \ T+Z_G :
\\ &\mathbb{E}[X^2]\leq 1, \ Z_G\sim \mathcal{N}( 0,1),
\\ &T \ \text{discrete with finite power}. %\ |T| = N  \text{ and } d_{\min(T)}^2>0.
\end{align}
\label{ptp:InputOutput state}
\end{subequations}
%where state $T \sim P_T(t)$ and the noise $Z \sim \mathcal{N}(z;0,1)$ and input has the second order constraint $E[X^2] \leq 1$. In this context, one can think of $T$ as interference or noise. Specifically, we will be interested in $T$ such that
It is well known~\cite[Section 7.4]{elgamalkimbook} that the capacity of the channel with random state in~\eqref{ptp:InputOutput state} is 
\begin{align}
C=\max_{P_X} I(X;Y) \leq \max_{P_X} I(X;Y|T)= \mug(\snr).
\end{align}
From~\cite{ZamirGaussianIsNotToBad} we know that $X=X_G \sim  \mathcal{N}(0,1)$ is at most 1/2~bit from the capacity  $C$, but the value of the capacity is unknown. In particular it is not know whether the gap to the interference free capacity $\mug(\snr)-C$ is a bounded function of $\snr$.

Assume we use the input $X=X_G \sim  \mathcal{N}(0,1)$, as a Gaussian input is not too bad for an additive noise channel~\cite{ZamirGaussianIsNotToBad}; 
assume also that $d_{\min(T)}>0$; then the achievable rate $R$ satisfies
\begin{subequations}
\begin{align}
R &\geq \mug(\snr) - {\gap}_{\eqref{ach:ptp:state}},
\label{ach:ptp:state last}
\\  
{\gap}_{\eqref{ach:ptp:state}} &:= 
 \frac{1}{2}\log\left(\frac{2\pi\eu}{12}\right)
+\frac{1}{2}\log\left(1+\frac{12}{d_{\min(T)}^2}
%\frac{ |h|^2 \mathcal{E}_{T}}{|h|^2 \mathcal{E}_{T}+1+\snr}
\right),
\label{ach:ptp:state last gap}
\end{align}
\label{ach:ptp:state}
\end{subequations}
since
\begin{align*}
  & I(X_G; \sqrt{\snr} \ X_G+h \ T+Z_G )
\\&= \underbrace{h(\sqrt{\snr} \ X_G+h \ T+Z_G ) - h(\sqrt{\snr} \ X_G+Z_G)}_{\geq \mud\left(\frac{h}{\sqrt{1+\snr}}T\right) \geq H(T)-{\gap}_{\eqref{ach:ptp:state}}}
\\&- \underbrace{\left(h(h \ T+Z_G )-h(Z_G)\right)}_{\leq H(T)}
\\&+ \underbrace{\left(h(\sqrt{\snr} \ X_G+Z_G)+h(Z_G)\right)}_{= \mug(\snr)}.
\end{align*}
Thus, as long as %$\frac{12}{d_{\min(T)}^2}\frac{ |h|^2 \mathcal{E}_{T}}{|h|^2 \mathcal{E}_{T}+1+\snr}$ 
$d_{\min(T)}$ is lower-bounded by a constant, it is possible to achieve the interference-free capacity to within the constant gap in~\eqref{ach:ptp:state last gap} even when the state is unknown to both the transmitter and the receiver.

The rate expression in~\eqref{ach:ptp:state} can be readily used to lower bound the achievable rate in a G-IC where one user has a Gaussian input and the other a discrete input and where the discrete input is treated as noise,  as we shall do in the next sections.

\section{{TINnoTS} with Mixed Inputs Achievable Rate Region and an Outer Bound for the G-IC} 
\label{sec:ach}

For the G-IC in~\eqref{eq:block awgn ic} we now evaluate the {TINnoTS} region in~\eqref{eq:RL:TINnoTS} with inputs
\begin{subequations}
\begin{align}
   X_i &= \sqrt{1-\delta_i}\ X_{iD}
        + \sqrt{\delta_i}  \ X_{iG}, \ i \in[1:2]: %\quad \delta_i \in[0,1], 
\\ &\quad X_{iD}\sim \pam\left(N_i,\sqrt{\frac{12}{N_i^2-1}}\right), %\ i \in[1:2],
\\ &\quad X_{iG}\sim \mathcal{N}(0,1), %\ i \in[1:2],
\\ \mathbf{p}&:=[N_1,N_2,\delta_1,\delta_2]\in \mathbb{N}\times\mathbb{N}\times[0,1]\times[0,1],
\label{eq:GICmixedinput defB}
\end{align}
where the random variables $X_{ij}$ are independent for $i\in[1:2]$ and $j\in\{D,G\}$.
The input in~\eqref{eq:GICmixedinput} has four parameters, collected in the vector $\mathbf{p}$, namely: 
the number of points $N_i\in\mathbb{N}$ and the power split $\delta_i\in[0,1]$, for  $i\in[1:2]$, which must be chosen carefully in order to match a given outer bound.
%A careful choice of $\mathbf{p}$ will lead to the desired results in different parameter regimes. 
\label{eq:GICmixedinput}
\end{subequations}

\begin{prop}%[Achievable region $\mathcal{R}_{\text{in}}(\snr,\inr)$]
\label{prop:ach-with-mixedinput}
For the G-IC 
%symmetric G-IC in~\eqref{eq:block awgn ic}
the {TINnoTS} region in~\eqref{eq:RL:TINnoTS} contains the region $\mathcal{R}_{\text{in}}$ %(\snr,\inr)
defined as
\begin{align}
\mathcal{R}_{\text{in}} := \bigcup %%% -- GENERAL ASYMMETRIC
\left \{ 
\begin{array}{l}
0\leq R_1 \leq \mud \left(S_1\right)
+\mug \left(\frac{|h_{11}|^2\delta_1}{1+|h_{12}|^2\delta_2} \right) 
-\min \left(\log(N_2),\mug\left(\frac{|h_{12}|^2 (1-\delta_2)}{1+ |h_{12}|^2\delta_2} \right) \right)\\
0\leq R_2 \leq \mud \left(S_2\right)
+\mug \left(\frac{|h_{22}|^2\delta_2}{1+|h_{21}|^2\delta_1} \right)
-\min \left(\log(N_1),\mug \left(\frac{|h_{21}|^2(1-\delta_1)}{1+ |h_{21}|^2\delta_1} \right) \right)
\end{array} 
\right\},  
\label{eq: rates mixed inputs WITH Union}
\end{align}
where the union is over all possible parameters $[N_1,N_2,\delta_1,\delta_2]\in \mathbb{N}^2\times[0,1]^2$ for the mixed inputs in~\eqref{eq:GICmixedinput}
and where the equivalent discrete constellations seen at the receivers are
\begin{subequations}
\begin{align}
%%% -- GENERAL ASYMMETRIC
&{S}_1:= \frac{1}{\sqrt{1+|h_{11}|^2\delta_1+|h_{12}|^2\delta_2}}(\sqrt{1-\delta_1}h_{11}X_{1D}+\sqrt{1-\delta_2}h_{12}X_{2D}), \label{eq:sumSetAtRx1}
\\
&{S}_2:= \frac{1}{\sqrt{1+|h_{21}|^2\delta_1+|h_{22}|^2\delta_2}}(\sqrt{1-\delta_1}h_{21}X_{1D}+\sqrt{1-\delta_2}h_{22}X_{2D}). \label{eq:sumSetAtRx2}
\end{align}
\label{eq:GICmixedinputinnerBound}
\end{subequations}
\end{prop}
\begin{IEEEproof}
Due to the symmetry of the problem we derive a lower bound on $I(X_2;Y_2)$ only by following steps similar to those in~\eqref{ach:ptp:state}; a lower bound on $I(X_1;Y_1)$ follows by swapping the role of the users. 
Let $Z_G\sim \mathcal{N}(0,1)$. An achievable $R_2$ must satisfy $R_2 \leq I(X_2;Y_2)$ where
\begin{align*}
I(X_2;Y_2)&=I(X_2; h_{21}X_1+h_{22}X_2+Z_G)
\\&=\underbrace{\left[ h\left(\frac{\sqrt{1-\delta_1}h_{21}X_{1D}+\sqrt{1-\delta_2}h_{22}X_{2D}}{\sqrt{1+|h_{21}|^2\delta_1+|h_{22}|^2\delta_2}}+Z_G\right)-h(Z_G)\right]}_{\geq \mud\left( {S}_2 \right) \ \text{by Proposition~\ref{prop:lowbound OW generalized}}}
\\&-\underbrace{\left[h\left(\frac{\sqrt{1-\delta_1}}{\sqrt{1+|h_{21}|^2\delta_1}}h_{21}X_{1D}+Z_G\right) -h(Z_G) \right]}_{
\leq \min\left(\log(N_1),\frac{1}{2}\log\left(1+\frac{|h_{21}|^2(1-\delta_1)}{1+|h_{21}|^2\delta_1}\right)\right)  \ \text{by Remark~\ref{rem:lowbound OW generalized AWGN}} }
\\&+\frac{1}{2}\log\left(1+|h_{21}|^2\delta_1+|h_{22}|^2\delta_2\right)-\frac{1}{2}\log(1+|h_{21}|^2\delta_1).
\end{align*}
By considering  
the union over all possible choices of parameters for the mixed inputs
we obtain the achievable region in~\eqref{eq: rates mixed inputs WITH Union}.
%This concludes the proof.
\end{IEEEproof}

In the following sections we shall show that our {TINnoTS} region with mixed inputs
in Proposition~\ref{prop:ach-with-mixedinput} 
is to within an additive gap of the outer bound region given by:
\begin{prop}
\label{prop:ETWouter}
The capacity region of the G-IC is contained in
\begin{subequations}
\begin{align}
\mathcal{R}_{\text{out}} =\Big\{ \quad 
%%% -- GENERAL ASYMMETRIC
R_1 &\leq \mug \left(|h_{11}|^2\right),
  \ \text{cut-set bound},
\label{eq:R upper classical IC cuset r1}
\\
R_2 &\leq \mug \left(|h_{22}|^2\right),
  \ \text{cut-set bound},
\label{eq:R upper classical IC cuset r2}
\\
R_1+R_2 & \leq  \Big[\mug\left(|h_{11}|^2\right)-\mug\left(|h_{21}|^2\right)\Big]^+ 
             +\mug(|h_{21}|^2+|h_{22}|^2),  
  \ \text{from~\cite{kramer_outer}},
\label{eq:R upper classical IC kra1}
\\
R_1+R_2 & \leq  \Big[\mug\left(|h_{22}|^2\right)-\mug\left(|h_{12}|^2\right)\Big]^+ 
             +\mug(|h_{11}|^2+|h_{12}|^2),  
  \ \text{from~\cite{kramer_outer}},
\label{eq:R upper classical IC kra2}
\\
R_1+R_2 & \leq \mug \left(|h_{12}|^2 +\frac{|h_{11}|^2 }{1+|h_{21}|^2} \right)
            + \mug \left(|h_{21}|^2 +\frac{|h_{22}|^2 }{1+|h_{12}|^2} \right),
  \ \text{from~\cite{etkin_tse_wang}},
\label{eq:R upper classical IC etw}
\\
2R_1+R_2& \leq \mug(|h_{11}|^2+|h_{12}|^2)+\mug\left(|h_{21}|^2+\frac{|h_{22}|^2}{1+|h_{12}|^2}\right)
  \notag\\&   +\Big[\mug\left(|h_{11}|^2\right)-\mug\left(|h_{21}|^2\right)\Big]^+,
  \ \text{from~\cite{etkin_tse_wang}},
\label{eq:out:2r1r2}
\\
R_1+2R_2 & \leq \mug(|h_{21}|^2+|h_{22}|^2)+\mug\left(|h_{12}|^2+\frac{|h_{11}|^2}{1+|h_{21}|^2}\right)
  \notag\\&   +\Big[\mug\left(|h_{22}|^2\right)-\mug\left(|h_{12}|^2\right)\Big]^+,
  \ \text{from~\cite{etkin_tse_wang}} \quad \Big\}.
\label{eq:out:r12r2}
\end{align}
\label{eq:R upper classical IC}
For the classical G-IC where all nodes are synchronous and possess full codebook knowledge, 
%the region $\mathcal{R}_{\text{out}}(\snr,\inr)$ 
this outer bound is tight in strong interference $\{|h_{21}|^2 \geq |h_{11}|^2, \ |h_{12}|^2 \geq |h_{22}|^2 \}$~\cite{sato_strong}
%(i.e., for $\snr\leq \inr$)
and achievable to within $1/2$~bit otherwise~\cite{etkin_tse_wang}. 
\end{subequations}
\end{prop}

The key step to match, to within an additive gap, %each point on the closure of 
the outer bound region $\mathcal{R}_{\text{out}}$ in Proposition~\ref{prop:ETWouter} to our {TINnoTS} achievable region with mixed inputs $\mathcal{R}_{\text{in}}$ in Proposition~\ref{prop:ach-with-mixedinput}
is to carefully choose the mixed input parameter vector $[N_1,N_2,\delta_1,\delta_2]$.  This `carefully picking of the mixed input parameters' is the objective of Section~\ref{sec:cap within gap sym}.

\section{Symmetric Capacity Region to within a Gap}
\label{sec:cap within gap sym}
The main result of this paper is:
\begin{thm}
\label{thm:Cap:gap sym}
For the symmetric G-IC, as defined in~\eqref{eq:awgn sym}, the {TINnoTS} achievable region %$\mathcal{R}_{\text{in}}$
in~\eqref{eq: rates mixed inputs WITH Union}, with the parameters for the mixed inputs chosen as indicated in Table~\ref{table:input},
and the outer bound %$\mathcal{R}_{\text{out}}$  
in~\eqref{eq:R upper classical IC} are to within a gap of: 
\begin{itemize}

\item Very Weak Interference:
$\snr \geq \inr (1+\inr)$: 
\[
{\gap} \leq \frac{1}{2}~\text{bits},
\]

\item Moderately Weak Interference Type2:
$ \snr < \inr(1+\inr), \ \frac{1+\snr}{1+\inr+\frac{\snr}{1+\inr}} > \frac{1+\inr+\frac{\snr}{1+\inr}}{1+\frac{\snr}{1+\inr}}$: 
\begin{align*}
{\gap} %_{\eqref{eq: gap weak 2}}=
\frac{1}{2}\log \left(\frac{608 \ \pi \eu}{27}\right) \approx 3.79
\end{align*}

\item Moderately Weak Interference Type1:
$\inr \leq \snr, \ \frac{1+\snr}{1+\inr+\frac{\snr}{1+\inr}} \leq \frac{1+\inr+\frac{\snr}{1+\inr}}{1+\frac{\snr}{1+\inr}}$: 
\[
 {\gap}%_{\eqref{eq: gap weak 1}} 
 \leq   \frac{1}{2}\log \left(\frac{ 16\pi \eu}{3 }\right)
        %\frac{1}{2}\log\left(4\pi\eu \right)
 +\frac{1}{2}\log\left(1+45 \cdot \frac{(1+1/2\ln(1+\min(\inr,\snr)))^2}{\gamma^2}\right) ~\text{bits}, \ 
\]
except for a set of measure $\gamma$ for any $\gamma \in (0,1]$,

\item Strong Interference:
$\snr < \inr < \snr(1+\snr)$: 
\[
 {\gap}%_{\eqref{eq:achregion for par:strong}}
 \leq \frac{1}{2}\log\left(\frac{2\pi\eu}{3}\right)
  +\frac{1}{2}\log\left(1+8\cdot\frac{\left(1+1/2\ln(1+\min(\inr,\snr))\right)^2}{\gamma^2}\right)~\text{bits},
\]
except for a set of measure $\gamma$ for any $\gamma \in (0,1]$,

\item  Very Strong Interference:
$\inr \geq \snr(1+\snr)$:
\[
 {\gap}%_{\eqref{eq:achregion for par:strong}} 
 \leq \frac{1}{2}\log\left(\frac{2\pi\eu}{3}\right)\approx 1.25~\text{bits}. 
\]

\end{itemize}
\end{thm}

Before we move to the proof of Theorem~\ref{thm:Cap:gap sym}, we would like to offer our thoughts on why a $\loglog$ gap is obtained in some regimes up to an outage set of controllable measure (the larger the measure of the channel gains for which the derived gap does not hold, the lower the gap).
We start by noticing that, for the symmetric G-IC, whenever the {TINnoTS} region with our mixed input is optimal to within a {\it constant gap} then the gap result holds for {\it all channel gains}. Otherwise, the optimality is to within a {\it $\loglog$ gap} and holds for all channel gains {\it up to an outage set}.

We found a $\loglog$ gap up to an outage set whenever the sum-rate upper bound $\min\big(\text{eq.\eqref{eq:R upper classical IC kra1}, eq.\eqref{eq:R upper classical IC kra2}\big)}$ is active, which in gDoF corresponds to the regime $\alpha\in(2/3,2)$ meaning that the interference is neither very weak nor very strong. 
It thus natural to ask: (a) whether the $\loglog$ gap and/or the `up to an outage set' condition are necessary (not a consequence of the achievable scheme used), and (b) whether a $\loglog$ gap and the `up to an outage set' condition are necessarily always together.
We do not have answers to these questions, but we provide our perspective next.

The sum-rate bounds in~\eqref{eq:R upper classical IC kra1} and~\eqref{eq:R upper classical IC kra2} were originally derived for the classical two-user IC in Gaussian noise in~\cite{kramer_outer} and then extended to any memoryless two-user IC with source cooperation / generalized feedback in~\cite{tuninetti-outer}, and then to any memoryless cooperative two-user IC (where each node can have an input and an output to the channel) in~\cite{rini:ITW20102} -- see also $K$-user extensions in~\cite{tuninetti2011k,DM:JSAC}.  In~\cite{rini:ITW20102} it was noted that surprisingly these bounds hold for a broad class of two-user IC-type channels, which includes for example cognitive ICs and certain ICs with cooperation. The difference is that the mutual information optimization is over all product input distributions for the classical IC, while it is over all joint input distributions for the cooperative or cognitive IC. The ability to correlate inputs in well known to only  increase the rates by a constant number of bits; thus, up to a constant gap, channel models from the basic classical IC to the intricate cognitive IC have the same sum-rate upper bound in some regimes. Note that for the real-valued cognitive G-IC for example, the sum-rate bound is achievable to within 1/2~bit for all channel gains by using Dirty Paper Coding. 
It is not clear at this point whether the $\loglog$ gap up to an outage set for the classical G-IC is thus a fundamental consequence of the fact that the upper bound can be achieved to within a constant gap with sophisticated coding techniques (such as Dirty Paper Coding for the cognitive G-IC) but not with simpler ones (essentially rate splitting and superposition coding as in the  Han-Kobayashi scheme) allowed for the classical G-IC. 

Another intriguing observation is that these bounds also determine the optimality of ``everybody gets half the cake''-DoF result for the $K$-user G-IC~\cite{Jafar:2008:alignment,MatahariRealInteferenceAlignment}. For the three-user G-IC with fixed channel gains it is well known that the DoF are discontinuous at rational channel gains~\cite{ordentlichetal}. This seems to suggest, at least for $\alpha=1$, that a gap result up to an outage set 
%(which seems to imply a general $\loglog$ gap result) 
is actually {\it fundamental} and not a consequence of the achievable scheme used. Whether the converse result of~\cite{ordentlichetal}  for $\alpha=1$ can be extended to the whole regime $\alpha\in(2/3,2)$ is an open question. 
We also note that a constant (not $\loglog$) gap result up to an outage set for the whole regime $\alpha\in(2/3,2)$ was found in~\cite{KuserIC}; in this case the achievable region was based on a {\it multi-letter} scheme inspired by compute-and-forward. It is not clear at this point whether {\it single-letter} schemes, such as out {TINnoTS}, are fundamentally suboptimal compared to multi-letter ones.

\begin{table*}
 \centering
 \caption{Parameters for the mixed inputs in~\eqref{eq:GICmixedinput}, as used in the proof of Theorem~\ref{thm:Cap:gap sym}. 
 Notation: for $\mathbf{p}=[N_1,N_2,\delta_1,\delta_2]$ we let
 $\mathbf{p}^{\prime}=[N_2,N_1,\delta_2,\delta_1]$. 
 We also define $\alpha=\lim_{\snr \to \infty} \frac{\log(\inr)}{\log(\snr)}$.
%  \dt{recheck before submitting}
 }
 \label{table:input}
 \begin{tabular} { |l| c |}
 \hline
 Regime 
 & 
 Input Parameter $\mathbf{p}$ in~\eqref{eq:GICmixedinput} 
 
 % % --------------- 
 % -- START REGIME
 \\ \hline
 $\snr \geq \inr(1+\inr),$
 &
 $\mathbf{p}_t \cup \mathbf{p}^{\prime}_t$, for all ${t\in[0,1]}$;
 \\
 $ \alpha\in[0,1/2]$ 
 &
 $\mathbf{p}_t :=  \left[1, 1, t, 1\right]$;
 \\
 (very weak) constant gap
 &  
 %nothing here
 % -- END REGIME

 % ---------------
 % -- START REGIME
 \\ \hline
 $ \snr < \inr(1+\inr), \ \frac{1+\snr}{1+\inr+\frac{\snr}{1+\inr}} > \frac{1+\inr+\frac{\snr}{1+\inr}}{1+\frac{\snr}{1+\inr}},$ 
 &
 $\mathbf{p}_{1,t}  \cup \mathbf{p}_{2,t} \cup    \mathbf{p}_{2,t}^\prime$, for all $t \in [0,1]$;
 \\
 $\alpha\in(1/2,2/3)$ 
 &
 $\mathbf{p}_{1,t}$ : values can be found in~\eqref{eq:parameters R1+R2 face weak type 2};
 \\
 (moderately weak 2) constant gap
 &  
 $\mathbf{p}_{2,t}$ : values can be found in~\eqref{eq:parameters 2R1+R2 face weak type 2: second choice};

 % ---------------
 % -- START REGIME
 \\ \hline
 $\inr \leq \snr, \ \frac{1+\snr}{1+\inr+\frac{\snr}{1+\inr}} \leq \frac{1+\inr+\frac{\snr}{1+\inr}}{1+\frac{\snr}{1+\inr}}, $
 &
 $\mathbf{p}_{1,t} \cup \mathbf{p}_{2,t} \cup    \mathbf{p}_{2,t}^\prime$, for all $t \in [0,1]$;
 \\
 $ \alpha\in[2/3,1]$ 
 & 
 $\mathbf{p}_{1,t}$ values can be found in~\eqref{eq:achregion for par:moderate choiceALL: R1+R2sumrate};
 \\
 (moderately weak 1) log-log gap
 &
 $\mathbf{p}_{2,t}$ values can be found in~\eqref{eq:achregion for par:moderate choiceALL: 2R1+R2};

 % ---------------
 % -- START REGIME
 \\ \hline
 $\snr < \inr < \snr(1+\snr),$ 
 &
 $\mathbf{p}_{t}$, for all ${t\in[0,1]}$;
 \\
 $ \alpha\in(1,2)$
 & 
 $\mathbf{p}_{t}$ : values can be found in~\eqref{eq:achregion for par:strong choiceN};
 \\
 (strong) log-log gap
 &

 % ---------------
 % -- START REGIME
 \\ \hline
 $\inr \geq \snr(1+\snr), $ 
 &
 $\mathbf{p} = [\points(\snr),\points(\snr),0,0]$;
 \\
 $\alpha\in[2,\infty)$ 
 &
 %nothing here
 \\
 (very strong) constant gap
 &
 
 \\ \hline
\end{tabular}
\end{table*}

%\clearpage

\begin{IEEEproof}
The parameters of the mixed inputs in~\eqref{eq:GICmixedinput} are chosen as indicated in Table~\ref{table:input} depending on the regime of operation. We now analyze each regime separately.

% -------------------------------------------------------------------------------------------------- %
\subsection{Very strong interference, i.e., $\inr \geq \snr(1+\snr)$} 
\label{par:verystrong}

\paragraph*{Outer Bound}
In the very strong interference regime the capacity of the classical G-IC is given by 
\begin{align}
\mathcal{R}^{(\text{\ref{par:verystrong}})}_{\text{out}} %(\snr,\inr)
&= \left \{ %(R_1,R_2):
\begin{array}{l}
0\leq R_1 \leq \mug \left(\snr\right)
\\
0\leq R_2 \leq \mug \left(\snr\right)
\end{array} \right\}.
\label{eq:achregion for par:verystrong:RU}
\end{align}

\paragraph*{Inner Bound}
The capacity of the classical G-IC in this regime is achieved by sending only common messages from Gaussian codebooks;
a receiver first decodes the interfering message, strips it from the received signal, and then decodes the intended message in an equivalent interference-free channel.
Even though joint decoding is not allowed in our {TINnoTS} region, we shall see that the discrete part of the input behaves as a common message (as if it could be decoded at the non-intended destination). We therefore do not send the Gaussian portion of the input (as Gaussian inputs treated as noise increase the noise floor of the receiver) and in~\eqref{eq:GICmixedinput} we set 
\begin{subequations}
\begin{align}
&N_1=N_2=N= \points\left(\snr\right),
\label{eq:very strong:N}
\\
&\delta_1=\delta_2=\delta=0,
\label{eq:very strong:delta}
\end{align}
\end{subequations}
resulting in 
%$\mud({S}_1)=\mud({S}_2)=:\mud(S)$ in the achievable region in Proposition~\ref{prop:ach-with-mixedinput}, where
\begin{align}
{S}_1 \sim {S}_2 \sim {S}, \quad
%\mud({S}_1)=\mud({S}_2)=\mud(S) : 
{S} := \sqrt{\snr}X_{1D}+\sqrt{\inr}X_{2D},
\label{eq:very strong:S=S1=S2}
\end{align}
for the received constellations in~\eqref{eq:GICmixedinputinnerBound}.
%With this choice of $\delta$ and $N$ it follows immediately that the constellation seen at the receivers in~\eqref{eq:GICmixedinputinnerBound} satisfy $\mud({S}_1)=\mud({S}_2)=:\mud(S)$ in the achievable region in Proposition~\ref{prop:ach-with-mixedinput}. 
The number of points and the minimum distance for the constellation ${S}$ in~\eqref{eq:very strong:S=S1=S2}
%of the constellations seen at receivers in ~\eqref{eq:sumSetAtRx1}-\eqref{eq:sumSetAtRx2}
can be computed from Proposition~\ref{prop:combNOOUTAGE} as follows.
If we identify $|h_x|^2=\snr$, $|h_y|^2=\inr$, $|X|=|Y|=N$, $d_{\min(X)}^2=d_{\min(Y)}^2 = \frac{12}{N^2-1}$, then the condition in~\eqref{eq:thm:comb conditions} reads $N^2 \snr \leq \inr$, which is readily verified since $N^2 \snr \leq (1+\snr)\snr $ by definition of $N$ in~\eqref{eq:very strong:N}, and $(1+\snr)\snr \leq \inr$ by the definition of the very strong interference regime. 
We therefore have
\begin{align}
%&
%\mathcal{E}_{{S}} = \snr + \inr,
%\\
&|{S}| = N^2, \quad \text{with equally likely points},
\\
&\frac{   d_{\min({S})}^2 }{12} 
=
\min\{\snr,\inr\} \frac{1}{N^2-1}
= \frac{\snr }{N^2-1}.
\end{align}
By plugging these values in Proposition~\ref{prop:ach-with-mixedinput}, an achievable rate region is
\begin{subequations}
\begin{align}
\mathcal{R}^{(\text{\ref{par:verystrong}})}_{\text{in}} %(\snr,\inr) 
&= \left \{ %(R_1,R_2):
\begin{array}{l}
0\leq R_1 \leq r_0
\\
0\leq R_2 \leq r_0
\end{array} \right\}
\ \text{such that}
\label{eq:achregion for par:verystrong:RL}
\\
r_0
  &\geq  \mud \left({S} \right)
- \min \left( \log(N), \mug \left(\inr \right) \right)
\notag
\\&\geq %   since \min \left( \log(N), \mug \left(\inr \right) \right) \leq \log(N)
\left[
\log(N^2)
- \frac{1}{2}\log\left(\frac{2\pi\eu}{12}\right) 
- \frac{1}{2}\log\left(1+\frac{N^2-1}{\snr} %\frac{\snr + \inr}{\snr + \inr+1}
\right)\right]^+
- \log(N)
\notag
\\&\geq \mug \left( \snr \right) - {\gap}_{\eqref{eq:achregion for par:verystrong}},
\\{\gap}_{\eqref{eq:achregion for par:verystrong}} &:=  
%  \frac{1}{2}\log\left(\frac{1+\snr}{N^2}\right)
%+ \frac{1}{2}\log\left(\frac{2\pi\eu}{12}\right)
%+ \frac{1}{2}\log\left(1+\frac{N^2-1}{\snr }%\frac{\snr + \inr}{\snr + \inr+1}
%\right)
%\notag
%\\&\leq \frac{1}{2}\log\left(\frac{2\pi\eu}{3}\right)\approx 1.25~\text{bits},
\frac{1}{2}\log\left(\frac{2\pi\eu}{3}\right) \approx 1.25~\text{bits}, 
\label{eq:achregion for par:verystrong:gap}
\end{align}
where the gap in~\eqref{eq:achregion for par:verystrong:gap} is as for the point-to-point Gaussian channel without states in Section~\ref{sec:ptpExample}.%, see also Fig.~\ref{fig:1stexamplegap}. 
%where the maximum ${\gap}_{\eqref{eq:achregion for par:verystrong}}$ is attained for %$\inr\gg1, \ 
%$\snr = 3 -\epsilon$ with $0<\epsilon \ll 1$, for which $N = \points(\snr)=1$.
\label{eq:achregion for par:verystrong}
\end{subequations}

\paragraph*{Gap}
It is immediate to see that the achievable region in~\eqref{eq:achregion for par:verystrong} and the upper bound in~\eqref{eq:achregion for par:verystrong:RU} are at most to within ${\gap}_{\eqref{eq:achregion for par:verystrong}}$~bits of one another, where ${\gap}_{\eqref{eq:achregion for par:verystrong}}$ is given in~\eqref{eq:achregion for par:verystrong:gap}.

% -------------------------------------------------------------------------------------------------- %
\subsection{Strong (but not very strong) interference, i.e., $\snr  <  \inr < \snr(1+\snr)$} % \Longleftrightarrow 1< \alpha< 2
\label{par:strong}

\paragraph*{Outer Bound}
The capacity region of the G-IC in this regime is 
%given by the constraints in~\eqref{eq:R upper classical IC cuset r1},~\eqref{eq:R upper classical IC cuset r2} and~\eqref{eq:R upper classical IC kra}, which can be rewritten as
\begin{align}
\mathcal{R} _{\text{out}}^{(\text{\ref{par:strong}})} %(\snr,\inr)
&=  
\left\{\begin{array}{l}
0\leq R_1 \leq \mug\left(\snr\right)
\\ 
0\leq R_2 \leq \mug\left(\snr\right)
\\
R_1+R_2 \leq \mug(\snr+\inr)
\end{array}\right\}
\notag
\\&= \bigcup_{t\in[0,1]} 
\left\{\begin{array}{ll}
0\leq R_1 &\leq
 \frac{1-t}{2}\log\left(1+\frac{\inr}{1+\snr}\right)
+\frac{t  }{2}\log\left(1+\snr\right)
\\&=: \mug\left(\snr_{0,a,t}\right)
\\ 
0\leq R_2 &\leq
 \frac{1-t}{2}\log\left(1+\snr\right)
+\frac{t  }{2}\log\left(1+\frac{\inr}{1+\snr}\right)
\\&=: \mug\left(\snr_{0,b,t}\right) 
\end{array}\right\},
\label{eq:outerregion for par:strong}
\end{align}
where $t\in[0,1]$ is the time-sharing parameter (i.e., by varying $t$ we obtain all points on the dominant face of the capacity region described by $R_1+R_2 = \mug(\snr+\inr)$).

\paragraph*{Inner Bound}
The capacity of the classical G-IC in this regime is achieved by sending only common messages from Gaussian codebooks,
and by performing joint decoding of the intended and interfering messages at both receivers. Similarly to the very strong interference regime, we do not send the Gaussian portion of the mixed inputs (i.e., $\delta_1=\delta_2=0$). Differently from the very strong interference regime, here we do not set the number of points of the discrete part of the inputs to be the same for the two users since the corner point of~\eqref{eq:outerregion for par:strong} for a fixed $t$ has $R_1\not=R_2$.
%for a reason that will become apparent shortly. 
Moreover, we lower bound the minimum distance of the sum-set constellations ${S}_1$ and ${S}_2$ in~\eqref{eq:GICmixedinputinnerBound} by using Proposition~\ref{prop:card:dmin:measurebound} as follows 
%%
%In the achievable region in~\eqref{eq: rates mixed inputs WITH Union} 
%%in Proposition~\ref{prop:ach-with-mixedinput}
%we have
\begin{align}
%&
%\mathcal{E}_{{S}_1} = \mathcal{E}_{{S}_2} = \snr + \inr,
%\\
%&|{S}_1|=|{S}_2|=N_1N_2 \ \text{a.e., where the points are equally likely,}
%\\
&\frac{d_{\min({S}_1)}^2}{12} \geq \kappa_{\gamma,N_1,N_2}^2\min\left(\frac{\snr}{N_1^2-1},\frac{ \inr}{N_2^2-1}, \max \left(\frac{\inr}{ N_1^2(N_2^2-1) },\frac{\snr}{ N_2^2(N_1^2-1) } \right)\right)  ,
\label{eq:achregion for par:strong dminS1lowerbound qaz}
\\
&\frac{d_{\min({S}_2)}^2}{12} \geq \kappa_{\gamma,N_1,N_2}^2 \min\left(\frac{\inr}{N_1^2-1},\frac{ \snr}{N_2^2-1}, \max \left(\frac{\inr }{ N_1^2(N_2^2-1) },\frac{\snr}{ N_2^2(N_1^2-1) } \right)\right),
\label{eq:achregion for par:strong dminS2lowerbound qaz}
\\
&\kappa_{\gamma,N_1,N_2} := \frac{\gamma/2}{1+\ln(\max(N_1,N_2))},
\label{eq:achregion for par:strong kappa qaz}
\end{align}
where the minimum distance lower bounds in~\eqref{eq:achregion for par:strong dminS1lowerbound qaz} and~\eqref{eq:achregion for par:strong dminS2lowerbound qaz} hold for all channel gains up to an outage set of Lebesgue measure less than $\gamma$ for any $\gamma \in (0,1]$.

By combining the bounds in~\eqref{eq:achregion for par:strong dminS1lowerbound qaz} and~\eqref{eq:achregion for par:strong dminS2lowerbound qaz} we obtain
\begin{align}
\min_{i\in[1:2]}\frac{d_{\min ({S}_i)}^2}{12} 
&\geq   \kappa_{\gamma,N_1,N_2}^2 \min\left( 
\frac{\min(\snr,\inr)}{\max(N_1^2,N_2^2)-1},
\frac{\max(\snr,\inr)}{N_1^2N_2^2-1} \right) \notag
\\& \stackrel{\text{for $\snr \leq \inr$}}{=} 
\kappa_{\gamma,N_1,N_2}^2 \min\left( 
\frac{\snr}{\max(N_1^2,N_2^2)-1},
\frac{\inr}{N_1^2N_2^2-1} \right).
%\\\kappa_{\gamma,N_1,N_2}&=\frac{\gamma}{2 (1+\ln(\max(N_1,N_2))}.
\label{eq:achregion for par:strong mindmin}
\end{align}
With~\eqref{eq:achregion for par:strong mindmin}, it can be easily seen that  
the achievable region %in~\eqref{eq: rates mixed inputs BEFORE Union}
in Proposition~\ref{prop:ach-with-mixedinput} can be written as the union over all $(N_1,N_2)$ of the region
\begin{subequations}
\begin{align}
\mathcal{R}_{\text{in}}^{(\text{\ref{par:strong}})}&\left([N_1,N_2,0,0]\right)  %\snr,\inr;
= \left \{ %(R_1,R_2):
\begin{array}{l}
0\leq R_1 \leq r_{1} 
\\
0\leq R_2 \leq r_{2} 
\end{array} \right\} \ \text{such that} 
\label{eq:achregion for par:strong regionB}
\\
r_{1} 
  &
  \geq \mud \left({S}_1\right)
-\min \Big( \log(N_2),\mug\left(\inr\right) \Big)
\geq \log(N_1)+\log(2) - {\gap}_{\eqref{eq:achregion for par:strong}},
\label{eq:achregion for par:strong r1}
\\
r_{2}  
  & 
\geq  \mud \left({S}_2\right)
-\min \Big( \log(N_1),\mug\left(\inr\right) \Big)
\geq \log(N_2)+\log(2) - {\gap}_{\eqref{eq:achregion for par:strong}},
\label{eq:achregion for par:strong r2}
\\
{\gap}_{\eqref{eq:achregion for par:strong}} &\leq 
  \log(2)
+ \frac{1}{2}\log\left(\frac{2\pi\eu}{12}\right) \notag
\\&
+ \frac{1}{2}\log\left(1+\frac{1}{\kappa_{\gamma,N_1,N_2}^2}\max\left( \frac{\max(N_1^2,N_2^2)-1}{\snr}, \frac{N_1^2N_2^2-1 }{ \inr}\right) %\frac{\snr + \inr}{\snr + \inr+1}
\right), 
\label{eq:achregion for par:strong gap}
\end{align}
where the expression for ${\gap}_{\eqref{eq:achregion for par:strong}}$ comes from the minimum distance expression in~\eqref{eq:achregion for par:strong mindmin}.
\label{eq:achregion for par:strong}
\end{subequations}

We next need to pick $N_1$ and $N_2$ in~\eqref{eq:achregion for par:strong}.
Our choice is guided by the expression of the `compound MAC' capacity region in this regime given by~\eqref{eq:outerregion for par:strong}. In our {TINnoTS} region, time-sharing is not allowed, but varying the number of points of the discrete constellations is; we therefore mimic time-sharing in~\eqref{eq:outerregion for par:strong} by choosing as number of points in the discrete part of the mixed inputs as follows: for some fixed $t$ we let
\begin{subequations}
\begin{align}
N_1 &= \points\left(\snr_{0,a,t}\right), \
\snr_{0,a,t} := \left(1+\frac{\inr}{1+\snr}\right)^{1-t}\left(1+\snr\right)^{t}-1,
\label{eq:achregion for par:strong choiceN1}
\\
N_2 &= \points\left(\snr_{0,b,t}\right), \
\snr_{0,b,t} := \left(1+\frac{\inr}{1+\snr}\right)^{t}\left(1+\snr\right)^{1-t}-1.
\label{eq:achregion for par:strong choiceN2}
\end{align}
\label{eq:achregion for par:strong choiceN}
\end{subequations}
The whole {TINnoTS} achievable region is obtained by taking union over $t\in[0,1]$ of the region in~\eqref{eq:achregion for par:strong} with the number of points as in~\eqref{eq:achregion for par:strong choiceN}.

\paragraph*{Gap}
Since $\frac{\inr}{1+\snr} \leq \snr  \leq \inr$ by the definition of the strong interference regime, we immediately have that in~\eqref{eq:achregion for par:strong choiceN} the equivalent SNRs satisfy $\max(\snr_{0,a,t},\snr_{0,b,t})\leq \snr$ for all $t\in[0,1]$. Thus, for the minimum distance expression in~\eqref{eq:achregion for par:strong mindmin}, we have 
\begin{align}
&\max(N_1^2,N_2^2)-1 \leq \max(\snr_{0,a,t},\snr_{0,b,t}) \leq \snr = \min(\snr,\inr),
\label{eq:condition-strong-to-further-bound-gap-1}
\\
&N_1^2N_2^2-1 \leq (\snr_{0,a,t}+1)(\snr_{0,b,t}+1)-1=\snr+\inr \leq 2\inr.
\label{eq:condition-strong-to-further-bound-gap-2}
\end{align}
Finally, since $\mug(x) \leq  \log(\points(x)) +\log(2)$, the inner bound in~\eqref{eq:achregion for par:strong} is at most ${\gap}_{\eqref{eq:achregion for par:strong}}$~bits from the outer bound in~\eqref{eq:outerregion for par:strong}, uniformly over all $t\in[0,1]$, where ${\gap}_{\eqref{eq:achregion for par:strong}}$ in~\eqref{eq:achregion for par:strong gap} can be further upper bounded thanks to~\eqref{eq:condition-strong-to-further-bound-gap-1}-\eqref{eq:condition-strong-to-further-bound-gap-2} as
%. By upper bounding ${\gap}_{\eqref{eq:achregion for par:strong}}$
%% so as to `eliminate' all channel parameters 
%we obtain that in the strong interference the {TINnoTS} scheme is optimal to within at most
\begin{align}
 {\gap}_{\eqref{eq:achregion for par:strong}} &\leq 
%  \log(2)
%+ \frac{1}{2}\log\left(\frac{2\pi\eu}{12}\right) \notag
%\\&
%+ \frac{1}{2}\log\left(1+\frac{1}{\kappa_{\gamma,N_1,N_2}^2}\max\left( \frac{\max(N_1^2,N_2^2)-1}{\snr}, \frac{N_1^2N_2^2-1}{ \inr}\right) %\frac{\snr + \inr}{\snr + \inr+1}
%\right) \notag\\
%&= 
\frac{1}{2}\log\left(\frac{2\pi\eu}{3} \left(1+\frac{\max(1,2)}{\kappa_{\gamma,N_1,N_2}^2}
\right)\right) \notag\\&
\leq \frac{1}{2}\log\left(\frac{2\pi\eu}{3} \left(1+8 \cdot \frac{\left(1+1/2\ln({1+\min(\snr,\inr)})\right)^2}{\gamma^2}\right)\right),
\end{align}
where $\gamma$ is the Lebesgue measure of the outage set over which the lower bounds on the minimum distance in~\eqref{eq:achregion for par:strong mindmin} does not apply. Recall that $\gamma$ is a tunable parameter that represents a tradeoff between gap and set of channel gains for which the gap result holds, i.e., by increasing the measure of the outage set we can reduce the gap, and vice-versa. A similar behavior was pointed out already in~\cite{KuserIC}. 

\begin{rem}\label{rem:why cannot use Prop nooutage}
Note that, had we been able to use Proposition~\ref{prop:combNOOUTAGE} instead of Proposition~\ref{prop:card:dmin:measurebound} to bound the minimum distance of the received constellations, we would have obtained a constant gap result instead of a $\loglog$ gap result. 
It turns out that in this regime the condition of Proposition~\ref{prop:combNOOUTAGE} is not satisfied -- the proof is very tedious and is not reported here for sake of space.
%To see that we, indeed, can not use Proposition~\ref{prop:combNOOUTAGE} please refer to  the discussion in Appendix~\ref{app: Prop2 validity}.
\end{rem}

%-----------------------------------------------------------------------------%
\subsection{Moderately weak interference, i.e. $\inr \leq \snr \leq (1+\inr) \inr$: general setup}
\label{par:weak}

The weak interference regime is notoriously more involved to analyze than the other regimes.
In this subsection we aim to derive a general framework to deal with the weak (but not very weak) interference regime.
Before we move into the gap derivation for this regime, let us summarize the key trick we developed in the strong interference regime to obtain a capacity result to within a gap: write the closure of the capacity outer bound in parametric form so as to get insight on how to choose the number of points of the discrete part of the mixed inputs. In the weak interference regime we will follow the same approach but the computations will be more involved because the capacity region outer bound in weak interference has three dominant faces (and not just one dominant face as in strong interference).

\paragraph*{Outer Bound}

In this regime, we express the upper bound in Proposition~\ref{prop:ETWouter} as the convex closure of its corner points, that is
\begin{subequations}
\begin{align}
\mathcal{R}_{\text{out}}^{(\text{\ref{par:weak}})} %(\snr,\inr)
 = \co\Big\{ 
(R_{1A},R_{2A})     &:= \left( \mug(\snr), c \right), \\
(R_{1B},R_{2B})     &:= \left( b-a, 2a-b \right), \\
%(R_\text{1sym},R_\text{2sym}) &:= \left( \frac{a}{2}, \frac{a}{2} \right), \\
(R_{1C},R_{2C})     &:= \left( 2a-b,b-a \right), \\
(R_{1D},R_{2D})     &:= \left( c, \mug(\snr) \right) \
\Big\},
\end{align}
where
\begin{align}
a &:= \min\left( 
\mug(\inr+\snr)+\mug\left(\snr\right)-\mug\left(\inr\right), 
2 \, \mug \left(\inr +\frac{\snr }{1+\inr} \right)
\right), \label{eq:out:rsum}
\\
 b &:= \mug\left(\inr+\frac{\snr}{1+\inr}\right)+ \mug(\snr+\inr)+\mug\left(\snr\right)-\mug\left(\inr\right),
 \label{eq:out:weired}
\\
 c &:= \mug(\inr+\snr)-\mug\left(\snr\right)
 + \mug \left(\inr +\frac{\snr }{1+\inr} \right)-\mug\left(\inr\right)
 \label{eq:out:evenmoreweired}
\\&= \mug \left(\frac{\inr }{1+\snr} \right)+ \mug \left(\frac{\snr }{(1+\inr)^2} \right)
\notag
\\&\leq \mug \left(\frac{\snr }{1+\snr} \right)+ \mug \left(\frac{\inr }{1+\inr} \right)
   \leq \log(2).
\notag
\end{align}
Under the constraint $\frac{\snr}{1+\inr}\leq \inr$  it can be verified numerically that actually $c\leq  0.5537$~bits (rather than $c\leq 1$~bit) attained for $\inr=\sqrt{3}+1$; however, for notational convenience we will use in the following $c\leq 1$~bit.
\label{eq:all-corner-points-in-weak}
\end{subequations}

An explicit expression for $\mathcal{R}_{\text{out}}^{(\text{\ref{par:weak}})}$
obtained by time-sharing between the corner points in~\eqref{eq:all-corner-points-in-weak}  is
\begin{subequations}
\begin{align}
&  \mathcal{R} _{\text{out}}^{(\text{\ref{par:weak}})} 
=     \mathcal{R}_{2R_1+R_2}^{(\text{\ref{par:weak}})}
 \cup  \mathcal{R}_{R_1+R_2}^{(\text{\ref{par:weak}})}
 \cup \mathcal{R}_{R_1+2R_2}^{(\text{\ref{par:weak}})},
 \quad \text{where}
\label{eq:outer bound in weak union of 3}
\\
&\mathcal{R}_{2R_1+R_2}^{(\text{\ref{par:weak}})}
= \bigcup_{t \in [0,1]} \left \{ \begin{array}{l} 
R_1 \leq tR_{1A} + (1-t)R_{1B} \\ 
R_2 \leq tR_{2A} + (1-t)R_{2B} \\ 
\end{array} \right\}, % sum-rate 
\label{eq:outer bound in weak otherrate}
\\
&\mathcal{R}_{R_1+R_2}^{(\text{\ref{par:weak}})}
= \bigcup_{t \in [0,1]} \left \{ \begin{array}{l} 
R_1 \leq tR_{1B} + (1-t)R_{1C} \\ 
R_2 \leq tR_{2B} + (1-t)R_{2C} \\ 
\end{array} \right\}, % other face 
\label{eq:outer bound in weak sumrate}
\\
&\mathcal{R}_{R_1+2R_2}^{(\text{\ref{par:weak}})}
= \bigcup_{t \in [0,1]} \left \{ \begin{array}{l} 
R_1 \leq tR_{1C} + (1-t)R_{1D} \\ 
R_2 \leq tR_{2C} + (1-t)R_{2D} \\ 
\end{array} \right\}. % other face with role users swapped
\label{eq:outer bound in weak otherfaceswapp}
\end{align}
\label{eq:outer bound in weak ALL}
\end{subequations}

Because the sum-rate upper bound in~\eqref{eq:all-corner-points-in-weak} is in the form 
\[
R_1+R_2\leq \text{eq.\eqref{eq:out:rsum}} 
= \min(\text{eq.\eqref{eq:R upper classical IC kra2}},\text{eq.\eqref{eq:R upper classical IC etw}}),
\] 
we will distinguish between two cases:
when the constraint in~\eqref{eq:R upper classical IC kra2} is active, referred to as {\it Weak1}, and  when the constraint in~\eqref{eq:R upper classical IC etw} is active, referred to as {\it Weak2}, that is, within $\inr \leq \snr \leq \inr(1+\inr)$ we further distinguish between
\begin{align}
&\text{Weak1:} \quad \frac{1+\snr}{1+\inr+\frac{\snr}{1+\inr}} \leq \frac{1+\inr+\frac{\snr}{1+\inr}}{1+\frac{\snr}{1+\inr}},
\label{eq: condition for Weak1}
\\
&\text{Weak2:} \quad \frac{1+\snr}{1+\inr+\frac{\snr}{1+\inr}} > \frac{1+\inr+\frac{\snr}{1+\inr}}{1+\frac{\snr}{1+\inr}}.
\label{eq: condition for Weak2}
\end{align}

\paragraph*{Inner Bound}
For the G-IC in weak interference the best know strategy is to send common and private messages from Gaussian codebooks, and for each of the receivers to jointly decode both common messages and the desired private message while treating the private message of the interferer as noise. 
Unlike in the strong and very strong interference regimes, in this case we will use  the Gaussian portion of the mixed inputs by setting $\delta_1$ and $\delta_2$ to be non-zero.
%so as to send private messages. 
Moreover, we will vary $(\delta_1,\delta_2)$ jointly with $(N_1,N_2)$ %number of points of the discrete part of the mixed inputs
to mimic time sharing and power control.

%\begin{rem}\label{rem:lower-ach-with-mixedinput}
In this regime, we further simplify the achievable rate region in~\eqref{eq: rates mixed inputs WITH Union} from Proposition~\ref{prop:ach-with-mixedinput} as follows
%\begin{subequations}
\begin{align}
&\mathcal{R}_{\text{in}}^{(\text{\ref{par:weak}})} %(\snr,\inr)
= \bigcup_{\small
\begin{array}{c}
[N_1,N_2,\delta_1,\delta_2]\in \mathbb{N}^2\times[0,1]^2 :\\
\max(\delta_1,\delta_2) \leq \frac{1}{1+\inr}\\
\end{array}}
\mathcal{R}_{\text{in}}^{(\text{\ref{par:weak}})}\left([N_1,N_2,\delta_1,\delta_2]\right),%\snr,\inr;
\quad \text{where}
\notag\\
&\mathcal{R}_{\text{in}}^{(\text{\ref{par:weak}})}\left([N_1,N_2,\delta_1,\delta_2]\right)  %\mathbf{p}\snr,\inr;
:= \left \{ %(R_1,R_2):
\begin{array}{l}
0\leq R_1 \leq \log(N_1) %+ \mug\left(\frac{\snr \delta_1}{1+\inr \delta_2}\right)
+ \mug\left(\snr \delta_1\right)
- \Delta_{\eqref{eq: rates mixed inputs BEFORE Union}} 
\\
0\leq R_2 \leq \log(N_2) %+ \mug\left(\frac{\snr \delta_2}{1+\inr \delta_1}\right)
+ \mug\left(\snr \delta_2\right)
- \Delta_{\eqref{eq: rates mixed inputs BEFORE Union}}
\\
\Delta_{\eqref{eq: rates mixed inputs BEFORE Union}} 
= \frac{1}{2}\log\left(\frac{\pi\eu}{3}\right) 
+ \frac{1}{2}\log\left(1+ \frac{12 }{\min_{i \in[1:2]} d_{\min ({S}_i)}^2} \right)
\end{array} \right\},
%\label{eq:achregion simplified}
\label{eq: rates mixed inputs BEFORE Union}
\end{align}
where the received constellations ${S}_1$ and ${S}_2$ are given in~\eqref{eq:GICmixedinputinnerBound}.
Note that, inspired by~\cite{etkin_tse_wang}, we restricted the power splits between the continuous and discrete parts of the mixed inputs to satisfy $\max(\delta_1,\delta_2) \leq \frac{1}{1+\inr}$.
The simplified form of the {TINnoTS} region with mixed inputs in~\eqref{eq: rates mixed inputs BEFORE Union} is obtained from~\eqref{eq: rates mixed inputs WITH Union} as follows.
For the achievable rate $R_1$  
we have
\begin{align*}
R_1 &\geq  \mud \left(S_1\right)
+\mug \left(\frac{\snr\delta_1}{1+\inr\delta_2} \right) 
-\min \left(\log(N_2),\mug\left(\frac{\inr (1-\delta_2)}{1+ \inr\delta_2} \right) \right)
\\&\stackrel{\rm (a)}{=} %H(S_1) = \log(N_1 N_2)
\left[\log(N_1 N_2)
- \frac{1}{2}\log\left(\frac{2\pi\eu}{12}\right) 
- \frac{1}{2}\log\left(1+\frac{12}{d_{\min(S_1)}^2}\right)\right]^+
\\&\quad 
+\mug \left(\frac{\snr\delta_1}{1+\inr\delta_2} \right) 
-\min \left(\log(N_2),\mug\left(\frac{\inr (1-\delta_2)}{1+ \inr\delta_2} \right) \right)
\\&\stackrel{\rm (b)}{\geq}
\log(N_1 N_2)
- \frac{1}{2}\log\left(\frac{2\pi\eu}{12}\right) 
- \frac{1}{2}\log\left(1+\frac{12}{d_{\min(S_1)}^2}\right)
+\mug \left(\frac{\snr\delta_1}{1+\inr\delta_2} \right) 
-\log(N_2)
\\&\stackrel{\rm (c)}{\geq}
\log(N_1)
+\mug \left(\frac{\snr\delta_1}{2} \right) + \frac{1}{2} \log(2)
-\Delta_{\eqref{eq: rates mixed inputs BEFORE Union}}
\\&\stackrel{\rm (d)}{\geq}
\log(N_1)
+\mug \left(\snr\delta_1 \right) 
-\Delta_{\eqref{eq: rates mixed inputs BEFORE Union}},
\end{align*}
where the (in)equalities are due to:
(a) because regardless of whether we use Proposition~\ref{prop:combNOOUTAGE} or Proposition~\ref{prop:card:dmin:measurebound} to compute the minimum distance for the received sum-set constellations 
${S}_1$ and ${S}_2$ in~\eqref{eq:GICmixedinputinnerBound}, these constellations always comprise
$|{S}_1|=|{S}_2|=N_1N_2$ equally likely points either exactly or almost surely;
(b) because  $\left[x\right]^+ \geq x$ and $\min(x,y) \leq x$;
(c) because we imposed $\max(\delta_1,\delta_2) \leq \frac{1}{1+\inr}$ and by definition of $\Delta_{\eqref{eq: rates mixed inputs BEFORE Union}}$ in~\eqref{eq: rates mixed inputs BEFORE Union}; and
(d) because $\log(1+x/2) \geq \log(1+x)-\log(2)$.
The rate expression for user~2 follows similarly.

For the evaluation of $\Delta_{\eqref{eq: rates mixed inputs BEFORE Union}}$,
the minimum distance of the received  constellations ${S}_1$ and ${S}_2$ defined in~\eqref{eq:GICmixedinputinnerBound} will be computed with either Proposition~\ref{prop:combNOOUTAGE} or Proposition~\ref{prop:card:dmin:measurebound}.
By using Proposition~\ref{prop:card:dmin:measurebound}, which is valid for all channel gains up to a set of controllable Lebesgue measure less than $\gamma$, for any $\gamma>0$, 
%to bound the minimum distances of received  constellations 
we have %that
\begin{subequations}
\begin{align}
%|{S}_1|=|{S}_2|&=N_1N_2 \ \text{a.e., where the points are equally likely,}\\
\frac{d_{\min({S}_1)}^2}{12}&\geq \kappa_{\gamma,N_1,N_2}^2 
\frac{ \min\left(\frac{(1-\delta_1)\snr}{N_1^2-1},\frac{(1-\delta_2) \inr}{N_2^2-1}, \max \left(\frac{(1-\delta_2)\inr }{ N_1^2(N_2^2-1) },\frac{(1-\delta_1)\snr}{ N_2^2(N_1^2-1) } \right)\right) }{1+\snr \delta_1+\inr \delta_2} 
\label{eq:achregion for par:moderate dminS1lowerbound eq1}
\\&\geq \kappa_{\gamma,N_1,N_2}^2 
\frac{1-\max(\delta_1,\delta_2)} {1+\snr \delta_1+\inr\delta_2}
\min\left( 
\frac{\snr}{N_1^2-1},
\frac{\inr}{N_2^2-1}, 
\frac{\max(\snr,\inr)}{N_1^2N_2^2-1}
\right)
\label{eq:achregion for par:moderate dminS1lowerbound eq2}
\\&\stackrel{\text{for $\inr\leq \snr$}}{=} \kappa_{\gamma,N_1,N_2}^2 
\frac{1-\max(\delta_1,\delta_2)} {1+\snr \delta_1+\inr\delta_2}
\min\left( 
\frac{\inr}{N_2^2-1}, 
\frac{\snr}{N_1^2N_2^2-1}
\right),  %%AD NOTE%%{\scriptsize \blue \text{THIS IS THE BOUND WE WANT TO USE FOR 2R1+R2} }
\label{eq:achregion for par:moderate dminS1lowerbound}
\\
\frac{d_{\min({S}_2)}^2}{12} &\geq 
%\kappa_{\gamma,N_1,N_2}^2 \cdot \frac{ \min\left(\frac{(1-\delta_1)\inr}{N_1^2-1},\frac{(1-\delta_2) \snr}{N_2^2-1}, \max \left(\frac{(1-\delta_2)\snr }{ N_1^2(N_2^2-1) },\frac{(1-\delta_1)\inr}{ N_2^2(N_1^2-1) } \right)\right) }{1+\snr \delta_2+\inr \delta_1}
%\label{eq:achregion for par:moderate dminS1lowerbound eq3}
%\\&\stackrel{\text{for $\inr\leq \snr$}}{=} 
\kappa_{\gamma,N_1,N_2}^2 
\frac{1-\max(\delta_1,\delta_2)} {1+\snr \delta_2+\inr\delta_1}
\min\left( 
\frac{\inr}{N_1^2-1}, 
\frac{\snr}{N_1^2N_2^2-1}
\right), %%AD NOTE%%{\scriptsize \blue \text{THIS IS THE BOUND WE WANT TO USE FOR 2R1+R2} }
\label{eq:achregion for par:moderate dminS2lowerbound}
\\
\kappa_{\gamma,N_1,N_2}&=\frac{\gamma/2}{1+1/2\ln(\max(N_1^2,N_2^2))}.
\label{eq:achregion for par:moderate kappakappa definitiontobeusedlater}
\end{align}
\label{eq:achregion for par:moderate kappakappa definitiontobeusedlater tuttituttissimi}
\end{subequations} 
If instead we use Proposition~\ref{prop:combNOOUTAGE} %to bound the minimum distances of the received constellations, 
we have %that
\begin{subequations}
\begin{align}
\min_{i \in[1:2]} \frac{d_{\min ({S}_i)}^2}{12}
=
\min_{(i,i')\in \{(1,2),(2,1)\}}
\frac{1}{1+\snr \delta_i+\inr \delta_{i'}} \
\min \left( \frac{(1-\delta_i) \snr}{N_i^2-1},\frac{(1-\delta_{i'}) \inr}{N_{i'}^2-1}\right), 
\label{eq:achregion for par:moderate all parameters dmin P2 min dminpippo}
\end{align}
which holds if 
\begin{align}
\inr (1-\delta_{i^{'}})\frac{N_{i^{'}}^2}{N_{i^{'}}^2-1}\leq \frac{\snr(1-\delta_i)}{N_{i}^2-1} \quad \forall(i,i')\in \{(1,2),(2,1) \}.
\label{eq:achregion for par:moderate all parameters dmin P2 min dmin condcondcond}
\end{align}
\label{eq:achregion for par:moderate all parameters dmin P2}
\end{subequations}

We observe that in~\eqref{eq: rates mixed inputs BEFORE Union} each achievable rate is bounded by the sum of two terms: 
one that depends on the number of points of the discrete part of the mixed inputs, and 
the other that depends on the continuous part of the mixed inputs through the power splits.
This is reminiscent of rate-splitting in the Han-Kobayashi achievable scheme, where each rate is written as the sum of the common-message rate and the private-message rate. 

The simplified Han-Kobayashi achievable region in~\cite{etkin_tse_wang} is known to achieve the outer bound in Proposition~\ref{prop:ETWouter} to within 1/2~bit; however, to the best of our knowledge, it is not known how much information should be conveyed through the private messages and how much through the common messages for a general rate-pair $(R_1,R_2)$ on the convex closure of the outer bound in Proposition~\ref{prop:ETWouter} and for a general set of channel parameters. Next we will identify the (to within 1/2 bit) optimal rate splits and use the found analytical closed-form  expressions for the common-message and private-message rates to come up with an educated guess for the values of the parameters of our mixed inputs.

Let $R_u=R_{u,p}+R_{u,c},$ where $R_{u,p}$ is the rate of the private message and $R_{u,c}$ is the rate of the common message for user $u\in[1:2]$. From the analysis of the symmetric LDA in~\cite[Lemma~4]{bresler_tse}, which gives the optimal gDoF region for the symmetric G-IC before Fourier-Motzkin elimination, it is not difficult to see that it is always optimal to set 
\begin{align}
R_{u,p} \approxeq \min\left(\mug\left(\frac{\snr}{1+\inr}\right), \frac{R_u}{2}\right), \ u\in[1:2], 
\label{eq:optimalratesplit}
\end{align}
where with $\approxeq$ we mean equality up to an additive term that grows slower than $\log(\snr)$ when $\snr\to\infty$.
We found that, with the exception of the sum-capacity for $\alpha\in(1/2,2/3)$, the optimal `rate splits' are unique and are given by~\eqref{eq:optimalratesplit}. 
These `rate splits' shed light on the interplay between private and common messages,
which was not immediately obvious from the outer bound in~\eqref{eq:R upper classical IC}.

In the following it will turn out to be convenient to think of
the discrete   part of a mixed input (contributing to the rate with the term $\log(N_i), i\in[1:2]$) as a `common message' and of
the continuous part of a mixed input (contributing to the rate with the term $\mug(S \delta_i), i\in[1:2]$) as a `private message'. 
We shall refer to this `mapping' of our {TINnoTS} scheme to the Han-Kobayashi scheme as the {\it discrete$\to$common~map}.
Note that there is a fundamental difference between a common message in the Han-Kobayashi achievable scheme and the discrete part of the mixed input in our scheme. In our scheme the interfering signal is treated as noise while in Han-Kobayashi achievable scheme the common message is jointly decoded, albeit non-uniquely, with the intended signals at the non-intended receiver. The {\it discrete$\to$common~map} is thus just intended to provide an educated guess on how to pick the parameters of our mixed input in the following analysis. 
We do not claim here that the {\it discrete$\to$common~map} is the only possible way to `match' our {TINnoTS} scheme to the Han-Kobayashi scheme. In fact, we will give an example later on where with the proposed {\it discrete$\to$common~map} we obtain a $\loglog$ gap, but with a {\it discrete$\to$private~map} we obtain a constant gap.
Although finding the smallest possible gap in each regime would be desirable, here for sake of simplicity we consistently use the {\it discrete$\to$common~map}.

With the inner and outer bounds defined, as well as the `rate splits', we are ready to determine an optimal (to within a gap) choice of parameters for the mixed inputs in the weak interference regime.  
Next, we will focus on the regime in~\eqref{eq: condition for Weak1} and the regime in~\eqref{eq: condition for Weak2} separately
and for each regime
we will match each point on the closure of the outer bound in~\eqref{eq:all-corner-points-in-weak} with an achievable region as in~\eqref{eq: rates mixed inputs BEFORE Union}.

\subsection{Moderately Weak Interference, subregime Weak1} 
\label{sec:weak:WeakType1}

The regime of interest here is the subset of $\inr \leq \snr \leq \inr(1+\inr)$ for which~\eqref{eq: condition for Weak1} holds.
For convenience, we analyze the regime $\inr \leq \snr \leq 1+\inr$ in Appendix~\ref{sec:gaps s=i TDMA} and  
focus next on the subset of $(1+\inr) \leq \snr \leq \inr(1+\inr)$ for which~\eqref{eq: condition for Weak1} holds.
The condition $1+\inr \leq \snr$ allows us to state $\frac{1+\snr }{1+\inr+\frac{\snr}{1+\inr}} \geq 1$ in the following.

\paragraph*{Outer Bound Corner Points and Rate Splits}

Whenever the condition in~\eqref{eq: condition for Weak1} holds, the outer bound in~\eqref{eq:R upper classical IC} is given by all the constraints in~\eqref{eq:R upper classical IC} except for the one in~\eqref{eq:R upper classical IC etw} -- in the symmetric case the constraints in~\eqref{eq:R upper classical IC kra1} and~\eqref{eq:R upper classical IC kra2} are the same. 
\begin{subequations}
The corner points for the outer bound region in~\eqref{eq:outer bound in weak ALL} are thus
\begin{align}
\text{eq.\eqref{eq:R upper classical IC cuset r1}}=\text{eq.\eqref{eq:out:2r1r2}} \Rightarrow 
(R_{1A},R_{2A})&= \left(
\mug(\snr),
\label{eq:outer bound in weak1 R1A}
    \right.\\&\quad \left.
%%%\mug(\snr+\inr)+\mug \left(\inr+\frac{\snr}{1+\inr} \right)-\mug(\inr)-\mug(\snr)
\mug\left(\frac{\snr}{1+\inr}\right)+\mug \left(\inr+\frac{\snr}{1+\inr} \right)-\mug(\snr)
\right);
\label{eq:outer bound in weak1 R2A}
\\
\text{eq.\eqref{eq:out:2r1r2}}=\text{eq.\eqref{eq:R upper classical IC kra1}} \Rightarrow  
(R_{1B},R_{2B})&= \left(
\mug \left(\inr+\frac{\snr}{1+\inr} \right), 
\label{eq:outer bound in weak1 R1B}
    \right.\\&\quad \left.
\mug(\snr)+\mug \left(\frac{\snr}{1+\inr}\right)-\mug \left(\inr+\frac{\snr}{1+\inr} \right)
\right);
\label{eq:outer bound in weak1 R2B}
\\
\text{eq.\eqref{eq:out:2r1r2}}=\text{eq.\eqref{eq:R upper classical IC kra1}} \Rightarrow  
(R_{1C},R_{2C})&= \left(
\mug(\snr)+\mug \left(\frac{\snr}{1+\inr}\right)-\mug \left(\inr+\frac{\snr}{1+\inr} \right),
\label{eq:outer bound in weak1 R1C}
    \right.\\&\quad \left.
\mug \left(\inr+\frac{\snr}{1+\inr} \right)
\right);
\label{eq:outer bound in weak1 R2C}
\\
\text{eq.\eqref{eq:R upper classical IC cuset r2}}=\text{eq.\eqref{eq:out:r12r2}} \Rightarrow 
(R_{1D},R_{2D})&= \left(
%%%\mug(\snr+\inr)+\mug \left(\inr+\frac{\snr}{1+\inr} \right)-\mug(\inr)-\mug(\snr),
\mug\left(\frac{\snr}{1+\inr}\right)+\mug \left(\inr+\frac{\snr}{1+\inr} \right)-\mug(\snr),
\label{eq:outer bound in weak1 R1D}
    \right.\\&\quad \left.
\mug(\snr)
\right).
\label{eq:outer bound in weak1 R2D}
\end{align}
\label{eq:outer bound in weak1 ALL CORNERS}
\end{subequations}

As explained before, inspired by the proposed {\it discrete$\to$common~map}, we choose to `split' the rates as:
\begin{enumerate}

\item
for the sum-rate face / region $\mathcal{R}_{R_1+R_2}$: %in~\eqref{eq:outer bound in weak1 sumrate}
we set $R_{1,p}=R_{2,p} \approxeq \mug \left( \frac{\snr}{1+\inr}\right)$.
% where $R_{u,p}$ is the `private' rate for user $u\in[1:2]$,

\item
for the other dominant face / region $\mathcal{R}_{2R_1+R_2}$: %in~\eqref{eq:outer bound in weak1 otherrate}
we set $R_{1p} \approxeq \mug \left( \frac{\snr}{1+\inr}\right)$ and $R_{2p} \approxeq \frac{R_2}{2}$;

\item
we will not explicitly consider the remaining dominant face / region $\mathcal{R}_{R_1+2R_2}$ %in~\eqref{eq:outer bound in weak1 otherfaceswapp}
because a gap result can be obtained by proceeding as for $\mathcal{R}_{2R_1+R_2}$ but with the role of the users swapped.

\end{enumerate}

% ============================
\paragraph*{Outer Bound $\mathcal{R}_{R_1+R_2}$}

With the corner point expressions in~\eqref{eq:outer bound in weak1 ALL CORNERS}
we write the outer bound sum-rate face in~\eqref{eq:outer bound in weak sumrate}
as
\begin{align}
&\mathcal{R}_{R_1+R_2}^{(\text{\ref{sec:weak:WeakType1}})}
= \bigcup_{t \in [0,1]} 
\left \{ \begin{array}{l} 
R_1 \leq \frac{t}{2} \log \left( \frac{1+\inr +\frac{\snr}{1+\inr}}{1+\frac{\snr}{1+\inr}}\right) 
    +\frac{1-t}{2}\log \left( \frac{1+\snr }{1+\inr+\frac{\snr}{1+\inr}} \right)
    +\frac{1}{2}\log \left(1+\frac{\snr}{1+\inr}\right) \\
 \qquad =: \mug(\snr_{1,a,t})+\mug \left( \frac{\snr}{1+\inr}\right)
\\ 
R_2 \leq \frac{1-t}{2} \log \left( \frac{1+\inr +\frac{\snr}{1+\inr}}{1+\frac{\snr}{1+\inr}}\right)
   +\frac{t}{2}\log \left( \frac{1+\snr}{1+\inr+\frac{\snr}{1+\inr}} \right)
   +\frac{1}{2}\log \left(1+\frac{\snr}{1+\inr}\right)\\
 \qquad =: \mug(\snr_{1,b,t})+\mug \left( \frac{\snr}{1+\inr}\right) 
\\
\end{array} \right\}.
\label{eq:outer bound in weak1 sumrate}
\end{align}

\paragraph*{Inner Bound for $\mathcal{R}_{R_1+R_2}$}

In order to approximately achieve the points in~\eqref{eq:outer bound in weak1 sumrate}, 
we pick 
\begin{subequations}
\begin{align}
N_1 &= \points\left(\snr_{1,a,t}\right), \
\snr_{1,a,t} := \left( \frac{1+\inr +\frac{\snr}{1+\inr}}{1+\frac{\snr}{1+\inr}}\right)^{t}\left( \frac{1+\snr }{1+\inr+\frac{\snr}{1+\inr}} \right)^{1-t}-1,
\label{eq:achregion for par:moderate choiceN1}
\\
N_2 &= \points\left(\snr_{1,b,t}\right), \
\snr_{1,b,t} := \left( \frac{1+\inr +\frac{\snr}{1+\inr}}{1+\frac{\snr}{1+\inr}}\right)^{1-t}\left( \frac{1+\snr}{1+\inr+\frac{\snr}{1+\inr}} \right)^{t}-1,
\label{eq:achregion for par:moderate choiceN2}
\\
\delta_1&=\frac{1}{1+\inr},
\label{eq:achregion for par:moderate choicedelta1}
\\
\delta_2&=\frac{1}{1+\inr}.
\label{eq:achregion for par:moderate choicedelta2}
\end{align}
\label{eq:achregion for par:moderate choiceALL: R1+R2sumrate}
\end{subequations}

\paragraph*{Gap for $\mathcal{R}_{R_1+R_2}$}

The gap between the outer bound region in~\eqref{eq:outer bound in weak1 sumrate} and the achievable rate region in~\eqref{eq: rates mixed inputs BEFORE Union} with the parameters as in~\eqref{eq:achregion for par:moderate choiceALL: R1+R2sumrate} is 
%\begin{subequations}
\begin{align*} 
\Delta_{R_1}
  &=
  \mug(\snr_{1,a,t})
+ \mug \left(\frac{\snr}{1+\inr} \right)
- \log(\points(\snr_{1,a,t}))
- \mug \left(\frac{\snr}{1+\inr} \right)
+ \Delta_{\eqref{eq: rates mixed inputs BEFORE Union}} 
\\&  \leq \log(2)+\Delta_{\eqref{eq: rates mixed inputs BEFORE Union}},
\end{align*}
where the term $\log(2)$ is the ``integrality gap'' $\log(\points(x))+\log(2) \geq \mug(x)$;
similarly, we have 
\begin{align*}
\Delta_{R_2} \leq \log(2)+\Delta_{\eqref{eq: rates mixed inputs BEFORE Union}}.
\end{align*}

We are thus left with bounding $\Delta_{\eqref{eq: rates mixed inputs BEFORE Union}}$ in~\eqref{eq: rates mixed inputs BEFORE Union}, which is related to the minimum distance of the received constellations ${S}_1$ and ${S}_2$ defined in~\eqref{eq:GICmixedinputinnerBound}. 
In Appendix~\ref{sec:gaps:WeakType1:R1+R2 face} we show that 
\begin{align}
\min_{i \in[1:2]}
\frac{d_{\min ({S}_i)}^2}{12}
&\geq \kappa_{\gamma,N_1,N_2}^2 \cdot \frac{3}{8},
\label{eq:sec:gaps:WeakType1:R1+R2 face}
\end{align}
where $\kappa_{\gamma,N_1,N_2}$ is given in~\eqref{eq:achregion for par:moderate kappakappa definitiontobeusedlater}, and
$\max(N_1^2,N_2^2) -1\leq \inr = \min(\snr,\inr)$.
With this, the gap for this face is bounded by
\begin{align}
{\gap}_{\eqref{eq:gap weak1 sumrate}} 
&\leq \max(\Delta_{R_1},\Delta_{R_2}) = \log(2)+\Delta_{\eqref{eq: rates mixed inputs BEFORE Union}}
\notag\\
&\leq \frac{1}{2}\log\left(\frac{4\pi\eu}{3} \right)+\frac{1}{2}\log\left(1+\frac{8}{3}\cdot \frac{1}{\kappa_{\gamma,N_1,N_2}^2}\right)   
\notag\\
&\leq \frac{1}{2}\log\left(\frac{4\pi\eu}{3} \right)+\frac{1}{2}\log\left(1+\frac{32}{3}\cdot \frac{(1+1/2\ln(1+\min(\snr,\inr)))^2}{\gamma^2}\right).  
\label{eq:gap weak1 sumrate}
\end{align}

% ============================
\paragraph*{Outer Bound $\mathcal{R}_{2R_1+R_2}$}

With the corner point expressions in~\eqref{eq:outer bound in weak1 ALL CORNERS}
we write the outer bound in~\eqref{eq:outer bound in weak otherrate}
as
\begin{align}
&\mathcal{R}_{2R_1+R_2}^{(\text{\ref{sec:weak:WeakType1}})}
= \bigcup_{t \in [0,1]} \left \{ \begin{array}{l} 
R_1 \leq \frac{t}{2}\log \left( \frac{1+\inr+\frac{\snr}{1+\inr}}{1+\frac{\snr}{1+\inr}}\right)
   +\frac{1-t}{2}\log \left( \frac{1+\snr}{1+\frac{\snr}{1+\inr}} \right)
   +\frac{1}{2}\log \left(1+\frac{\snr}{1+\inr}\right)\\
 \qquad =: \mug(\snr_{2,a,t}) +\mug \left( \frac{\snr}{1+\inr}\right)\\
R_2 \leq  \frac{t}{2} \log \left( \frac{1+\inr+\snr}{1+\inr+\frac{\snr}{1+\inr}}\cdot \frac{1+\snr}{1+\inr}\right)
  + (1-t) c \\
 \qquad =: \mug(\snr_{2,b,t})+ \frac{t}{2}\log \left(\frac{1+\snr}{1+\inr}\right)+(1-t) c
 \end{array} \right\}, %R1+2R2
\label{eq:outer bound in weak1 otherrate}
\end{align}
where $(1-t) c\leq c \leq \log(2)$, where the parameter $c$ is defined in~\eqref{eq:out:evenmoreweired}.

\paragraph*{Inner Bound for $\mathcal{R}_{2R_1+R_2}$}

In order to approximately achieve the points in~\eqref{eq:outer bound in weak1 otherrate} we pick
\begin{subequations}
\begin{align}
N_1    &= \points\left(\snr_{2,a,t}\right),  \
\snr_{2,a,t} := \left(  \frac{1+\inr+\frac{\snr}{1+\inr}}{1+\frac{\snr}{1+\inr}} \right)^{t}\left( \frac{1+\snr}{1+\frac{\snr}{1+\inr}}  \right)^{1-t} -1,
\label{eq:achregion for par:moderate choiceN1: 2R1+R2}
\\
N_2 &= \points\left(\snr_{2,b,t}\right), \
\snr_{2,b,t} := \left(  \frac{1+\inr+\snr}{1+\inr+\frac{\snr}{1+\inr}}\right)^{t}-1,
\label{eq:achregion for par:moderate choiceN2: 2R1+R2}
\\
\delta_1&=\frac{1}{1+\inr},\\
\delta_2&: \ \mug\left(\snr \delta_2\right)=\frac{t}{2}\log\left(\frac{1+\snr}{1+\inr}\right)
\Longleftrightarrow
\delta_2=\left( \left(\frac{1+\snr}{1+\inr} \right)^t-1 \right) \frac{1}{\snr}, 
\label{eq: choice of delta2: moderate: 2R1+R2}
\end{align}
\label{eq:achregion for par:moderate choiceALL: 2R1+R2}
\end{subequations}
where the power split $\delta_2$ in~\eqref{eq: choice of delta2: moderate: 2R1+R2} satisfies
\begin{align*}
\delta_2 \leq 
\frac{1-\inr/\snr}{1+\inr}
\leq \frac{1}{1+\inr},
\end{align*}
as required for the achievable rate region in~\eqref{eq: rates mixed inputs BEFORE Union}.

\paragraph*{Gap for $\mathcal{R}_{2R_1+R_2}$}

The gap between the outer bound region in~\eqref{eq:outer bound in weak1 otherrate} and the achievable rate region in~\eqref{eq: rates mixed inputs BEFORE Union} with the choice in~\eqref{eq:achregion for par:moderate choiceALL: 2R1+R2} is 
%\begin{subequations}
\begin{align*}
\Delta_{R_1}
&
= \mug \left(\snr_{2,a,t} \right)
+ \mug \left(\frac{\snr}{1+\inr}\right)
- \log \left(\points\left(\snr_{2,a,t}\right) \right)
- \mug \left(\frac{\snr}{1+\inr}\right)
+\Delta_{\eqref{eq: rates mixed inputs BEFORE Union}}
\\&\leq \log(2)+\Delta_{\eqref{eq: rates mixed inputs BEFORE Union}},
\end{align*}
and similarly
\begin{align*}
\Delta_{R_2}
&=  \mug \left(\snr_{2,b,t} \right)
+\frac{t}{2} \log \left(\frac{1+\snr}{1+\inr}\right)+ (1-t) c
-\log \left(\points\left(\snr_{2,b,t}\right) \right)
-\frac{t}{2} \log \left(\frac{1+\snr}{1+\inr}\right)
+\Delta_{\eqref{eq: rates mixed inputs BEFORE Union}}
\\&\leq \log(2)+\log(2)+\Delta_{\eqref{eq: rates mixed inputs BEFORE Union}},
\end{align*}
since $(1-t) c\leq c \leq \log(2)$, where the parameter $c$ is defined in~\eqref{eq:out:evenmoreweired}.

So we are left with bounding $\Delta_{\eqref{eq: rates mixed inputs BEFORE Union}}$ in~\eqref{eq: rates mixed inputs BEFORE Union}, which is related to the minimum distance of the received constellations ${S}_1$ and ${S}_2$ defined in~\eqref{eq:GICmixedinputinnerBound}.

In Appendix~\ref{sec:gaps:WeakType1:2R1+R2 face} we show that 
%\begin{subequations}
\begin{align}
\min_{i \in[1:2]}
\frac{d_{\min ({S}_i)}^2}{12}
&\geq \kappa_{\gamma,N_1,N_2}^2 \cdot \frac{4}{45}
\label{eq:sec:gaps:WeakType1:2R1+R2 face}
\end{align}
where $\kappa_{\gamma,N_1,N_2}$ is given in~\eqref{eq:achregion for par:moderate kappakappa definitiontobeusedlater},
and $\max(N_1^2,N_2^2) -1\leq \inr = \min(\snr,\inr)$.
With this, we finally  get that the gap for this face is bounded by
\begin{align}
{\gap}_{\eqref{eq:gap weak1 otherrate}} 
&\leq \max(\Delta_{R_1},\Delta_{R_2})
\notag\\&
\leq \frac{1}{2}\log\left(\frac{16\pi\eu}{3} \right)+\frac{1}{2}\log\left(1+\frac{45}{4}\cdot \frac{1}{\kappa_{\gamma,N_1,N_2}^2}\right)  
\notag\\
&\leq \frac{1}{2}\log\left(\frac{16\pi\eu}{3} \right)+\frac{1}{2}\log\left(1+45 \cdot \frac{(1+1/2\ln(1+\min(\snr,\inr)))^2}{\gamma^2}\right).  
\label{eq:gap weak1 otherrate}
\end{align}

% ============================
\paragraph*{Overall Gap for Weak1}
To conclude the proof for this sub-regime, the gap is the maximum between the gaps of the different faces and is given by
\begin{align}
{\gap}_{\eqref{eq: gap weak 1}} 
&\leq \max\left( {\gap}_{\eqref{eq:gap weak1 sumrate}}, {\gap}_{\eqref{eq:gap weak1 otherrate}} \right)
= {\gap}_{\eqref{eq:gap weak1 otherrate}}.
\label{eq: gap weak 1}
\end{align}

\subsection{Moderately Weak Interference, subregime Weak2} %as defined in~\eqref{eq: condition for Weak2}}
\label{sec: weak: weak type 2}
%Recall that the regime $\inr \leq \snr \leq 1+\inr$ was studied in Appendix~\ref{sec:gaps s=i TDMA}.
We focus here on the subset of $\inr \leq \snr \leq \inr(1+\inr)$ for which~\eqref{eq: condition for Weak2} holds.

\paragraph*{Outer Bound Corner Points and Rate Splits}

Under the condition in~\eqref{eq: condition for Weak2}, the outer bound in~\eqref{eq:R upper classical IC} is given by all the constraints except for the ones in~\eqref{eq:R upper classical IC kra1} and~\eqref{eq:R upper classical IC kra2}.

\begin{subequations}
The corner points are thus
\begin{align}
\text{eq.\eqref{eq:R upper classical IC cuset r1}}=\text{eq.\eqref{eq:out:2r1r2}} \Rightarrow 
(R_{1A},R_{2A})&= \left(
\mug(\snr),
\label{eq:outer bound in weak2 R1A}
    \right.\\&\quad \left.
%\mug(\snr+\inr)+\mug\left(\inr+\frac{\snr}{1+\inr}\right)-\mug(\inr)-\mug(\snr)
\mug\left(\frac{\snr}{1+\inr} \right)+\mug\left(\inr+\frac{\snr}{1+\inr}\right)-\mug(\snr)
\right);
\label{eq:outer bound in weak2 R2A}
\\
\text{eq.\eqref{eq:out:2r1r2}}=\text{eq.\eqref{eq:R upper classical IC etw}} \Rightarrow 
(R_{1B},R_{2B})&=\left(
\mug \left(\frac{\snr}{1+\inr} \right)+\mug(\snr)-\mug \left( \inr+\frac{\snr}{1+\inr}\right),
\label{eq:outer bound in weak2 R1B}
    \right.\\&\quad \left.
3 \mug \left(\inr+\frac{\snr}{1+\inr} \right)-\mug(\snr)- \mug \left(\frac{\snr}{1+\inr} \right)
\right);
\label{eq:outer bound in weak2 R2B}
\\
\text{eq.\eqref{eq:R upper classical IC etw}}=\text{eq.\eqref{eq:out:r12r2}} \Rightarrow 
(R_{1C},R_{2C})&=\left(
3 \mug \left(\inr+\frac{\snr}{1+\inr} \right)-\mug(\snr)- \mug \left(\frac{\snr}{1+\inr} \right),
\label{eq:outer bound in weak2 R1C}
    \right.\\&\quad \left.
\mug \left(\frac{\snr}{1+\inr} \right)+\mug(\snr)-\mug \left( \inr+\frac{\snr}{1+\inr}\right)
\right);
\label{eq:outer bound in weak2 R2C}
\\
\text{eq.\eqref{eq:R upper classical IC cuset r2}}=\text{eq.\eqref{eq:out:r12r2}} \Rightarrow 
(R_{1D},R_{2D})&= \left(
%\mug(\snr+\inr)+\mug\left(\inr+\frac{\snr}{1+\inr}\right)-\mug(\inr)-\mug(\snr)
\mug\left(\frac{\snr}{1+\inr} \right)+\mug\left(\inr+\frac{\snr}{1+\inr}\right)-\mug(\snr),
\label{eq:outer bound in weak2 R1D}
    \right.\\&\quad \left.
\mug(\snr)
\right).
\label{eq:outer bound in weak2 R2D}
\end{align}
\label{eq:outer bound in weak2 ALL CORNERS}
\end{subequations}

As explained before, inspired by the proposed {\it discrete$\to$common~map}, we choose to `split' the rates as:
\begin{enumerate}

\item
for the sum-rate face / region $\mathcal{R}_{R_1+R_2}$: %in~\eqref{eq:outer bound in weak2 sumrate}
we set
$R_{1,p}\approxeq \frac{R_1}{2}$ and $R_{2,p}\approxeq \frac{R_2}{2}$. 
% where $R_{u,p}$ is the `private' rate for user $u\in[1:2]$.

\item
for the other dominant face / region $\mathcal{R}_{2R_1+R_2}$: %in~\eqref{eq:outer bound in weak2 otherrate}
we set
$R_{1,p}\approxeq \mug \left( \frac{\snr}{1+\inr}\right)$ and $R_{2,p} \approxeq \frac{R_2}{2}$. 

\item
we will not explicitly consider the remaining dominant face / region $\mathcal{R}_{R_1+2R_2}$ %in~\eqref{eq:outer bound in weak1 otherfaceswapp}
because a gap result can be obtained by proceeding as for $\mathcal{R}_{2R_1+R_2}$ but with the role of the users swapped.

\end{enumerate}

% ============================
\paragraph*{Outer Bound $\mathcal{R}_{R_1+R_2}$}

With the corner point expressions in~\eqref{eq:outer bound in weak2 ALL CORNERS}
we write the outer bound sum-rate face in~\eqref{eq:outer bound in weak sumrate}
as
\begin{align}
 \mathcal{R}_{R_1+R_2}^{(\text{\ref{sec: weak: weak type 2}})} 
 &= \bigcup_{t\in [0,1]} \left \{ \begin{array}{l}
 R_1 \leq \frac{1-t}{2} \log \left( \frac{ \left(1+\frac{\snr}{1+\inr} \right)(1+\snr)}{1+\inr +\frac{\snr}{1+\inr}}\right)
           +\frac{t}{2} \log \left( \frac{ \left(1+\inr+\frac{\snr}{1+\inr} \right)^3}{ \left(1+\frac{\snr}{1+\inr} \right)(1+\snr)}\right)
 \\\qquad  =: 
  2\cdot \mug \left( \snr_{3,a,t} \right)
%  +\frac{1-t}{4} \log \left(  \frac{(1+\snr+\inr)(1+\snr)}{(1+\inr) \left(1+\inr +\frac{\snr}{1+\inr} \right)}\right)+\frac{t}{4} \log \left( \frac{ \left(1+\inr+\frac{\snr}{1+\inr} \right)^3(1+\inr)}{(1+\snr+\inr)(1+\snr)}\right)
\\
 R_1 \leq \frac{t}{2} \log \left( \frac{ \left(1+\frac{\snr}{1+\inr} \right)(1+\snr)}{1+\inr +\frac{\snr}{1+\inr}}\right)
       +\frac{1-t}{2} \log \left( \frac{ \left(1+\inr+\frac{\snr}{1+\inr} \right)^3}{ \left(1+\frac{\snr}{1+\inr} \right)(1+\snr)}\right)
 \\\qquad  =:
 2\cdot \mug \left( \snr_{3,b,t} \right)
 %+  \frac{1-t}{4} \log \left( \frac{ \left(1+\inr+\frac{\snr}{1+\inr} \right)^3(1+\inr)}{(1+\snr+\inr)(1+\snr)}\right)  +\frac{t}{4}  \log \left(  \frac{(1+\snr+\inr)(1+\snr)}{(1+\inr) \left(1+\inr +\frac{\snr}{1+\inr} \right)}\right)
\end{array}\right \}.
\label{eq:outer bound in weak2 sumrate}
\end{align}

\paragraph*{Inner Bound for $\mathcal{R}_{R_1+R_2}$}

In order to approximately achieve the points in $\mathcal{R}_{R_1+R_2}^{(\text{\ref{sec: weak: weak type 2}})}$ in~\eqref{eq:outer bound in weak2 sumrate} we pick
\begin{subequations}
\begin{align}
N_1 &= \points\left(\snr_{3,a,t}\right), \
\snr_{3,a,t} :=
\left( \frac{ \left(1+\frac{\snr}{1+\inr} \right)(1+\snr)}{1+\inr +\frac{\snr}{1+\inr}}\right)^{\frac{1-t}{2}}
\left( \frac{ \left(1+\inr+\frac{\snr}{1+\inr} \right)^3}{ \left(1+\frac{\snr}{1+\inr} \right)(1+\snr)}\right)^{\frac{t}{2}}
-1,
\label{eq:achregion for par:weak type 2 choiceN1 take1}
\\
N_2 &= \points\left(\snr_{3,b,t}\right), \
\snr_{3,b,t} := 
\left( \frac{ \left(1+\frac{\snr}{1+\inr} \right)(1+\snr)}{1+\inr +\frac{\snr}{1+\inr}}\right)^{\frac{t}{2}}
\left( \frac{ \left(1+\inr+\frac{\snr}{1+\inr} \right)^3}{ \left(1+\frac{\snr}{1+\inr} \right)(1+\snr)}\right)^{\frac{1-t}{2}}
-1,
\label{eq:achregion for par:weak type 2 choiceN2 take1}
\\
\delta_1 &:  \mug \left(\snr \delta_1\right)=\mug\left(\snr_{3,a,t}\right)
\Longleftrightarrow  
\delta_1 =\frac{\snr_{3,a,t}}{\snr}, 
\label{eq: choice delta1: weak 2 R1+R2 take1}
\\
\delta_2& : \mug \left( \snr \delta_2\right)=\mug\left(\snr_{3,b,t}\right)
\Longleftrightarrow
\delta_2 =\frac{\snr_{3,b,t}}{\snr},
\label{eq: choice delta2: weak 2 R1+R2 take1}
\end{align}
where
\begin{align*}
\max(\delta_1,\delta_2) = \frac{\max(\snr_{3,a,t},\snr_{3,b,t})}{\snr}  
\leq \frac{1}{1+\inr}, 
\end{align*}
as required for the achievable rate in~\eqref{eq: rates mixed inputs BEFORE Union};
the proof can be found in Appendix~\ref{sec:gaps:WeakType2:R1+R2 face}, eq.\eqref{eq: dsbound {sec:gaps:WeakType2:R1+R2 face}}.
\label{eq:parameters R1+R2 face weak type 2 take1}
\end{subequations}

\paragraph*{Gap for $\mathcal{R}_{R_1+R_2}$}

The gap between the outer bound region in~\eqref{eq:outer bound in weak2 sumrate} and the achievable rate in~\eqref{eq: rates mixed inputs BEFORE Union} with the parameters in~\eqref{eq:parameters R1+R2 face weak type 2} is 
\begin{align*}
\Delta_{R_1} &= 2 \mug \left( \snr_{3,a,t} \right)
  -\log\left(\points \left( \snr_{3,a,t}\right)\right)
  -\mug( \snr_{3,a,t})+\Delta_{\eqref{eq: rates mixed inputs BEFORE Union}}\\
  & \leq \log(2)+\Delta_{\eqref{eq: rates mixed inputs BEFORE Union}},
\end{align*}
and similarly
\begin{align*}
\Delta_{R_2} 
  & \leq \log(2)+\Delta_{\eqref{eq: rates mixed inputs BEFORE Union}}.
\end{align*}

We are then left with bounding $\Delta_{\eqref{eq: rates mixed inputs BEFORE Union}}$, which depends on minimum distances of the received sum-set constellations.

In Appendix~\ref{sec:gaps:WeakType2:R1+R2 face} we show 
\begin{align}
\min_{i\in[1:2]} \frac{d_{\min ({S}_i)}^2}{12} \geq \kappa_{\gamma,N_1,N_2}^2 \cdot \frac{1}{24},
\label{eq: dmin WeakType2:R1+R2 face take1}
\end{align}
where $\kappa_{\gamma,N_1,N_2}$ is given in~\eqref{eq:achregion for par:moderate kappakappa definitiontobeusedlater},
and $\max(N_1^2,N_2^2) -1\leq \inr = \min(\snr,\inr)$.
With this, we finally  get that the gap for this face is bounded by
\begin{align}
{\gap}_{\eqref{eq:sec:gaps:WeakType2:R1+R2 face take1}}
        &\leq \max(\Delta_{R_1},\Delta_{R_2}) = \log(2) +\Delta_{\eqref{eq: rates mixed inputs BEFORE Union}}
\notag\\&\leq \frac{1}{2}\log\left(\frac{4\pi\eu}{3} \right)+\frac{1}{2}\log\left(1+96 \cdot \frac{(1+1/2\ln(1+\min(\snr,\inr)))^2}{\gamma^2}\right).
\label{eq:sec:gaps:WeakType2:R1+R2 face take1}
\end{align}

% ============================
\paragraph*{Outer Bound $\mathcal{R}_{2R_1+R_2}$}

With the corner point expressions in~\eqref{eq:outer bound in weak2 ALL CORNERS}
we write the outer bound in~\eqref{eq:outer bound in weak otherrate}
as
\begin{align}
 \mathcal{R}_{2R_1+R_2}^{(\text{\ref{sec: weak: weak type 2}})}
&= \bigcup_{t\in [0,1]}  \left \{ \begin{array}{l}
 R_1 \leq \frac{1-t}{2} \log \left(  \frac{\left(1+\frac{\snr}{1+\inr} \right)(1+\snr)}{1+\inr +\frac{\snr}{1+\inr}}\right)+\frac{t}{2} \log \left(1+\snr\right) \\
 \qquad  =:\mug \left( \snr_{4,a,t}\right)+ \mug\left(\frac{\snr}{1+\inr}\right)
 \\
 R_2 \leq  \frac{1-t}{2}\log \left(\frac{ \left(1+\inr+\frac{\snr}{1+\inr} \right)^3}{1+\snr} % \right) +  \frac{1-t}{2} \log \left(   \left(1+\inr+\frac{\snr}{1+\inr}\right)
 \frac{(1+\inr) }{1+\inr+\snr} \right)+t c \\
 \qquad  =:\mug(\snr_{4,b,t}) + 
 \frac{1-t}{2} \log \left( \frac{ \left (1+\inr+\frac{\snr}{1+\inr} \right)(1+\inr) }{1+\inr+\snr} \right)+t c
\end{array}\right\},
\label{eq:outer bound in weak2 otherrate}
\end{align}
where $tc\leq c \leq \log(2)$, where the parameter $c$ is defined in~\eqref{eq:out:evenmoreweired}, and $0\leq t\leq 1$.

\paragraph*{Inner Bound for $\mathcal{R}_{2R_1+R_2}$}

In order to approximately achieve the points in $\mathcal{R}_{2R_1+R_2}$ in~\eqref{eq:outer bound in weak2 otherrate} we pick
\begin{subequations}
\begin{align}
N_1 &= \points\left(\snr_{4,a,t}\right), \
\snr_{4,a,t} :=\frac{1+\snr}{\left(1+\frac{\snr}{1+\inr}\right)^{t} \left( 1+\inr+\frac{\snr}{1+\inr}\right)^{1-t}}-1,
\label{eq:parameters 2R1+R2 face weak type 2: first choice N1}
\\
N_2 &= \points\left(\snr_{4,b,t}\right), \
\snr_{4,b,t} := \left( \frac{(1+\inr +\frac{\snr}{1+\inr})^{2} }{1+\snr} \right)^{1-t}-1,
\label{eq:parameters 2R1+R2 face weak type 2: first choice N2}
\\
\delta_1&:=\frac{1}{1+\inr},
\label{eq:parameters 2R1+R2 face weak type 2: first choice delta1}
\\
\delta_2&: \mug\left(\snr \delta_2\right) =\frac{1-t}{2} \log \left( \frac{1+\inr+\frac{\snr}{1+\inr}}{1+\frac{\snr}{1+\inr} }\right) 
%\notag\\&
\Longleftrightarrow
\delta_2= \left( \left( \frac{1+\inr+\frac{\snr}{1+\inr}}{1+\frac{\snr}{1+\inr} }\right)^{1-t}  -1\right)\frac{1}{\snr},
% \leq \frac{\snr_{4,b,t}}{\snr},
\label{eq:parameters 2R1+R2 face weak type 2: first choice delta2}
\end{align}
\label{eq:parameters 2R1+R2 face weak type 2: first choice ALL}
\end{subequations}
where in Appendix~\ref{sec:gaps:WeakType2:2R1+R2 face}, eq.\eqref{eq: d2bound {sec:gaps:WeakType2:2R1+R2 face}}, we show that
\begin{align*}
\delta_2 \leq \delta_1 = \frac{1}{1+\inr},
\end{align*}
as required for the achievable region in~\eqref{eq: rates mixed inputs BEFORE Union}.

\paragraph*{Gap for $\mathcal{R}_{2R_1+R_2}$}

The gap between the outer bound in~\eqref{eq:outer bound in weak otherrate} and achievable region in~\eqref{eq: rates mixed inputs BEFORE Union} with the parameters in~\eqref{eq:parameters 2R1+R2 face weak type 2: first choice ALL} is
\begin{align*}
\Delta_{R_1}&= \mug(\snr_{4,a,t}) + \mug \left(\frac{\snr}{1+\inr}\right)- \log (\points(\snr_{4,a,t}))-\mug(\snr \delta_1)+\Delta_{\eqref{eq: rates mixed inputs BEFORE Union}}
\\&\leq \log(2)+\Delta_{\eqref{eq: rates mixed inputs BEFORE Union}},
\end{align*}
and similarly we have that 
\begin{align*}
\Delta_{R_2}&= \mug(\snr_{4,b,t})+ \frac{1-t}{2} \log \left( \frac{ \left (1+\inr+\frac{\snr}{1+\inr} \right)(1+\inr) }{1+\inr+\snr} \right)+t c- \log (\points(\snr_{4,b,t}))-\mug(\snr \delta_2)+\Delta_{\eqref{eq: rates mixed inputs BEFORE Union}}
%& \leq \log(2) +tc +\Delta_{\eqref{eq: rates mixed inputs BEFORE Union}} \\
\\&\leq  \log(2) +\log(2)+\Delta_{\eqref{eq: rates mixed inputs BEFORE Union}},
\end{align*}
since $tc\leq c \leq \log(2)$, where the parameter $c$ is defined in~\eqref{eq:out:evenmoreweired}, and $0\leq t\leq 1$.

So, we are left with bounding $\Delta_{\eqref{eq: rates mixed inputs BEFORE Union}}$ which is related to the minimum distances of the sum-set constellations.  In Appendix~\ref{sec:gaps:WeakType2:2R1+R2 face} we show that
\begin{align}
\min_{i \in[1:2]}
\frac{d_{\min ({S}_i)}^2}{12}
&\geq \kappa_{\gamma,N_1,N_2}^2 \cdot   \frac{1}{8}, 
\label{eq: min  WeakType2:2R1+R2 face with loglog gap}
\end{align}
where $\kappa_{\gamma,N_1,N_2}$ is given in~\eqref{eq:achregion for par:moderate kappakappa definitiontobeusedlater}, and
that $\max(N_1^2,N_2^2)-1\leq \inr$;
with this, we finally get that the gap for this face is bounded by
\begin{align}
{\gap}_{\eqref{eq:gap weak2 otherrate take1}}
        &\leq \max(\Delta_{R_1},\Delta_{R_2}) = 2\log(2) +\Delta_{\eqref{eq: rates mixed inputs BEFORE Union}} 
\notag\\&\leq \frac{1}{2}\log\left(\frac{16\pi\eu}{3} \right)+\frac{1}{2}\log\left(1+32 \cdot \frac{(1+1/2\ln(1+\min(\snr,\inr)))^2}{\gamma^2}\right).
\label{eq:gap weak2 otherrate take1}
\end{align}

% ============================

\paragraph*{Overall Gap for Weak2}
To conclude the proof for this sub-regime, the gap is the maximum between the gaps of the different faces and is given by
\begin{align}
{\gap}_{\eqref{eq: gap weak 2 take 1}}
=\min(
{\gap}_{\eqref{eq:sec:gaps:WeakType2:R1+R2 face take1}}, 
{\gap}_{\eqref{eq:gap weak2 otherrate take1}}
) = {\gap}_{\eqref{eq:gap weak2 otherrate take1}}.
\label{eq: gap weak 2 take 1}
\end{align}

\paragraph*{Another Overall Gap for Weak2}
The choice of the mixed input parameters according to the {\it discrete$\to$common~map} 
in~\eqref{eq:parameters R1+R2 face weak type 2 take1} and
in~\eqref{eq:sec:gaps:WeakType2:R1+R2 face take1} 
led to the $\loglog$ gap  in~\eqref{eq: gap weak 2 take 1}.
This is so because we used Proposition~\ref{prop:card:dmin:measurebound} to bound the minimum distance.
% term in~\eqref{eq: rates mixed inputs BEFORE Union}.
A interesting question is whether Proposition~\ref{prop:combNOOUTAGE} could be used, possibly with a different choice of  mixed input parameters.

With a gDoF-type analysis, one can show that it is possible to verify the condition in Proposition~\ref{prop:combNOOUTAGE} 
        with the proposed choice of parameters in~\eqref{eq:parameters R1+R2 face weak type 2 take1}
but not with the input parameters in~\eqref{eq:sec:gaps:WeakType2:R1+R2 face take1}.
So, in this regime we are motivated to look at the {\it discrete$\to$private~map} as we hope to get a constant gap result for the whole region.
In Appendix~\ref{app:sec:CONSTANT GAP FOR WEAK 2} we show that in this regime it is possible to use Proposition~\ref{prop:combNOOUTAGE} and the {\it discrete$\to$private~map} to get a constant gap, namely, 
\begin{align}
{\gap}_{\eqref{eq: gap weak 2}} 
        &\leq \frac{1}{2}\log \left(\frac{608 \ \pi \eu}{27}\right) \approx 3.79. 
\label{eq: gap weak 2}
\end{align}

\subsection{Very weak interference, i.e., $\inr(1+\inr)\leq \snr$}
\label{par:veryweak}
In this regime the capacity of the classical G-IC is achieved to within a constant gap by Gaussian inputs, treating interference as noise and power control. This strategy is compatible without the {TINnoTS} scheme (i.e., set $N_1=N_2=1$ and vary $\delta_1$ and $\delta_2$), so the gap of 
\begin{align}
{\gap}_{\eqref{eq: gap very weak}}  \leq 1/2~\text{bit},
\label{eq: gap very weak}
\end{align}
as shown in~\cite{etkin_tse_wang} holds.

% -------------------------------------------------------------------------------------------------- %
This concludes the proof of Theorem~\ref{thm:Cap:gap sym}.
\end{IEEEproof}

\section{Gap for some Asymmetric Channels}\label{sec:cap within gap asymmetric}

In this Section we generalize the gap result of Theorem~\ref{thm:Cap:gap sym} to some general asymmetric settings.
%The main result of this section is
\begin{thm}
\label{thm:Cap:gap some asymmetric}
%We shall prove that 
For the general G-IC, except for the regime
\begin{subequations}
\begin{align}
     \frac{|h_{22}|^2}{1+|h_{21}|^2} < |h_{12}|^2 < \frac{|h_{22}|^2}{1+|h_{21}|^2}(1+|h_{11}|^2), \\
     \frac{|h_{11}|^2}{1+|h_{12}|^2} < |h_{21}|^2 < \frac{|h_{11}|^2}{1+|h_{12}|^2}(1+|h_{22}|^2),
\end{align}
\label{eq:def:open gen}
\end{subequations}
akin to the moderately weak interference regime for the symmetric setting,
the {TINnoTS} achievable region %$\mathcal{R}_{\text{in}}$ in~\eqref{eq: rates mixed inputs WITH Union}
and the outer bound in Proposition~\ref{prop:ETWouter} %$\mathcal{R}_{\text{out}}$  in~\eqref{eq:R upper classical IC}
are to within an additive gap that is either constant or of the order $O\left(\log \ln \left(\max(|h_{11}|^2,|h_{22}|^2)\right)\right)$.
% where here $\snr:=\max(|h_{11}|^2,|h_{22}|^2)$.
\end{thm}

\begin{rem}[Why is the regime in~\eqref{eq:def:open gen} excluded?]
\label{rem:whatNOTinasymcase}
The regime identified in~\eqref{eq:def:open gen} involves numerous special cases,
whose analysis gets very tedious. % and is outside of the scope of this paper.
We do however strongly believe that our gap result generalizes to this regime as well,
by using similar arguments to those developed so far.
We note that the analysis in the rest of this section for
the general asymmetric setting (which is characterized by four channel parameters)
is restricted to those cases where it suffices to consider at most one rate split
(thus reducing the number of parameters to be optimize for the mixed inputs)
and for which the approximately optimal rate region does not require bounds on $2R_1+R_2$ or $R_1+2R_2$
(thus reducing the achievability to the sum-capacity dominant face only).
\end{rem}

\begin{IEEEproof}
We shall treat different regimes separately in the rest of the section.

\subsection{Very Strong Interference}
\label{rem:very strong asym}
In the general asymmetric case, the very strong interference regime is the regime in which 
a receiver can decode the interfering message while treating its intended signal as noise 
at a higher rate than the intended receiver in the absence of interference;
this is the case when the channel gains satisfy~\cite{sato_strong}
\begin{subequations}
\begin{align}
&|h_{11}|^2(1+|h_{22}|^2) \leq |h_{21}|^2, 
\label{eq:def:verystrongRx1 gen}
\\
&|h_{22}|^2(1+|h_{11}|^2) \leq |h_{12}|^2.
\label{eq:def:verystrongRx2 gen}
\end{align}
\label{eq:def:verystrong gen}
\end{subequations}

\paragraph*{Outer Bound}
The capacity region of the classical G-IC in very strong interference coincides with that of two interference-free point-to-point links given by
\begin{align}
\mathcal{R}_{\text{out}}^{(\text{\ref{rem:very strong asym}})} %(\snr,\inr)
&= \left \{ %(R_1,R_2):
\begin{array}{l}
0\leq R_1 \leq \mug \left(|h_{11}|^2\right)
\\
0\leq R_2 \leq \mug \left(|h_{22}|^2\right)
\end{array} \right\}.
\label{eq:def:verystrongOuter gen}
\end{align}

\paragraph*{Inner Bound}
The outer bound in~\eqref{eq:def:verystrongOuter gen} can be matched to within a constant gap by our {TINnoTS} scheme by choosing, similarly to the symmetric case discussed in Section~\ref{par:verystrong}, the mixed inputs in~\eqref{eq:GICmixedinput} with
\begin{subequations}
\begin{align}
&N_1=\points\left( \beta |h_{11}|^2\right) : N_1^2 - 1 \leq  \beta |h_{11}|^2,
\label{eq:def:verystrongChoiseN1 gen}
\\
&N_2=\points\left( \beta |h_{22}|^2\right) : N_2^2 - 1 \leq  \beta |h_{22}|^2,
\label{eq:def:verystrongChoiseN2 gen}
\\
&\delta_1=0,
\\
&\delta_2=0,
\end{align}
for some $ \beta \leq 1.$
\label{eq:def:verystrongChoiseALL gen}
\end{subequations}
The reason for the factor $ \beta $ in~\eqref{eq:def:verystrongChoiseALL gen} will be clear shortly (in Appendix~\ref{app:sec:CONSTANT GAP FOR WEAK 2} we use $ \beta =3/4$ for similar reasons; we could have used here  $ \beta =3/4$ as well, but we will find next a value that gives a smaller gap).
 
We next show that Proposition~\ref{prop:combNOOUTAGE} is applicable for the choice of mixed input parameters as in~\eqref{eq:def:verystrongChoiseALL gen}. In particular, we aim to show that 
\begin{subequations}
\begin{align}
&N_1 |h_{11}| d_{\min(X_1)} \leq  |h_{12}| d_{\min(X_2)} \ \text{(for the received sum-set constellation at receiver~1)}
\\ 
&N_2 |h_{22}| d_{\min(X_2)} \leq  |h_{21}| d_{\min(X_1)} \ \text{(for the received sum-set constellation at receiver~2)},
\end{align}
\label{eq:very strong asym dminP2 starting point}
\end{subequations}
or equivalently that
\begin{subequations}
\begin{align}
&\frac{N_1^2}{N_1^2-1} \cdot \frac{|h_{11}|^2}{1+|h_{11}|^2} \cdot  \frac{N_2^2-1}{|h_{22}|^2} \leq  \frac{|h_{12}|^2}{|h_{22}|^2(1+|h_{11}|^2)},
\\ 
&\frac{N_2^2}{N_2^2-1} \cdot \frac{|h_{22}|^2}{1+|h_{22}|^2} \cdot  \frac{N_1^2-1}{|h_{11}|^2} \leq  \frac{|h_{21}|^2}{|h_{11}|^2(1+|h_{22}|^2)}.
\end{align}
\label{eq:very strong asym dminP2 equivalent}
\end{subequations}
The condition in~\eqref{eq:very strong asym dminP2 equivalent} is verified, given the channel gain relationship in~\eqref{eq:def:verystrong gen}, if
\begin{subequations}
\begin{align}
&\frac{N_1^2}{N_1^2-1} \cdot \frac{|h_{11}|^2}{1+|h_{11}|^2} \cdot  \frac{N_2^2-1}{|h_{22}|^2} \leq 1,
\\ 
&\frac{N_2^2}{N_2^2-1} \cdot \frac{|h_{22}|^2}{1+|h_{22}|^2} \cdot  \frac{N_1^2-1}{|h_{11}|^2} \leq 1.
\end{align}
\label{eq:very strong asym dminP2 sufficinet}
\end{subequations}
It can be easily seen that $ \beta =0.8277$ satisfies~\eqref{eq:very strong asym dminP2 sufficinet} whenever $2\leq \min(N_1,N_2)$.
For the found $ \beta $ we therefore have that the received constellations have $|{S}_1| =|{S}_2| = N_1 N_2$  equally likely points
and minimum distance
\begin{align}
&\min_{i\in[1:2]}  \frac{ d_{\min({S}_{i})}^2 }{12}  
= \min\left(
\frac{|h_{11}|^2}{N_1^2-1} ,  \frac{|h_{22}|^2}{N_2^2-1} \right) \geq \frac{1}{ \beta }.
\end{align}

Thus, by following similar steps as in Section~\ref{par:verystrong}, the achievable region becomes 
\begin{subequations}
\begin{align}
\mathcal{R}_{\text{in}}^{(\text{\ref{rem:very strong asym}})} %(\snr,\inr)
&= \left \{ %(R_1,R_2):
\begin{array}{l}
0\leq R_1 \leq r_1
\\
0\leq R_2 \leq r_2
\end{array} \right\}
\ \text{such that}
\\
r_1
  &\geq 
   \mud \left({S}_1 \right)
- \min \left( \log(N_2), \mug \left(|h_{12}|^2 \right) \right)
\geq \log(N_1) - {\Delta}_{\eqref{eq:very strong asym achreg}},
\\
r_2
  &\geq  \mud \left({S}_2 \right)
- \min \left( \log(N_1), \mug \left(|h_{21}|^2 \right) \right)
\geq \log(N_2) - {\Delta}_{\eqref{eq:very strong asym achreg}},
\\
{\Delta}_{\eqref{eq:very strong asym achreg}} 
 &\leq \frac{1}{2}\log\left(\frac{2\pi\eu}{12}\right)
 +    \frac{1}{2}\log\left(1+\frac{12}{\min_{i\in[1:2]}d_{\min({S}_i)}^2}\right)
 \leq \frac{1}{2}\log\left(\frac{2\pi\eu}{12}( \beta +1)\right).
\end{align}
\label{eq:very strong asym achreg}
\end{subequations}

\paragraph*{Gap} 
We can easily see, by comparing the inner bound in~\eqref{eq:very strong asym achreg}
with the outer bound in~\eqref{eq:def:verystrongOuter gen},
that for the general asymmetric G-IC in very strong interference the {TINnoTS} region is optimal to within
\begin{align}
{\gap}_{\eqref{eq:very strong asym gap}}
&\leq
{\Delta}_{\eqref{eq:very strong asym achreg}}
+\log(2)+\frac{1}{2}\log\left(\frac{1}{ \beta }\right)
\notag\\
&\leq \frac{1}{2}\log\left(\frac{2\pi\eu}{3}\ \frac{1+ \beta }{ \beta }\right) 
\stackrel{ \beta =0.8277}\approx 1.8260~\text{bits},
\label{eq:very strong asym gap}
\end{align}
where the term $\log(2)$ is the integrality gap and the term $\frac{1}{2}\log\left(\frac{1}{ \beta }\right)$ because of the reduced number of points in~\eqref{eq:def:verystrongChoiseALL gen}.

\subsection{Strong (but not Very Strong) Interference}
\label{rem:strong asym}
\begin{subequations}
For the general asymmetric case,  the strong interference regime is defined as
\begin{align}
|h_{21}|^2 &\geq |h_{11}|^2, %\geq \frac{|h_{21}|^2}{1+|h_{22}|^2}, 
\label{eq:def strong but not very strong in general user1}
\\
|h_{12}|^2 &\geq |h_{22}|^2. %\geq \frac{|h_{12}|^2}{1+|h_{11}|^2}.
\label{eq:def strong but not very strong in general user2}
\end{align}
The strong (but not very strong) interference regime is the set of channel gains that satisfy the condition in~\eqref{eq:def strong but not very strong in general} but not the condition in~\eqref{eq:def:verystrong gen}.
\label{eq:def strong but not very strong in general}
\end{subequations}

\paragraph*{Outer Bound}
The capacity region of the general G-IC in the strong interference regime 
is given by the `compound MAC' region
\begin{align}
\mathcal{R}_{\text{out}}^{(\text{\ref{rem:strong asym}})} %(\snr,\inr)
&= \left \{ %(R_1,R_2):
\begin{array}{l}
0 \leq R_1 \leq \mug \left(|h_{11}|^2\right)
\\
0 \leq R_2 \leq \mug \left(|h_{22}|^2\right)
\\ 
R_1+R_2 \leq  \mug \left(\min \left(|h_{11}|^2+|h_{12}|^2,|h_{22}|^2+|h_{21}|^2\right) \right)
\end{array} \right\}.
\label{eq: cap: assymetric: strong int}
\end{align}

\paragraph*{Inner Bound}
The outer bound in~\eqref{eq: cap: assymetric: strong int} can be matched to within a gap by our {TINnoTS} scheme by choosing, similarly to the symmetric case discussed in detail in Section~\ref{par:strong}, the parameters of the mixed inputs as
\begin{subequations}
\begin{align}
&N_1=\points\left(\snr_{5,a,t}\right), \notag\\&
\snr_{5,a,t}= (1+|h_{11}|^2)^{1-t} \left(\frac{1+\min \left(|h_{11}|^2+|h_{12}|^2,|h_{22}|^2+|h_{21}|^2\right)}{1+|h_{22}|^2} \right)^t-1 \leq |h_{11}|^2, %\text{which satisfies $N_1 |h_{11}| d_{\min(X_1)} \leq |h_{12}| d_{\min(X_2)}$},
\label{eq:def:strongChoiseN1 gen}
\\
&N_2=\points\left(\snr_{5,b,t}\right), \notag\\&
\snr_{5,b,t}= (1+|h_{22}|^2)^{t} \left(\frac{1+\min \left(|h_{11}|^2+|h_{12}|^2,|h_{22}|^2+|h_{21}|^2\right)}{1+|h_{11}|^2} \right)^{1-t}-1 \leq |h_{22}|^2,
\label{eq:def:strongChoiseN2 gen}
\\
&\delta_1=0,
\\
&\delta_2=0,
\end{align}
where the upper bounds on $\snr_{5,a,t}$ and $\snr_{5,b,t}$ are a consequence of not being in very strong interference, i.e.,
\[
\min \left(1+|h_{11}|^2+|h_{12}|^2,1+|h_{22}|^2+|h_{21}|^2 \right) \leq (1+|h_{11}|^2)(1+|h_{22}|^2).
\]
\end{subequations}
Next, by using Proposition~\ref{prop:card:dmin:measurebound} we have
\begin{subequations}
\begin{align}
%&
%\\
&\frac{d_{\min({S}_1)}^2}{12} \geq \kappa_{\gamma,N_1,N_2}^2\min\left(\frac{|h_{11}|^2}{N_1^2-1},\frac{ |h_{12}|^2}{N_2^2-1}, 
\max \left(\frac{|h_{12}|^2}{ N_1^2(N_2^2-1) },\frac{|h_{11}|^2}{ N_2^2(N_1^2-1) } \right)\right),
\label{eq:achregion for par:strong dminS1lowerbound}
\\
&\frac{d_{\min({S}_2)}^2}{12} \geq \kappa_{\gamma,N_1,N_2}^2 \min\left(\frac{|h_{21}|^2}{N_1^2-1},\frac{ |h_{22}|^2}{N_2^2-1}, \max \left(\frac{|h_{21}|^2 }{ N_1^2(N_2^2-1) },\frac{|h_{22}|^2}{ N_2^2(N_1^2-1) } \right)\right),
\label{eq:achregion for par:strong dminS2lowerbound}
\end{align}
where the bounds in~\eqref{eq:achregion for par:strong dminS2lowerbound qwertyuiop}
hold up to a set of measure $\gamma$ and where $\kappa_{\gamma,N_1,N_2}$ is defined in~\eqref{eq:achregion for par:moderate kappakappa definitiontobeusedlater}.
\label{eq:achregion for par:strong dminS2lowerbound qwertyuiop}
\end{subequations}
By recalling the channel gain relationship, by noting that 
\begin{align*}
N_1^2N_2^2 -1 \leq \min \left(|h_{11}|^2+|h_{12}|^2,|h_{22}|^2+|h_{21}|^2\right)
\end{align*}
and by combining the two bounds in~\eqref{eq:achregion for par:strong dminS2lowerbound qwertyuiop} we get
\begin{align*}
\min_{i\in[1:2]} \frac{d_{\min({S}_i)}^2}{12}
\geq \min\left(1, \frac{\max(|h_{11}|^2,|h_{12}|^2)}{|h_{11}|^2+|h_{12}|^2}, \frac{\max(|h_{21}|^2,|h_{22}|^2)}{|h_{22}|^2+|h_{21}|^2} \right) \geq \frac{1}{2}.
\end{align*}

\paragraph*{Gap}
By following the same reasoning and bounding steps as we did for the symmetric case in Section~\ref{par:strong}, we get that the proposed achievable scheme is optimal to within a gap of 
\begin{align}
{\gap}_{\eqref{eq:gap asym strong not very strong}}\leq
\frac{1}{2}\log\left(\frac{2\pi\eu}{3} \right)+
\frac{1}{2}\log\left(1+8 \cdot \frac{\left(1+ 1/2\ln\left(1+\max\big(|h_{11}|^2,|h_{22}|^2\big)\right) \right)^2}{\gamma^2}
\right)~\text{bits}.
\label{eq:gap asym strong not very strong}
\end{align}

\subsection{Mixed Interference}
\label{rem:mixed asym}
The mixed interference regime occurs when one receiver experiences strong interference while the other experiences weak interference.
This regime does not appear in the symmetric case, where both receiver are either in strong interference or in weak interference.
The mixed interference is defined as~\cite{etkin_tse_wang}
\begin{subequations}
\begin{align}
\text{either} \ \ 
\big\{ |h_{21}|^2 \geq |h_{11}|^2, \ |h_{12}|^2 \leq |h_{22}|^2  \big\},   
\label{eq:def mixed in general user1}
\\
\text{or} \ \ 
\big\{ |h_{21}|^2 \leq |h_{11}|^2, \ |h_{12}|^2 \geq |h_{22}|^2  \big\}.    
\label{eq:def mixed in general user2}
\end{align}
\label{eq:def mixed in general}
\end{subequations}
In this Section we shall only focus on the sub-regime
\begin{align}
|h_{21}|^2 \geq \frac{|h_{11}|^2}{1+|h_{12}|^2}(1+|h_{22}|^2), \ |h_{12}|^2 \leq |h_{22}|^2,
\label{eq:def specialmixed in general}
\end{align}
for which the rate region, as we shall see, does not require bounds on $2R_1+R_2$ or $R_1+2R_2$.
The regime $|h_{12}|^2 \geq \frac{|h_{22}|^2}{1+|h_{21}|^2}(1+|h_{11}|^2), \ |h_{21}|^2 \leq |h_{11}|^2$
can be analyzed similarly by swapping the role of the users.

\paragraph*{Outer Bound}
An outer bound to the capacity region of the general G-IC when~\eqref{eq:def specialmixed in general} holds is given by the `Z-channel' outer bound~\cite{sason}
\begin{align}
\mathcal{R}_{\text{out}}^{(\text{\ref{rem:mixed asym}})} 
&= \left \{ 
\begin{array}{l}
0\leq R_1 \leq \mug \left(|h_{11}|^2\right) \\
0\leq R_2 \leq \mug \left(|h_{22}|^2\right) \\
R_1+R_2   \leq \mug\left(|h_{22}|^2\right)+\mug\left(\frac{|h_{11}|^2}{1+|h_{12}|^2}\right) \\
\end{array} \right\} \notag\\
&= \bigcup_{t\in[0,1]}\left \{ 
\begin{array}{l}
0 \leq R_1 \leq (1-t)\mug \left(|h_{11}|^2\right)+t \mug\left(\frac{|h_{11}|^2}{1+|h_{12}|^2}\right) \\
\quad = \mug(\snr_{6,a,t}) \\
0 \leq R_2 \leq (1-t)\left( \mug\left(|h_{22}|^2\right) - \mug\left(|h_{11}|^2\right)
+ \mug\left(\frac{|h_{11}|^2}{1+|h_{12}|^2}\right) \right)+t  \mug(|h_{22}|^2) \\
\quad = \mug(\snr_{6,b,t}) + \frac{1}{2}\log\left(\frac{1+|h_{22}|^2}{1+|h_{12}|^2}\right)\\
\end{array} \right\}.
\label{eq:upper classical IC Z}
\end{align}

\paragraph*{Inner Bound}
The shape of the outer bound in~\eqref{eq:upper classical IC Z} suggests that a matching, to within a gap, inner region could be found by following steps similar to those used for the analysis of the strong interference regime (i.e., parameterize the points on the dominate sum-capacity face). The difference between this sub-regime  
and the strong interference regime is that here $R_2$ should be a combination of common and private rates
because receiver~1 experiences weak interference (while receiver~2 experiences strong interference). 
Note that the interfering channel gain at receiver~2, $h_{21}$, does not appear in the outer bound in~\eqref{eq:upper classical IC Z}.
We therefore set 
\begin{subequations}
\begin{align}
&N_1=\points\left(\snr_{6,a,t}\right) : %\notag\\&
\snr_{6,a,t}= \left(1+\frac{|h_{11}|^2}{1+|h_{12}|^2}\right)^{t}\left(1+|h_{11}|^2\right)^{1-t}-1
         \leq |h_{11}|^2,
\label{eq:def:mixedChoiseN1 gen}
\\
&N_2=\points\left(\snr_{6,b,t}\right) : %\notag\\&
\snr_{6,b,t}= \left(1+|h_{12}|^2\right)^{t} \left(1+\frac{|h_{12}|^2}{1+|h_{11}|^2}\right)^{1-t} -1
         \leq |h_{12}|^2,
\label{eq:def:mixedChoiseN2 gen}
\\
&\delta_1=0,
\label{eq:def:mixedChoised1 gen}
\\
&\delta_2=\frac{1}{1+|h_{12}|^2},
\label{eq:def:mixedChoised2 gen}
\end{align}
\label{eq:def:mixedChoiseALL gen}
\end{subequations}
in the achievable region in Proposition~\ref{prop:ach-with-mixedinput}, which becomes
\begin{subequations}
\begin{align}
\mathcal{R}^{(\text{\ref{par:mixed}})}_{\text{in}} %(\snr,\inr) 
&= \left \{ %(R_1,R_2):
\begin{array}{l}
0\leq R_1 \leq \log(N_1)-\Delta_{(\text{\ref{par:mixed}})}
\\
0\leq R_2 \leq \log(N_2)
+ \mug \left(\frac{|h_{22}|^2}{1+|h_{12}|^2} \right)
-\Delta_{(\text{\ref{par:mixed}})}
\end{array} \right\},
\label{eq:achregion for mixed}
%\ 
\\
\Delta_{(\text{\ref{par:mixed}})}&
= \frac{1}{2}\log\left(\frac{2\pi\eu}{12}\right) 
+ \frac{1}{2}\log\left(1+\frac{12}{\min_{i\in[1:2]} d_{\min(S_i)}^2}\right),
\label{par:mixed Delta expression}
\\
{S}_1 &= \frac{1}{\sqrt{1+ \frac{|h_{12}|^2}{1+|h_{12}|^2}}}\left( h_{11}X_{1D}+\sqrt{ \frac{|h_{12}|^2}{1+|h_{12}|^2}}h_{12}X_{2D}\right),
\label{par:mixed S1 expression}
\\
{S}_2 &= \frac{1}{\sqrt{1+|h_{22}|^2\frac{1}{1+|h_{12}|^2}}}\left( h_{21}X_{1D}+\sqrt{ \frac{|h_{12}|^2}{1+|h_{12}|^2}}h_{22}X_{2D}\right). 
\label{par:mixed S2 expression}
\end{align}
\label{par:mixed}
\end{subequations}
Next, by using Proposition~\ref{prop:card:dmin:measurebound}, we bound the minimum distance of the received constellations ${S}_1$ and ${S}_2$ as
\begin{align*}
\frac{d_{\min({S}_1)}^2}{12 \ \kappa_{\gamma,N_1,N_2}^2} 
&
\geq %\kappa_{\gamma,N_1,N_2}^2 \cdot 
\frac{1}{1+|h_{12}|^2 \delta_2}\min\left(\frac{|h_{11}|^2}{N_1^2-1},\frac{ (1-\delta_2) |h_{12}|^2}{N_2^2-1}, \max \left(\frac{(1-\delta_2)|h_{12}|^2}{ N_1^2(N_2^2-1) },\frac{|h_{11}|^2}{ N_2^2(N_1^2-1) } \right)\right)
\\&%\frac{d_{\min({S}_1)}^2}{12} & 
\stackrel{\rm (a)}{\geq} %\kappa_{\gamma,N_1,N_2}^2 \cdot 
\frac{1}{1+|h_{12}|^2 \delta_2}\min\left(1,(1-\delta_2), \frac{\max \left((1-\delta_2)|h_{12}|^2,|h_{11}|^2\right)}{N_1^2N_2^2-1}\right)
\\&
\geq %\kappa_{\gamma,N_1,N_2}^2 \cdot 
\frac{1-\delta_2}{1+|h_{12}|^2 \delta_2}\min\left(1, \frac{\max \left(|h_{12}|^2,|h_{11}|^2\right)}{N_1^2N_2^2-1}\right)
\\&
\stackrel{\rm (b)}{\ge} 
\frac{|h_{12}|^2}{2+|h_{12}|^2}\min\left(1,\frac{\max \left(|h_{12}|^2,|h_{11}|^2\right)}{|h_{12}|^2+|h_{11}|^2}\right)
\\& \stackrel{\rm (c)}{\ge} 
\frac{1}{3}\min\left(1,\frac{1}{2} \right) = \frac{1}{6},
\end{align*}
where the inequalities follow since:
(a) by using the bounds in~\eqref{eq:def:mixedChoiseN1 gen} and~\eqref{eq:def:mixedChoiseN2 gen}, 
(b) because $N_1^2 N_2^2-1 \leq |h_{12}|^2+|h_{11}|^2$ from~\eqref{eq:def:mixedChoiseN1 gen} and~\eqref{eq:def:mixedChoiseN2 gen}, and
(c) by assuming $|h_{12}|^2 \geq 1$.

Similarly we have that 
\begin{align*}
\frac{d_{\min({S}_2)}^2}{12 \ \kappa_{\gamma,N_1,N_2}^2} 
&\geq 
\frac{1}{1+|h_{22}|^2\delta_2} \min\left(\frac{|h_{21}|^2}{N_1^2-1},\frac{ (1-\delta_2)|h_{22}|^2}{N_2^2-1}, \max \left(\frac{|h_{21}|^2 }{ N_2^2(N_1^2-1) },\frac{(1-\delta_2)|h_{22}|^2}{ N_1^2(N_2^2-1) } \right)\right)
\\&
\stackrel{\rm (a)}{\geq} 
\frac{1-\delta_2}{1+|h_{22}|^2\delta_2} \min\left(\frac{|h_{21}|^2}{|h_{11}|^2},\frac{ |h_{22}|^2}{|h_{12}|^2}, \frac{ \max (|h_{21}|^2 ,|h_{22}|^2)}{ |h_{12}|^2+|h_{11}|^2 }\right)
\\&
\stackrel{\rm (b)}{\geq} 
\frac{1-\delta_2}{1+|h_{22}|^2\delta_2} 
\min\left( \frac{1+|h_{22}|^2}{1+|h_{12}|^2},\frac{|h_{22}|^2}{|h_{12}|^2}, 
\frac{ \max \left( |h_{11}|^2\frac{1+|h_{22}|^2}{1+|h_{12}|^2},|h_{12}|^2 \frac{|h_{22}|^2}{|h_{12}|^2} \right)}
{ |h_{12}|^2+|h_{11}|^2 }\right)
%\\&
\\&
\stackrel{\rm (b)}{\geq} 
 \frac{1-\delta_2}{1+|h_{22}|^2\delta_2} \
 \frac{1+|h_{22}|^2}{1+|h_{12}|^2} \ \frac{1}{2}
\\&
= 
 \frac{1+|h_{22}|^2}{1+|h_{12}|^2+|h_{22}|^2 } \
 \frac{|h_{12}|^2}{1+|h_{12}|^2} \ \frac{1}{2}
\\&
\stackrel{\rm (c)}{\geq} 
\frac{2}{3} \ \frac{1}{2} \ \frac{1}{2} = \frac{1}{6},
\end{align*}
where the inequalities follow since:
(a) by using the bounds in~\eqref{eq:def:mixedChoiseN1 gen} and~\eqref{eq:def:mixedChoiseN2 gen} and because $N_1^2 N_2^2-1 \leq |h_{12}|^2+|h_{11}|^2$, 
(b) by the channel gain relationship in~\eqref{eq:def specialmixed in general}, and 
(c) by assuming $1 \leq |h_{12}|^2$ and since by assumption of this regime $|h_{12}|^2 \leq |h_{22}|^2$.
Note that the assumption $1 \leq |h_{12}|^2$ is without loss of generality since if
$|h_{12}|^2 < 1$ (i.e., interference below the noise floor of the receiver) then TIN with Gaussian inputs achieves the capacity outer bound
(in this case essentially two interference-free point-to-point links) to within 1/2~bit.

This shows that
\begin{align}
\min_{i \in[1:2]}
\frac{d_{\min ({S}_i)}^2}{12}
&\geq \kappa_{\gamma,N_1,N_2}^2 \cdot \frac{1}{6},
\label{eq:dminSi mixed asym}
\end{align}
up to an outage set of measure no more than $\gamma$, where $\gamma$ affects $\kappa_{\gamma,N_1,N_2}$.

\paragraph*{Gap}
By following the same reasoning and bounding steps as we did for the symmetric case, 
we get that the proposed achievable scheme is optimal to within a gap of 
\begin{align*}
\Delta_{R_1}
&\leq
\mug(\snr_{6,a,t})-\log(\points(\snr_{6,a,t}))+\Delta_{(\text{\ref{par:mixed}})}
\\&
\leq \Delta_{(\text{\ref{par:mixed}})}+\log(2),
\\
\Delta_{R_2}
&\leq
\mug(\snr_{6,b,t}) + \frac{1}{2}\log\left(\frac{1+|h_{22}|^2}{1+|h_{12}|^2}\right)
-\log(\points(\snr_{6,b,t}))
-\mug \left(\frac{|h_{22}|^2}{1+|h_{12}|^2} \right)
+\Delta_{(\text{\ref{par:mixed}})}
\\&
\leq \Delta_{(\text{\ref{par:mixed}})}+\log(2),
\end{align*}
where we used the fact that $\log(\points(x)) \geq \mug(x) - \log(2)$.
By including the minimum distance bound in~\eqref{eq:dminSi mixed asym}
into the expression for $\Delta_{(\text{\ref{par:mixed}})}$ in~\eqref{par:mixed Delta expression},
and by noticing that $\max(N_1^2,N_2^2)-1 \leq \max(|h_{11}|^2,|h_{12}|^2) \leq \max(|h_{11}|^2,|h_{22}|^2)$ 
by the channel gain relationship in~\eqref{eq:def specialmixed in general},
we finally get
\begin{align}
{\gap}_{\eqref{eq:gap asym smixed not all though}}
&\leq
  \frac{1}{2}\log\left(\frac{2\pi\eu}{3}\right) 
+ \frac{1}{2}\log\left(1+\frac{6}{\kappa_{\gamma,N_1,N_2}^2}\right),
\notag\\&\leq
\frac{1}{2}\log\left(\frac{2\pi\eu}{3} \right)+
\frac{1}{2}\log\left(1+24 \cdot \frac{\left(1+ 1/2\ln\left(1+\max\big(|h_{11}|^2,|h_{22}|^2\big)\right) \right)^2}{\gamma^2}
\right). %~\text{bits}.
\label{eq:gap asym smixed not all though}
\end{align}

\subsection{Weak Interference}
\label{rem:mod weak asym}

For the general asymmetric G-IC, the weak interference is defined as 
\begin{subequations}
\begin{align}
|h_{21}|^2 \leq |h_{11}|^2,     
\label{eq:def weak in general user1}
\\
|h_{12}|^2 \leq |h_{22}|^2,   
\label{eq:def weak in general user2}
\end{align}
which involves numerous special cases
whose analysis gets very tedious and is outside of the scope of this paper -- see also Remark~\ref{rem:whatNOTinasymcase}.
\label{eq:def weak in general}
\end{subequations}

\subsection{Very Weak Interference}
\label{rem:very weak asym}
The very weak interference regime characterized in as~\cite{TINweakKuser} is defined as
\begin{subequations}
\begin{align}
     |h_{12}|^2  \leq \frac{|h_{22}|^2}{1+|h_{21}|^2}, \\
     |h_{21}|^2  \leq \frac{|h_{11}|^2}{1+|h_{12}|^2}. 
\end{align}
\label{eq:def:veryweak gen}
\end{subequations}
In this regime, the outer bound to the capacity region of the classical G-IC is
\begin{align}
\mathcal{R}^{(\text{\ref{rem:very weak asym}})}_{\text{out}}
&= \left \{ %(R_1,R_2):
\begin{array}{l}
R_1 \leq \mug \left(|h_{11}|^2\right)
\\
R_2 \leq \mug \left(|h_{22}|^2\right)
\\
R_1+R_2 \leq
  \mug\left(|h_{12}|^2+\frac{|h_{11}|^2}{1+|h_{21}|^2}\right)
+ \mug\left(|h_{21}|^2+\frac{|h_{22}|^2}{1+|h_{12}|^2}\right)
\end{array} \right\}
\label{eq:def:veryweakOuter gen}
\end{align}
and is achievable to within $1/2$~bit by Gaussian inputs with power control and TIN. 
Since the optimal strategy for the classical G-IC is compatible with our {TINnoTS} with mixed inputs,
we conclude that a mixed-input is optimal to within $1/2$~bit in this regime.

This concludes the proof of Theorem~\ref{thm:Cap:gap some asymmetric}.
\end{IEEEproof}

\section{{TINnoTS} is gDoF optimal}
\label{sec:gDoF:cap}
In this section we show one of the consequences of Theorem~\ref{thm:Cap:gap sym}, namely that {TINnoTS} is gDoF optimal almost surely.
The notion of gDoF has been introduced in~\cite{etkin_tse_wang} and has become an important metric that sheds lights on the behavior of the capacity  when exact capacity results are not available. The gDoF region is the set    
\begin{align}
\mathcal{D} := \left \{ (d_1,d_2) \in\mathbb{R}^2_+: 
d_i :=\lim_{\tiny
\begin{array}{l}
\inr=\snr^\alpha, \\
\snr \to \infty \\
\end{array}
%\inr=\snr^\alpha, \ \snr \to \infty
} \frac{R_i}{\frac{1}{2}\log(1+\snr)},  i \in [1:2], \ (R_1,R_2) \ \text{is achievable}  \right\}.
\label{eq:low:gDofregion}
\end{align}
The $\loglog$ additive gap result of Theorem~\ref{thm:Cap:gap sym} implies that: 
\begin{thm}
\label{thm:gDoF:cap}
For the symmetric G-IC the {TINnoTS} achievable scheme with mixed inputs
%region $\mathcal{R}_{\text{in}}(\snr,\inr)$ in~\eqref{eq: rates mixed inputs WITH Union} 
is gDoF optimal for all channel gains up to a set of zero measure.
\end{thm}
\begin{IEEEproof}
We must show that as $\snr \to \infty$ the gap between the {TINnoTS} inner bound 
%$\mathcal{R}_{\text{in}}(\snr,\inr)$ 
and outer bound in Proposition~\ref{prop:ETWouter}, %$\mathcal{R}_{\text{out}}(\snr,\inr)$
normalized by $\mug(\snr)$, goes to zero almost everywhere. 

In the proof of Theorem~\ref{thm:Cap:gap sym} we showed that for the very strong, the weak2 and the very weak interference regimes
%and the  from~\cite{etkin_tse_wang} for Very Weak Inference
the gap between inner and outer bounds is $O(1)$ everywhere. Therefore, since $\lim_{\snr \to \infty} \frac{O(1)}{\mug(\snr)}=0$, the result follows.  

For the strong and weak1 interference regimes the gap is of the form $\loglogG$ for any $\gamma \in (0,1]$. Therefore,  by choosing $\gamma$ to  be
\begin{align*}
\gamma(\snr) := \frac{1}{(\log \min(\snr,\snr^\alpha)^p}, \ \text{for some $p>0$ independent of $\snr$},
\end{align*}
we have that $\lim_{\snr \to \infty} \frac{\loglogG}{\mug(\snr)}=0 $  and the measure of the outage set $\gamma(\snr)$ vanishes as $\snr\to\infty$. 
This concludes the proof. 
\end{IEEEproof}

\section{Totally Asynchronous and Codebook Oblivious G-IC}
\label{sec: async and obl IC}
The only requirement for the implementation of the {TINnoTS} inner bound in~\eqref{eq:RL:TINnoTS} is to have symbol synchronization and knowledge of the channel gains at all the terminals. Therefore, our {TINnoTS} achievable strategy applies  to a large class of channels, besides the model considered thus far. Next, we outline two such examples for which very little was known in the past.

The first example is the block asynchronous G-IC, which is information unstable~\cite{totAsynchIC} and thus no single-letter capacity expression can be derived for it. Nonetheless, we are able to show that the capacity  of this channel is to within  a gap of the capacity of the fully synchronized channel.
The second example is the G-IC with partial codebook knowledge at both receivers~\cite{simion_codebook}, which prevents using joint decoding or successive interference cancellation at the decoders. Still, we are able to show that the capacity  of this channel is to within a gap of the capacity of the channel with full codebook knowledge.

The applications to oblivious and asynchronous ICs somewhat surprisingly implies that much less ``global coordination'' between nodes is needed than one might initially expect: synchronism and codebook knowledge might not be critical if one is happy with `approximate' capacity results. 

\subsection{Block Asynchronous G-IC}
Consider a G-IC with the following input-output relationship
\begin{subequations}
\begin{align}
Y_1^n&=h_{11}X_1^n+h_{12}X_2^{n-D_1}+Z_1^n, \\ 
Y_2^n&=h_{21}X_1^{n-D_2}+h_{22}X_2^n+Z_2^n, 
\end{align}
\label{eq:Def:Total Asynch G-IC}
\end{subequations}
where the delay $D_i$, $i \in[1:2]$, is chosen at the beginning of the transmission and held fixed thereafter.  The channel is termed totally asynchronous if delay is uniform on all $n$~\cite{totAsynchIC}. Except for the introduction of random delay  all  definitions are identical to those given in Section~\ref{sec:Intro}.  In~\cite{totAsynchIC} it has been shown that $\mathcal{R}_{\text{in}}^\text{TINnoTS}$ in~\eqref{eq:RL:TINnoTS} is achievable for the channel in~\eqref{eq:Def:Total Asynch G-IC}. Moreover,  because lack of synchronization can only harm communications, the outer bound %in~\eqref{eq:R upper classical IC}
in Proposition~\ref{prop:ETWouter} is a valid outer bound for the asynchronous G-IC.
Therefore, all of our previous results hold and we have: 
\begin{lem}
For the block asynchronous G-IC the {TINnoTS} achievable region is to within an additive gap of the capacity of the fully synchronized G-IC,
where the gap is given in Theorems~\ref{thm:Cap:gap sym} and~\ref{thm:Cap:gap some asymmetric}.
\end{lem}

\subsection{IC with No Codebook Knowledge}
IC with partial codebook knowledge, or oblivious receivers (IC-OR), has been introduced in~\cite{simion_codebook}. 
This channel model is practically relevant  because it  models the inability to use sophisticated decoding techniques such as joint decoding or successive inference cancellation.  Recently, in~\cite{DytsoCodebookJournal},  for the IC-OR with partial codebook knowledge at one receiver, it has been shown that   using  Gaussian input at the transmitter corresponding to the oblivious receiver and a mixed input at the transmitter corresponding to non-oblivious receiver is to within a constant gap from the capacity of the classical G-IC with full codebook knowledge. In~\cite{simion_codebook} it was shown that for IC-OR with both oblivious receivers the capacity is given by 
\begin{align}
\mathcal{C}^{\text{IC-OR}} &= \bigcup_{P_QP_{X_1|Q}P_{X_2|Q} }
\left\{ 
\begin{array}{l} 
 R_1 \leq I(X_1;Y_1|Q)\\
 R_2 \leq I(X_2;Y_2|Q)\\
\end{array} \right \}.
\label{cap:IC-OR}
\end{align} 
Note that the region in~\eqref{cap:IC-OR}  is very similiar to {TINoTS} region in~\eqref{eq:RL:TINwithT} and $\mathcal{C}^{\text{IC-OR}}$ is upper bounded by the classical G-IC outer bound in Proposition~\ref{prop:ETWouter}. 
The set of optimizing distributions for~\eqref{cap:IC-OR} and the cardinality bound for the alphabet of $Q$ are not known~\cite[Section III.A]{simion_codebook}. Based on our previous results, we have that:
\begin{lem}
For the  G-IC with partial codebook knowledge the {TINnoTS} achievable region is to within an additive gap of the capacity of the G-IC with full codebook knowledge, where the gap is given in Theorems~\ref{thm:Cap:gap sym} and~\ref{thm:Cap:gap some asymmetric}.
\end{lem}

\section{{TINnoTS} with Mixed Inputs in Practice}
\label{sec:practical}

\subsection{A Simple {TINnoTS} Receiver in Very Strong Interference} 
\label{sec: simple receiver}
In the Introduction we mentioned that the optimal MAP decoder in an additive non-Gaussian noise channel, 
which one could implement for TIN when treating a non-Gaussian interference as noise, 
could be very complex. In the following we give an example of an approximate
MAP decoder that is very simple to implement, thus making {TINnoTS} competitive in practical applications.

Let $X_1,X_2$ be from the $\pam\left(N, d\right)$ with  $N=2Q+1, Q\in\mathbb{N},$ and $d^2=\frac{12}{N^2-1}=\frac{3}{Q(Q+1)}$. 
The restrictions to an odd number of points is just for simplicity of writing the constellation points.
The received signal is
\begin{align*}
Y = \left(\sqrt{\snr} \ n_1  + \sqrt{\inr} \ n_2\right) d + Z_G, \ Z_G \sim \mathcal{N}(0,1),
\end{align*}
for some $(n_1,n_2)\in[-Q:Q]^2$ chosen independently with uniform probability. 
The condition in~\eqref{eq:thm:comb condition2} is verified when
\begin{align}
(2Q+1)^2\snr \leq \inr,
\label{eq:nonoverlap in practice very strong}
\end{align}
which corresponds to the very strong interference regime.
In the regime identified by~\eqref{eq:nonoverlap in practice very strong},
i.e., where the received points do not `overlap' as in Fig~\ref{fig:noOverLap},
 the decoder could simply ``modulo-out'' the interference 
by ``folding'' the signal $Y$ onto the interval $\mathcal{I} := [-\sqrt{\inr}d/2, +\sqrt{\inr}d/2]$.
By doing so the resulting signal, given by
\begin{align*}
Y^{\prime} 
%= \left[\sqrt{\snr} \ n_1 d + Z_G \right]_{\text{mod} \ \mathcal{I}}
=  \left[\sqrt{\snr} \ n_1 d + Z^{\prime}\right]_{\text{mod} \ \mathcal{I}}, \ Z^{\prime} := [Z_G]_{\text{mod} \ \mathcal{I}},
\end{align*}
would be interference-free.
Since
\begin{align*}
\Pr[Y^{\prime} \not= \sqrt{\snr} \ n_1 d + Z_G]
  &\leq \Pr[|Z_G| \geq\sqrt{\inr}d/2- \sqrt{\snr} Q d]
\\&\stackrel{\text{from~\eqref{eq:nonoverlap in practice very strong}}}{\leq} \Pr[|Z_G| \geq\sqrt{\snr}/2],
\end{align*}
and since $\Pr[|Z_G| \geq\sqrt{\snr}/2]$ is also an upper bound to the probability of error for PAM input on a Gaussian channel,
we see that the simple modulo operation at the receiver results in a symbol-error rate that is at most double that of an interference-free Gaussian channel with the same PAM input.

\subsection{Actual vs. Analytic Gap} 
\label{sec: Actual vs. Analytic Gap}
Here we compare the gap derived in Theorems~\ref{thm:Cap:gap sym} and~\ref{thm:Cap:gap some asymmetric} to the actual gap evaluated numerically.
The point is to show that our analytical closed-form (worst case scenario) bounds can be quite conservative and thus underestimate the actual achievable rates. 

For example, we showed that in the very strong interference regime
%,
the {TINnoTS} achievable region with discrete inputs is at most $\frac{1}{2} \log\left(\frac{2 \pi \eu}{3}\right)$~bits from capacity;
the capacity in this case is the same as two parallel interference-free links.
Consider the symmetric G-IG in very strong interference and the symmetric rate $R_1=R_2=R_\text{sym}(\snr)$ with the same PAM input for each user, where the number of points is chosen as  in \eqref{eq:very strong:N}.
Fig.~\ref{fig:NumericalGapInVeryStrong} shows $\gap(\snr) := \mug(\snr)-R_\text{sym}(\snr)$ vs. $\snr$ expressed in dB, where
\begin{itemize}
\item
the red   line is the theoretical gap from Theorem~\ref{thm:Cap:gap sym}, approximately $\frac{1}{2} \log\left(\frac{2 \pi \eu}{3}\right)=1.25$~bits;
\item
the green line is the gap by lower bounding $R_\text{sym}(\snr)$ with the Ozarow-Wyner-B bound in Proposition~\ref{prop:lowbound OW generalized}, where the minimum distance of the received constellation was computed exactly (rather than lower bounded by Proposition~\ref{prop:combNOOUTAGE}); the gap in this case is approximately $0.75$~bits;
\item
the magenta line is the gap by lower bounding $R_\text{sym}(\snr)$ by the `full DTD-ITA'14 bound' in~\eqref{eq:ITA with Shaping}, the gap in this case is approximately $0.37$~bits;
\item
the cyan    line is the gap when $R_\text{sym}(\snr)$ is evaluated by Monte Carlo simulation;
the gap in this case tends to the ultimate ``shaping loss'' $\frac{1}{2} \log\left(\frac{\pi \eu}{6}\right)=0.25$~bits at large $\snr$;
this shows that the actual gap is about 1~bit lower than the theoretical gap;
\end{itemize}
The figure also shows that the lower bound in~\eqref{eq:ITA with Shaping} actually gives the tightest lower bound for the mutual information, but it is unfortunately not easy to deal with analytically.

We next consider the symmetric G-IG in strong interference.
Theorem~\ref{thm:Cap:gap sym} upper bounds the gap in this regime 
by $\gap(\snr) \leq \frac{1}{2}\log\left(\frac{2\pi\eu}{3}\right)+\frac{1}{2}\log\left(1+8\frac{\left(1+1/2\ln({1+\snr})\right)^2}{\gamma^2}\right)$
where $\gamma\in(0,1]$ is the measure of the outage set (i.e., those channel gains for which the gap lower bound is not valid).
If we were to make the measure of the outage set very small, %
then we could end up finding that the gap is actually larger than capacity.
Consider the case $\snr=30$~dB and $\inr=\snr^{1.49}=44.7$~dB;
with $\gamma=0.1$ it easy to see that $\frac{1}{2}\log\left(\frac{2\pi\eu}{3}\right)+\frac{1}{2}\log\left(1+8\frac{\left(1+1/2\ln({1+\snr})\right)^2}{\gamma^2}\right) = 6.977$~bits, which is larger than the interference-free capacity $\mug(\snr) = 4.9836$~bits.
This implies that our bounding steps, done for the sake of analytical tractability and especially meaningful  at high SNR, are too crude for this specific example (where 
our result states the trivial fact that zero rate for each user is achievable to within $\mug(\snr)$ bits). 
We aim to convey next that, despite the fact that the  closed-form gap result underestimates the achievable rates, it nonetheless provides valuable insights into the performance of practical systems, that is, that {TINnoTS} with discrete inputs performs quite well in the strong interference regime (where capacity is achieved by Gaussian codebooks and joint decoding of interfering and intended messages).
To this end, Fig.~\ref{fig:NumericalGapInStrong} shows the achievable rate region for the symmetric G-IC with $\snr=30$~dB and $\inr=\snr^{1.49}=44.7$~dB and where the users employ a PAM input with the number of points given by \eqref{eq:achregion for par:strong choiceN}. We observe
\begin{itemize}
\item
The navy blue line shows the pentagon-shaped capacity region in~\eqref{eq:outerregion for par:strong}.
\item
The red point at the origin is the lower bound on the achievable rates from Theorem~\ref{thm:Cap:gap sym} with $\gamma=0.1$.
\item
The green line is the achievable region when the rates are lower bounded by the Ozarow-Wyner-B bound in Proposition~\ref{prop:lowbound OW generalized}, where the minimum distances of the received constellations were computed exactly (rather than lower bounded by Proposition~\ref{prop:card:dmin:measurebound}). 
\item
For the magenta line we used the DTD-ITA'14-A lower bound in~\eqref{eq:ITA with Shaping}; 
\item
For the cyan line we evaluated the rates by Monte Carlo simulation.
\end{itemize}
The reason why the green region has so many `ups and downs' is because the Ozarow-Wyner-B bound in Proposition~\ref{prop:lowbound OW generalized} depends on the constellation through its minimum distance; as we already saw in Fig.~\ref{fig:minimumDistanceBehaviour}, the minimum distance is very sensitive to the fractional values of the channel gains, which makes the corresponding bound looks very irregular.
On the other hand, the magenta region is based on the lower bound in~\eqref{eq:ITA with Shaping}, which depends on the whole distance spectrum of the received constellation and as a consequence the corresponding bound looks smoother.
The cyan region is the smoothest of all; its largest gap occurs at the symmetric rate point and is less than 0.7~bits -- as opposed to the theoretical gap of 4.9836~bits.
We thus conclude that, despite the large theoretical gap, a PAM input is quit competitive in this example.

\begin{figure}
        \centering
        \begin{subfigure}[a]{0.5\textwidth}
                \includegraphics[width=8cm]{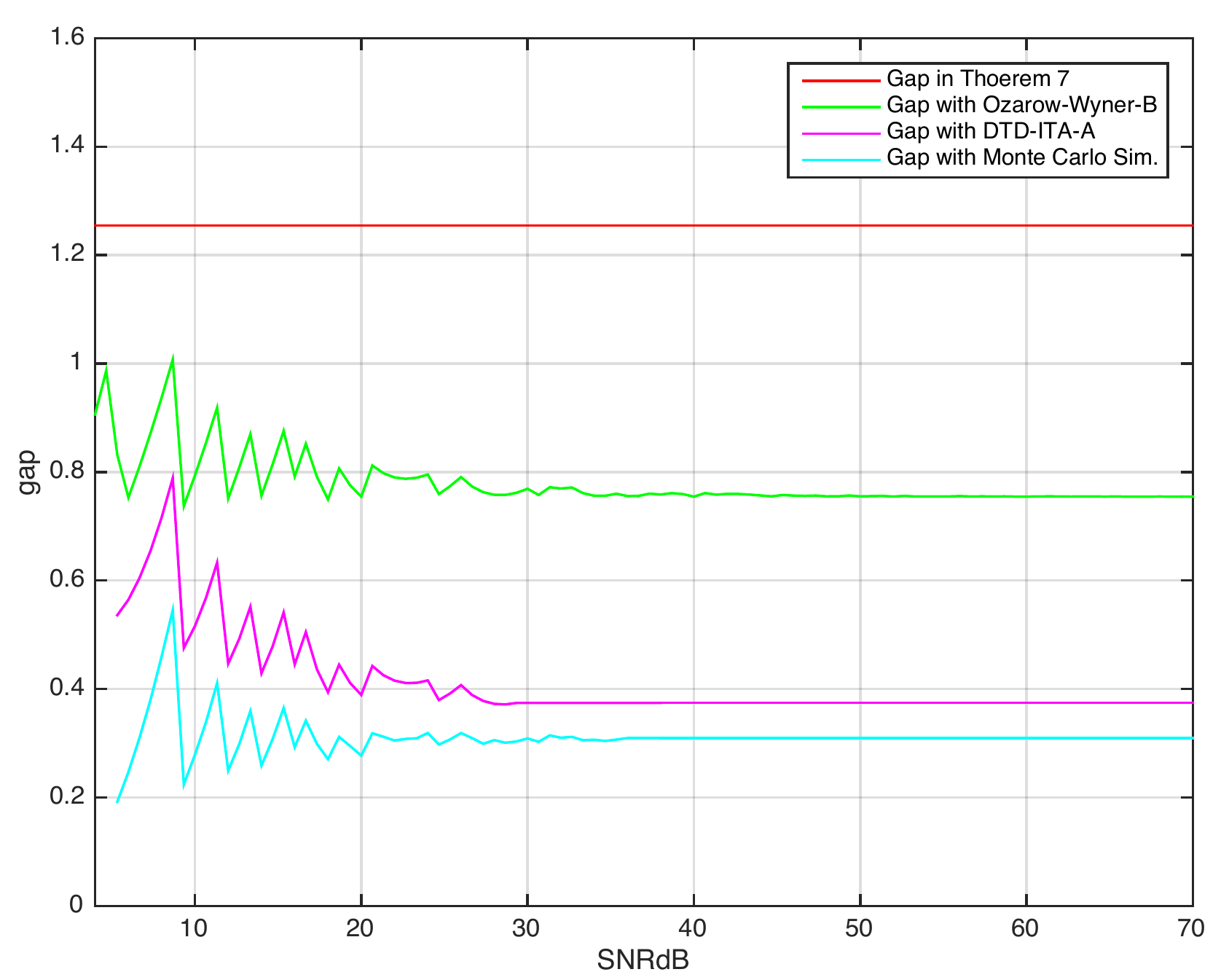}
                \caption{Gap in the very strong interference regime vs. $\snr$ for $\inr=(\snr \ (1+\snr))^{1.2} \approx \snr^{2.4}$.}
                \label{fig:NumericalGapInVeryStrong}
        \end{subfigure}%
        ~ %add desired spacing between images, e. g. ~, \quad, \qquad, \hfill etc.
        \begin{subfigure}[a]{0.5\textwidth}
                \includegraphics[width=8cm]{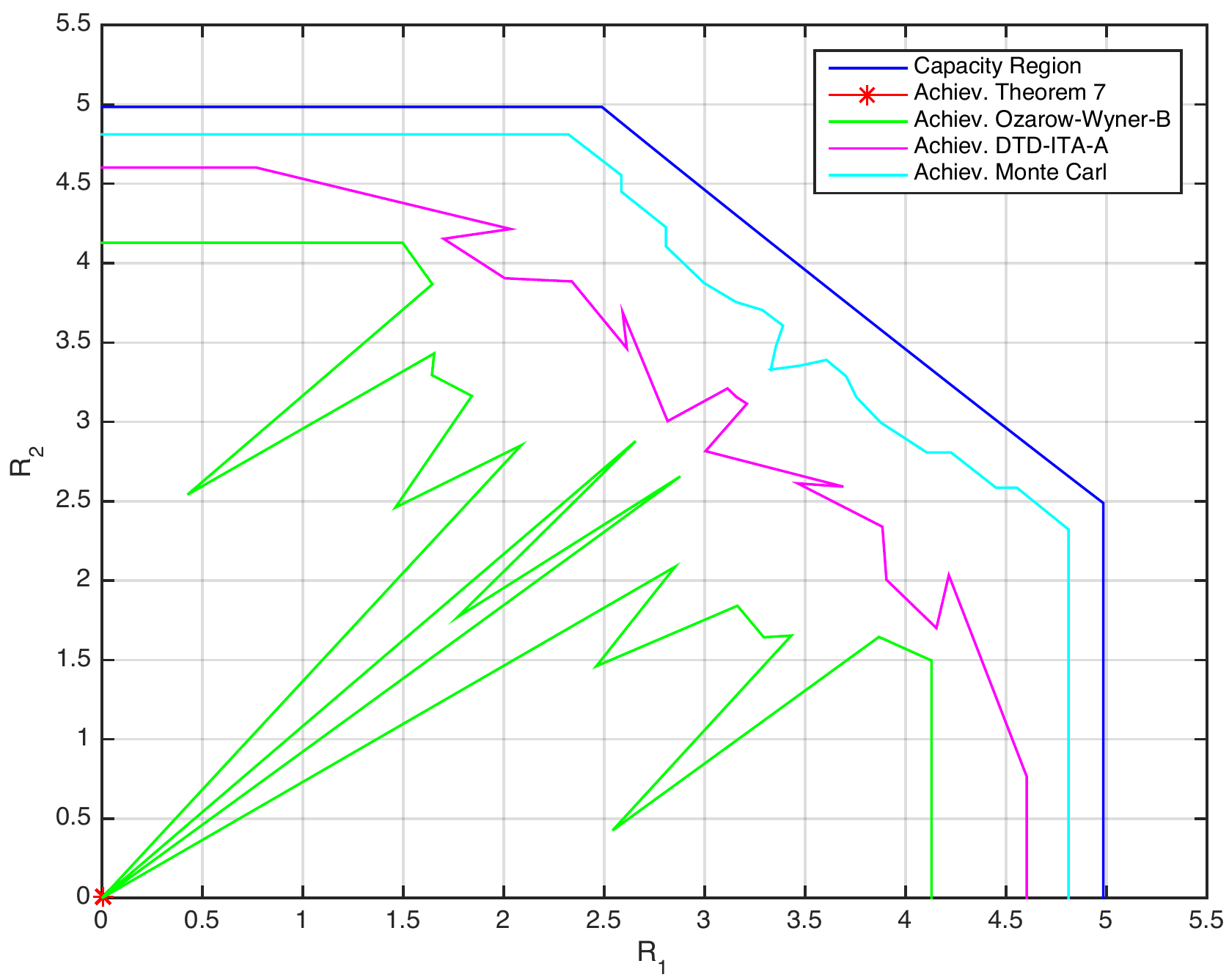}
                \caption{Rate region in the strong interference regime for $\snr=30$~dB and $\inr=\snr^{1.49}=44.7$~dB.}
                \label{fig:NumericalGapInStrong}
        \end{subfigure}
        ~ %add desired spacing between images, e. g. ~, \quad, \qquad, \hfill etc.
        \caption{Comparing Analytic with Numerical  Gaps.}
        \label{fig:whatthe*isthis}
\end{figure}

\subsection{Mixed (Gaussian+Discrete) vs. Discrete (Discrete+Discrete) Inputs}
\label{sec: Mixed (Gaussian+Discrete) vs. Discrete (Discrete+Discrete) Inputs}
In the previous Sections we showed that %from theoretical point of view mixed inputs  and 
{TINnoTS} with mixed (Gaussian+Discrete) inputs achieves the capacity to within a gap for several channels of interest. 
Practically, it may be interesting to understand what performance can be guaranteed when inputs are fully discrete, i.e., they do not contain a Gaussian component.

For the symmetric G-IC the following can be shown.
Consider the {TINnoTS} region with $X_u\sim \pam(N_u,d_u)$ such that the power constraints are met, that is, $\frac{N_u^2-1}{12}d_u^2\leq 1$ for all $u\in[1:2],$ and lower bound the mutual informations with Proposition~\ref{prop:lowbound OW generalized}. Then, {TINnoTS} achieves the outer bound in Proposition~\ref{prop:ETWouter} in very weak and in strong interference only, that is, for those regimes where `rate splitting' was not used in Theorem~\ref{thm:Cap:gap sym}. The proof of this result is omitted for sake of space. Thus it appears that  in the moderately weak interference regime mixed inputs composed of `two-layers' are necessary.

The next question we ask is thus whether we can show the same gap result of Theorem~\ref{thm:Cap:gap sym} for the moderately weak interference regime by using inputs that are the superposition of two PAM constellations, rather than a PAM and a Gaussian.
The next proposition shows that 
the answer is in the affirmative, i.e., 
it is possible to `switch' between Gaussian+Discrete and Discrete+Discrete inputs up to an additive gap.
\begin{prop}
\label{prop:switch D with G}

Let 
\begin{align*}
X_D&:=\  X_\text{c}+ X_\text{p}, \\
\text{ where } 
X_\text{c} &\sim \text{ discrete} : 
 d_{\min(X_\text{c})}>0, \\
X_\text{p} &\sim \text{ discrete} : %
 d_{\min(X_\text{p})}>0, \\
X_M&:=\  X_\text{c}+ X_\text{g}, 
\\
\text{ where } X_\text{g} &\sim \mathcal{N}(0,\mathbb{E}[|X_\text{g}|^2]) \text{ such that }  \mathbb{E}[|X_\text{p}|^2]=\mathbb{E}[|X_\text{g}|^2], 
\end{align*}
where $X_\text{c}, X_\text{g}$ and $X_\text{p}$ are mutually independent.
Then, for $Z_G\sim \mathcal{N}(0,1)$ independent of everything else, we have
\begin{align*}
I(X_D;gX_D+Z_G)-I(X_M;gX_M+Z_G) &\leq \frac{1}{2}\log(2), \\
I(X_M;gX_M+Z_G)-I(X_D;gX_D+Z_G) &\leq \frac{1}{2}\log \left( \frac{\pi \eu}{3}\right)+\frac{1}{2} \log \left(1+\frac{12}{g^2 \  d_{\min(X_{D})}^2}\right).
\end{align*}
\end{prop}
\begin{IEEEproof}
The first inequality follows since
\begin{align*}
I(X_D;gX_D+Z_G)
  &= I(X_\text{c}, X_\text{p};g X_\text{c}+g X_\text{p}+Z_G)
\\&= I(X_\text{p};g X_\text{c}+g X_\text{p}+Z_G)
   + I(X_\text{c};g X_\text{c}+g X_\text{p}+Z_G|X_\text{p})
%& = h\left(g X_\text{c}+g X_\text{p}+Z_G\right)-h(Z_G)\\
%&= h\left(g X_\text{c}+g X_\text{p}+Z_G\right)-h(Z_G)-h(g X_\text{c}+Z_G)+h(g X_\text{c}+Z_G)\\
\\&= I\left(X_\text{p};g X_\text{p}+N\right)|_{N:=g X_\text{c}+Z_G} +I(X_\text{c};g X_\text{c}+Z_G)
\\& \stackrel{\rm (a)}{\le}  I\left(X_\text{g};g X_\text{g}+N\right)|_{N:=g X_\text{c}+Z_G}
  +\frac{1}{2}\log(2)+I(X_\text{c};g X_\text{c}+Z_G)\\
&= I(X_M;gX_M+Z)+\frac{1}{2}\log(2),
\end{align*}
where in (a) we used~\cite[Theorem 1]{ZamirGaussianIsNotToBad}, which states that a Gaussian input for non-Gaussian additive noise channel results in at most $1/2$ bit loss. 

The second inequality follows since
\begin{align*}
I(X_M;gX_M+Z_G) 
  & \leq \mug(g^2 {\rm Var}[X_M]) = \mug(g^2 {\rm Var}[X_D]) %\stackrel{\rm (a)}{\leq} I(X_G;g X_G+Z_G)
\\& \stackrel{\rm (b)}{\leq} I(X_D;g X_D+Z_G)+\frac{1}{2} \log \left( \frac{\pi \eu}{3}\right)+\frac{1}{2} \log \left(1+\frac{12 }{g^2 \   d_{\min(X_{D})}^2}\right),
\end{align*}
where in (b) we used the bound in Proposition~\ref{prop:lowbound OW generalized}.
\end{IEEEproof}

The question left is thus why `two-layer' inputs, i.e., that comprise two random variables, are needed for approximate optimality in the moderately weak interference regime. 
Although at this point we do not have an answer for this question, the intuition for the moderately weak interference regime  is as follows.  With `single-layer' PAM inputs and for the given power constraints, the number of points needed to attain a desired rate pair on the convex closure of the outer bound  result in a minimum distance at the receivers that is too small.  It may be that with `two-layer' PAM inputs one effectively soft-estimates one of the layers whose effect can thus be removed from the received signal, thereby behaving as if there was an interfering common message jointly decoded at the non-intended receiver. Further investigation is needed to understand whether `multi-layer' inputs are indeed necessary.

\section{Conclusion} 
\label{sec:concl}

We evaluated a very simple, generally applicable lower bound, that neither requires joint decoding nor block synchronization, to the capacity of the Gaussian interference channel.

This treating-interference-as-noise lower bound without time-sharing was evaluated for inputs that are a mixture of discrete and Gaussian random variables. 
We showed that, through careful choice of the mixed input parameters, namely the number of points of the discrete part and the amount of power assigned to the Gaussian part (that in general depends on the channel gains and on which point on the convex closure of the outer bound one wants to attain)

the capacity of the classical Gaussian interference channel can be attained to within a gap.  
This result is of interest in several channels where this lower bound applies, such as block asynchronous channels and channels with partial codebook knowledge. 
Extension to other channel models and to more than two users are the subject of current investigation.

\section*{Acknowledgment}
The authors wish to thank Prof. Ramin Takloo-Bighash of the Department of Mathematics, Statistics and Computer Science at the University of Illinois at Chicago for numerous valuable discussions. 
The authors would also like to acknowledge interesting discussions on the meaning of ``treat interference as noise''  and its practical implications during ITW'2105 with Profs. Gerhard Kramer, Giuseppe Caire and  Shlomo Shamai.

\appendices

\section{Proof of~\eqref{eq:AD gen}}
\label{app:ADlowerbound}
To prove the lower bound in~\eqref{eq:AD gen} we first find a lower bound on the differential entropy of  $Y= X_D+Z_G$. To that end let $p_i:= \mathbb{P}[X_D=s_i], i\in[1:N]$, then $Y$ has the following Gaussian mixture density\begin{align}
Y \sim  P_{Y}(y) := \sum_{i\in[1:N]} p_i\mathcal{N}(y; s_i,1).
\label{eq:gauss mixture PY}
\end{align}
where
\[
\mathcal{N}(x; \mu,\sigma^2) := \frac{1}{\sqrt{2\pi\sigma^2}}\eu^{-\frac{(x-\mu)^2}{2\sigma^2}}, \quad x\in\mathbb{R}.
\]
We have
\begin{align*}
%h(Y)&=-\int P_{Y}(y) \log( P_{Y}(y) ) dy\\
-h(Y)&=\int P_{Y}(y) \log( P_{Y}(y) ) dy\\
&\stackrel{ \rm (a)}{ \le}\log \int P_{Y}(y) P_{Y}(y) dy\\
&  = \log \int \left( \sum_{i\in[1:N]} p_i\mathcal{N}(y;  s_i,1) \right)^2 dy\\
&  = \log  \left( \sum_{(i,j)\in[1:N]^2} p_i p_j \ \int \mathcal{N}(y;   s_i,1)  
 \mathcal{N}(y;    s_j,1) dy \right) \\
&  = \log  \left( \sum_{(i,j)\in[1:N]^2} p_i p_j \frac{1}{\sqrt{4\pi}} \eu^{\frac{- (s_i-s_j)^2}{4}}
 \int \mathcal{N}\left(y;    \frac{s_i+s_j}{2},\frac{1}{2}\right)  dy \right) \\
& \stackrel{ \rm (b)}{=} \log  \left( \sum_{(i,j)\in[1:N]^2} p_i p_j \frac{1}{\sqrt{4\pi}} \eu^{-\frac{(s_i-s_j)^2}{4}}\right)  \\
& \stackrel{ \rm (c)}{\le} \log  \left( 
 \sum_{i\in[1:N]} p_i^2 \frac{1}{\sqrt{4\pi}}
+\sum_{i\in[1:N]} p_i(1-p_i) \frac{1}{\sqrt{4\pi}} \eu^{-\frac{  d_{\min(X_D)}^2}{4}}\right) \\
& \stackrel{ \rm (d)}{\le} -\log(N\sqrt{4\pi}) + \log  \left( 
 1 + (N-1) \eu^{-\frac{  d_{\min(X_D)}^2}{4}}\right) ,
\end{align*}
which implies
\begin{align}
 I(X_D;  X_D+Z_G)&=h(X_D+Z_G)- h(Z_G) \geq \log(N)-{\gap}_{\text{\eqref{eq:gapAD}}},
\notag\\
{\gap}_{\text{\eqref{eq:gapAD}}} &:= \frac{1}{2}\log\left(\frac{\eu}{2}\right)
 + \log  \left( 1 + (N-1) \eu^{-\frac{ d_{\min(X_D)}^2}{4}}\right),
\label{eq:gapAD}
\end{align}
where the (in)equalities follow from:
(a) Jensen's inequality, 
(b) $\int \mathcal{N}(y;  \mu,\sigma^2)  dy=1$, 
(c) 
$d_{\min(X_D)} \leq |s_i-s_j|, \forall i\not=j$, 
(d) by maximizing over the $\{p_i, i\in[1:N]\}$. 
Combining this bound with the fact that mutual information is non-negative proves the claimed lower bound.

\section{Proof of Proposition~\ref{prop:card:dmin:measurebound}}
\label{app:Proof of Minimum Distance measure bound}
For convenience let $\mathcal{S}:=\supp(h_xX+h_yY)$. %and $s \in S$. 
To proof that $|S| = |X||Y|$ a.e. we look at the measure of the set such  that $|S| \neq |X||Y|$,  that is, a set for which there exists $s_i=h_xx_{i}+h_yy_{i}$ and $s_j=h_xx_{j}+h_yy_{j}$  such that $s_i=s_j$ for some $i \neq j$; 
hence, we are interested in characterizing the measure of the set   
\begin{align}
A := \Big\{(h_x,h_y) \in \mathbb{R}^2:  \begin{array}{c} h_xx_{i}+h_yy_{i}= h_xx_{j}+h_yy_{j}\\	 (x_{i},y_{i}) \neq (x_{j},y_{j}) \end{array},
 \forall  x_{i}, x_{j} \in X  \text{ and } 	\forall  y_{i}, y_{j} \in Y\, \,  \Big\}.
 \label{eq: Bad set}
\end{align}

Define 
\begin{align}
A(i,j)= \left\{(h_x,h_y) \in \mathbb{R}^2: h_xx_{i}+h_yy_{i}= h_xx_{j}+h_yy_{j}, \text{ s.t. } (x_{i},y_{i}) \neq (x_{j},y_{j}) \right \}.
\end{align}
By the sub-additivity of measure we have
\begin{align}
m(A)=m \left( \bigcup_{ i,j} A(i,j) \right) \leq \sum_{ i,j } m(A(i,j)) \label{eq:measure:|X||Y|}.
\end{align}
For fixed $x_{i},x_{j},y_{i},y_{j}$ the  set $A(i,j)$ is a line in $(h_x,h_y) \in \mathbb{R}^2$  and hence  
\[
m \left( A(i,j) \right)=0.
\] 
Thus, in~\eqref{eq:measure:|X||Y|} we have a countable sum of sets of measure zero,  which implies that  $m(A)=0$. 

Next, we bound the minimum distance %. .  We want to find
$d_{\min(\mathcal{S})}:=\min_{i\not=j} \{|s_i-s_j|: s_i,s_j \in \mathcal{S} \}$
with $|s_i-s_j|=|h_xx_{i} +h_y y_{i}-h_xx_{j}-h_y y_{j}|$. 
We distinguish two cases: 
\begin{enumerate}
\item[Case 1)] 
$x_{i}=x_{j}$ and  $y_{i}\neq y_{j}$, or  $x_{i}\neq x_{j}$ and  $y_{i}= y_{j}$: 
then trivially
\begin{align*}
&|s_i-s_j|
%&=|h_xx_{i} +h_y y_{i}-h_xx_{j}-h_y y_{j}|=|h_y y_{i}-h_y y_{j}|
%\\&=|h_y|d_{\min(Y)} \ |z_i-z_j| 
\geq |h_y| d_{\min(Y)}, \ \text{or}
%\end{align*}
%Case 2) is similar to Case 1 and 
\\
%\begin{align*}
&|s_i-s_j|
\geq |h_x| d_{\min(X)}.
\end{align*}
\item[Case 2)] 
$x_{i}\neq x_{j}$ and  $y_{i} \neq y_{j}$: Let $z_* \in \mathbb{Z}$, then 
\begin{align*}
|s_i-s_j|
&=|h_xx_{i} +h_y y_{i}-h_xx_{j}-h_y y_{j}|\\
&=|h_x (x_{i}- x_{j})-h_y(y_{j}-y_{i})|\\
&=|h_xd_{{\rm min}(X)} (z_{xi}- z_{xj})-h_yd_{{\rm min}(Y)}(z_{yj}-z_{yi})|\\ % 
&=|a z_x-bz_y|
\end{align*}
where
$a=h_xd_{{\rm min}(X)}$,
$b=h_yd_{{\rm min}(Y)}$,
$z_x=(z_{xi}-z_{xj})$ and 
$z_y=(z_{yj}-z_{yi})$.
Hence, by Lemma~\ref{lem:card:dmin:measurebound} in Appendix~\ref{app:ProofOfAuxiliaryLemma}  we have that 
\begin{align*}
|s_i-s_j| &\geq \gamma \max \left(\frac{|h_x|d_{{\rm min}(X)}}{ 2 |Y|(1+\log(|X|))},\frac{|h_y|d_{{\rm min}(Y)}}{ 2 |X|(1+\log(|Y|))} \right)\\
& \geq \kappa_{\gamma,|X|,|Y|} \max \left(\frac{|h_x|d_{{\rm min}(X)}}{  |Y|},\frac{|h_y|d_{{\rm min}(Y)}}{  |X|} \right)
\end{align*}
up to an outage set of measure $\gamma$ where $\kappa_{\gamma,|X|,|Y|}:=\frac{\gamma}{1+\ln (\max(|X|,|Y|)}$ and $\gamma \in (0,1]$.
Next, by taking the minimum over both cases we arrive at the result in Proposition~\ref{prop:card:dmin:measurebound}.
\end{enumerate}

\section{Minimum Distance Auxiliary Lemma}
\label{app:ProofOfAuxiliaryLemma}
\begin{lem}
\label{lem:card:dmin:measurebound}
Let $a,b \in \mathbb{R}$  and $z_x,z_y \in \mathbb{Z}$.
The function
\begin{align*}
f(z_x,z_y)= \min |az_x-bz_y|
\end{align*}
subject to the constrains 
\begin{align*}
 &z_x \in [-N_x:N_x] / \{0\},\\
 &z_y \in [-N_y:N_y] / \{0\},
\end{align*}
satisfies
\begin{align*}
f(z_x,z_y) \geq  \gamma \max \left(\frac{b}{ 2 N_x(1+\ln(N_y))},\frac{a}{ 2 N_y(1+\ln(N_x))} \right)
\end{align*}
for all $(a,b) \in \mathbb{R}^2$ except for an outage set of measure $\gamma$ for any $\gamma \in (0,1]$.
\end{lem}
\begin{IEEEproof}
First observe that w.l.o.g. we can assume that $a,b \in \mathbb{R}^{+}$ and $z_x \in [1:N_x]$ and $z_y \in [1:N_y]$.
This is because if ${\rm sign}(az_x) \neq {\rm sign}(bz_y)$ then the function is minimized by $|z_x|=1$ and $|z_y|=1$ 
and attains a value of $f=|a|+|b|$.
Define 
\begin{align*}
a_{xx}&=\frac{a}{\lceil a \rceil} \in [0, 1],\\
b_{yy}&=\frac{b}{\lceil b \rceil} \in [0, 1],\\
\bar{a}&=\lceil a \rceil \in \mathbb{N},\\
\bar{b}&=\lceil b \rceil \in \mathbb{N},
\end{align*}
and let 
\begin{align*}
A_{\epsilon}&= \left \{ (a_{xx},b_{yy}) \in [0,1]^2: \min_{1 \leq z_x \leq N_x, \, 1 \leq z_y \leq N_y} |a_{xx}\bar{a}z_x-b_{yy}\bar{b}z_y| > \epsilon \right\}\\
&= \bigcap_{1 \leq z_x \leq N_x, \, 1 \leq z_y \leq N_y} \left \{ (a_{xx},b_{yy}) \in [0,1]^2:  |a_{xx}\bar{a}z_x-b_{yy}\bar{b}z_y| > \epsilon \right\}\\
&=\bigcap_{1 \leq z_x \leq N_x, \, 1 \leq z_y \leq N_y} A(z_x,z_y),
\end{align*}
 where $A_{\epsilon}(z_x,z_y)=\left \{ (a_{xx},b_{yy}) \in [0,1]^2:  |a_{xx}\bar{a}z_x-b_{yy}\bar{b}z_y| > \epsilon \right\}$ and for some $\epsilon>0$.
% So,
The shape of $A_{\epsilon}(z_x,z_y)$ is shown on Fig.~\ref{fig:outageStrip}.
Let $A_\epsilon^c$ be the complement of $A_\epsilon$ where we have 
\begin{align}
 A^c_{\epsilon}&=\bigcup_{1 \leq z_x \leq N_x, \, 1 \leq z_y \leq N_y} A_\epsilon^c(z_x,z_y)
\end{align}
where  $A^c_{\epsilon}(z_x,z_y)=\left \{ (a_{xx},b_{yy}) \in [0,1]^2:  |a_{xx}\bar{a}z_x-b_{yy}\bar{b}z_y| \leq \epsilon \right\}$.
 
Next, we find the measure of the set   $ A_\epsilon^c$ as follows:
\begin{align*}
m(A^c_{\epsilon})
  &= m \left( \bigcup_{1 \leq z_x \leq N_x, \, 1 \leq z_y \leq N_y} A_\epsilon^c(z_x,z_y)\right) 
\\&\leq \sum_{1 \leq z_x \leq N_x, \, 1 \leq z_y \leq N_y} m( A_\epsilon^c(z_x,z_y))
\end{align*}
where to the inequality are  due to the sub-additive of measure.

\begin{figure}
\center
\includegraphics[width=12cm]{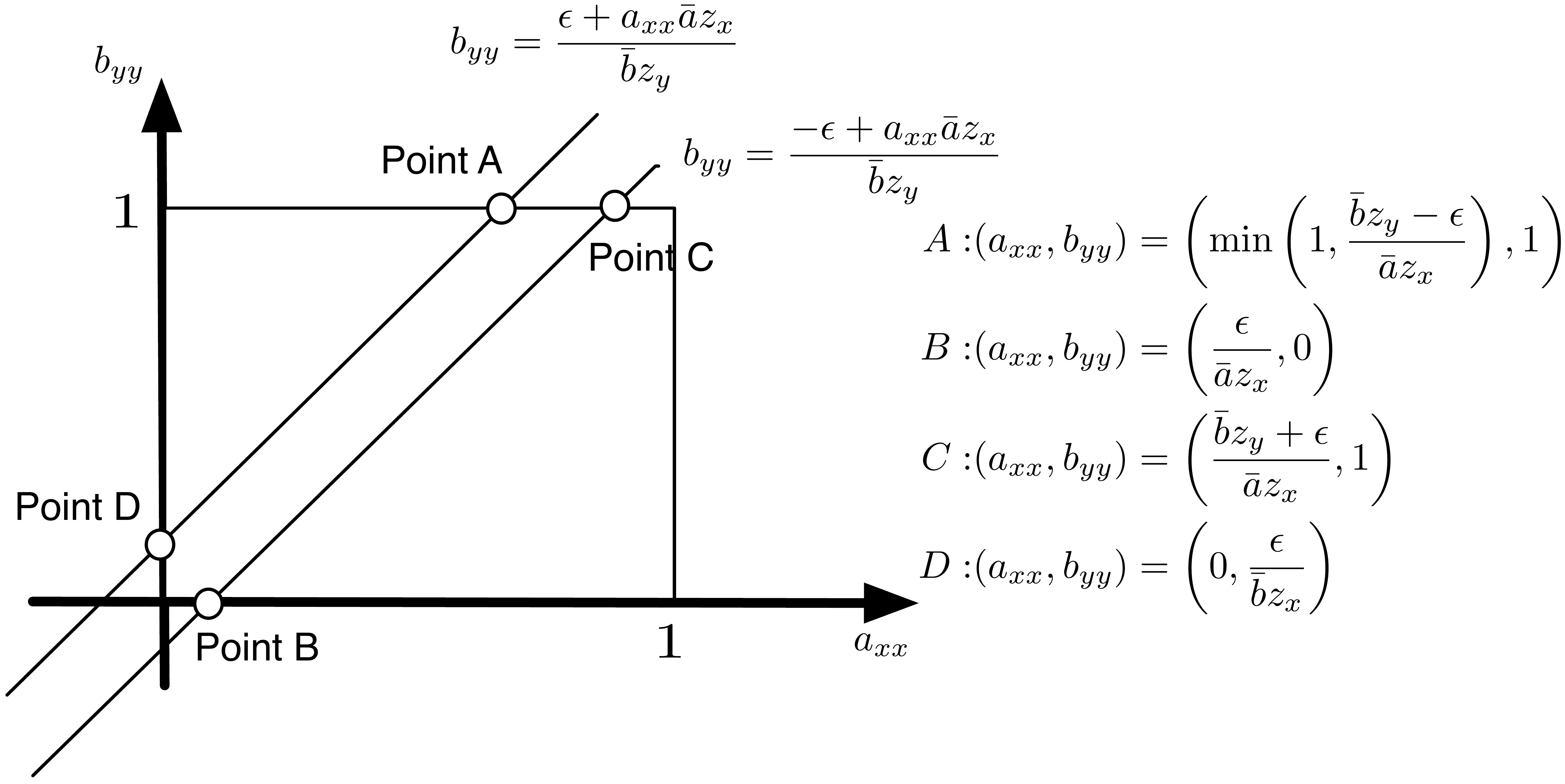}
\caption{Shape of the outage strip.}
\label{fig:outageStrip}
\end{figure}

Next, we compute $m(A_\epsilon^c(z_x,z_y))$ as follows
\begin{align*}
&m(A_{\epsilon}^c(z_x,z_y))
\\&=\int_{a_{xx}=0}^{a_{xx}=\text{Point} A} \frac{\epsilon+a_{xx}\bar{a}z_x}{b z_y}  da_{xx}-\int_{a_{xx}=\text{Point B}}^{a_{xx}=\text{Point A}}\frac{-\epsilon+a_{xx}\bar{a}z_x}{\bar{b} z_y}  da_{xx}\\
&=\int_{a_{xx}=0}^{a_{xx}=\min \left(1,\frac{\bar{b}z_y-\epsilon}{\bar{a}z_x} \right)} \frac{\epsilon+a_{xx}\bar{a}z_x}{b z_y}  da_{xx}-\int_{a_{xx}=\frac{\epsilon}{\bar{a}z_x}}^{a_{xx}=\min \left(1,\frac{\bar{b}z_y-\epsilon}{\bar{a}z_x} \right)}\frac{-\epsilon+a_{xx}\bar{a}z_x}{\bar{b} z_y}  da_{xx}\\
&=\left. \frac{{a_{xx}}\, \left(2\, \epsilon + \bar{a}\, {a_{xx}}\, {z_x}\right)}{2\, \bar{b}\, {z_y}} \right |^{a_{xx}=\min \left(1,\frac{\bar{b}z_y-\epsilon}{\bar{a}z_x} \right)}_{a_{xx}=0}-   \left( -\left. \frac{{a_{xx}}\, \left(2\, \epsilon - \bar{a}\, {a_{xx}}\, {z_x}\right)}{2\, \bar{b}\, {z_y}} \right |^{a_{xx}=\min \left(1,\frac{\bar{b}z_y-\epsilon}{\bar{a}z_x} \right)}_{a_{xx}=\frac{\epsilon}{\bar{a}z_x}}  \right)\\
&=\frac{\min\left(1,- \frac{\left(\epsilon - \bar{b}\, {z_y}\right)}{\bar{a}\, {z_x}}\right)\, \left(2\, \epsilon + \bar{a}\, {z_x}\, \min\left(1,- \frac{\left(\epsilon - \bar{b}\, {z_y}\right)}{a\, {z_x}}\right)\right)}{2\, \bar{b}\, {z_y}}-\frac{{\left(\epsilon - \bar{a}\, z_x\, \min\left(1,- \frac{\left(\epsilon - \bar{b}\, {z_y}\right)}{\bar{a}\, {z_x}}\right)\right)}^2}{2\, \bar{a}\, \bar{b}\, {z_x}\, {z_y}}\\
&=- \frac{\epsilon\, \left(\epsilon - 4\, \bar{a}\, {z_x}\, \min\left(1,- \frac{\left(\epsilon - \bar{b}\, {z_y}\right)}{\bar{a}\, {z_x}}\right)\right)}{2\, \bar{a}\, \bar{b}\, {z_x}\, {z_y}}.
\end{align*}
Next, compute $m(A_{\epsilon}^c)$ as follows
\begin{align*}
&m(A_{\epsilon}^c)
\\&= \sum_{1 \leq z_x \leq N_x, \, 1 \leq z_y \leq N_y} - \frac{\epsilon\, \left(\epsilon - 4\, \bar{a}\, {z_x}\, \min\left(1,- \frac{\left(\epsilon - \bar{b}\, {z_y}\right)}{\bar{a}\, {z_x}}\right)\right)}{2\, \bar{a}\, \bar{b}\, {z_x}\, {z_y}}\\
&=\sum_{1 \leq z_x \leq N_x, \, 1 \leq z_y \leq N_y} \frac{-\epsilon^2}{2\, \bar{a}\, \bar{b}\, {z_x}\, {z_y}}+\sum_{1 \leq z_x \leq N_x, \, 1 \leq z_y \leq N_y} \frac{4 \epsilon \, \bar{a}\, {z_x}\, \min\left(1,- \frac{\left(\epsilon - \bar{b}\, {z_y}\right)}{\bar{a}\, {z_x}}\right)}{2\, \bar{a}\, \bar{b}\, {z_x}\, {z_y}}\\
& \leq \sum_{1 \leq z_x \leq N_x, \, 1 \leq z_y \leq N_y} \frac{4 \epsilon \, \bar{a}\, {z_x}\, \min\left(1,- \frac{\left(\epsilon - \bar{b}\, {z_y}\right)}{\bar{a}\, {z_x}}\right)}{2\, \bar{a}\, \bar{b}\, {z_x}\, {z_y}}.
\end{align*}

The term $\min\left(1,- \frac{\left(\epsilon - \bar{b}\, {z_y}\right)}{\bar{a}\, {z_x}}\right)$ can be upper bounded in two different ways
\begin{align}
\min\left(1,- \frac{\left(\epsilon - \bar{b}\, {z_y}\right)}{\bar{a}\, {z_x}}\right) &\leq 1,
\label{eq:first bound on min}\\
\min\left(1,- \frac{\left(\epsilon - \bar{b}\, {z_y}\right)}{\bar{a}\, {z_x}}\right) &\leq - \frac{\left(\epsilon - \bar{b}\, {z_y}\right)}{\bar{a}\, {z_x}}.
\label{eq:second bound on min}
\end{align}
With the first upper bound in~\eqref{eq:first bound on min} we get
\begin{align}
m(A_{\epsilon}^c)
  &\leq \sum_{1 \leq z_x \leq N_x, \, 1 \leq z_y \leq N_y} \frac{4 \epsilon \, \bar{a}\, {z_x}\,}{2\, \bar{a}\, \bar{b}\, {z_x}\, {z_y}} \notag
\\&\leq \frac{2 \epsilon N_x (1+\ln(N_y))}{\bar{b}} \label{eq: fist bound on measure}
\end{align}
where for the last inequality we have used  $\sum_{zy=1}^{N_y} \frac{1}{z_y} \leq 1+\ln(N_y)$.
With the second upper bound in~\eqref{eq:second bound on min} we get 
\begin{align}
m(A_{\epsilon}^c)
  &\leq \sum_{1 \leq z_x \leq N_x, \, 1 \leq z_y \leq N_y} \frac{4 \epsilon \,  (\bar{b}z_y-\epsilon)\,}{2\, \bar{a}\, \bar{b}\, {z_x}\, {z_y}}  \notag
\\&\leq \sum_{1 \leq z_x \leq N_x, \, 1 \leq z_y \leq N_y} \frac{4 \epsilon \,  }{2\, \bar{a}\,  {z_x}} \notag
\\&\leq \frac{2 \epsilon N_y (1+\ln(N_x))}{\bar{a}}.  \label{eq: second bound on measure}
\end{align}
So by taking the tightest of the two bounds  in~\eqref{eq: fist bound on measure} and in~\eqref{eq: second bound on measure} we get 
\begin{align*}
m(A_{\epsilon}^c)& \leq \min \left(\frac{2 \epsilon N_x (1+\ln(N_y))}{\bar{b}} ,\frac{2 \epsilon N_y (1+\ln(N_x))}{\bar{a}} \right).
\end{align*}

Now let $m(A_{\epsilon}^c)=\gamma$ for some $\gamma \in [0,1]$ then we have that
\begin{align*}
\gamma \leq \epsilon \min \left(\frac{ 2 N_x(1+\ln(N_y))}{\bar{b}},\frac{ 2 N_y(1+\ln(N_x))}{\bar{a}} \right).
\end{align*}
Next, by solving for $\epsilon$  in terms of measure of the outage,
\begin{align*}
\epsilon 
  &\geq  \frac{\gamma}{\min \left(\frac{ 2 N_x(1+\ln(N_y))}{\bar{b}},\frac{ 2 N_y(1+\ln(N_x))}{\bar{a}} \right) }
\\&=\gamma \max \left(\frac{\bar{b}}{ 2 N_x(1+\ln(N_y))},\frac{\bar{a}}{ 2 N_y(1+\ln(N_x))} \right)
\\& \geq  \gamma \max \left(\frac{b_{yy}\bar{b}}{ 2 N_x(1+\ln(N_y))},\frac{a_{xx}\bar{a}}{ 2 N_y(1+\ln(N_x))} \right)
\\&=\gamma \max \left(\frac{b}{ 2 N_x(1+\ln(N_y))},\frac{a}{ 2 N_y(1+\ln(N_x))} \right).
\end{align*}
This concludes the proof.
\end{IEEEproof}

% ============================
\section{Gap for $\inr \leq \snr \leq 1+\inr$}
\label{sec:gaps s=i TDMA}

\paragraph*{Outer Bound for $\inr \leq \snr \leq 1+\inr$}
It is well know that when $\snr\approx\inr$ time-division is approximately optimal~\cite{etkin_tse_wang}.
In this regime we outer bound the capacity region by the sum-rate constraint in~\eqref{eq:R upper classical IC kra1} only, which in the symmetric case is
\begin{align*}
R_1 + R_2 
  &\leq \mug\left(\snr\right)-\mug\left(\inr\right)+\mug\left(\snr+\inr\right)
\\&=    \mug\left(\snr\right)+\mug\left(\frac{\snr}{1+\inr}\right)
\\&\leq \mug\left(\snr\right)+\frac{1}{2}\log(2),
\end{align*}
that is
\begin{align*}
\mathcal{R} _{\text{out}}^{(\text{\ref{sec:gaps s=i TDMA}})}
&= \bigcup_{t \in [0,1]} \left \{ \begin{array}{l}
R_1 \leq t       \left(\mug\left(\snr\right)+\frac{1}{2}\log(2)\right) \\
R_2 \leq + (1-t) \left(\mug\left(\snr\right)+\frac{1}{2}\log(2)\right) \\
\end{array} \right\}.
\end{align*}

\paragraph*{Inner Bound for $\inr \leq \snr \leq 1+\inr$}
We only use the discrete part of the mixed inputs and set
\begin{subequations}
\begin{align}
N_1 &= \points\left(\snr_{1,t}\right), \
\snr_{1,t} := \left(1+\snr\right)^{t}-1 \leq \snr,
\\
N_2 &= \points\left(\snr_{2,t}\right), \
\snr_{2,t} := \left(1+\snr\right)^{1-t}-1 \leq \snr,
\\
\delta_1&=0,
\\
\delta_2&=0.
\end{align}
Note that 
\begin{align}
N_1^2N_2^2-1 \leq (1+\snr_{1,t})(1+\snr_{2,t})-1 = \snr.
\end{align}
\label{eq:all N bounds for snr=inr}
\end{subequations}

We lower bound the minimum distance of the sum-set constellations as in~\eqref{eq:achregion for par:strong mindmin} and we get

\begin{align*}
\min_{i\in[1:2]}\frac{d_{\min ({S}_i)}^2}{12} 
&\geq   \kappa_{\gamma,N_1,N_2}^2 \min\left( 
\frac{\min(\snr,\inr)}{\max(N_1^2,N_2^2)-1},
\frac{\max(\snr,\inr)}{N_1^2N_2^2-1} \right) 
\\& \stackrel{\text{for $\inr \leq \snr$ and~\eqref{eq:all N bounds for snr=inr}}}{\geq}  
\kappa_{\gamma,N_1,N_2}^2 \min\left( 
\frac{\inr}{\snr},
\frac{\snr}{\snr} \right)
\\& \stackrel{\text{for $\snr \leq 1+\inr$}}{\geq}   
\kappa_{\gamma,N_1,N_2}^2 \min\left( 
\frac{\inr}{1+\inr},
1 \right)
\\& = 
\kappa_{\gamma,N_1,N_2}^2 
\frac{\inr}{1+\inr}
\\& \stackrel{\text{for $1 \leq \inr$}}{\geq} 
\kappa_{\gamma,N_1,N_2}^2 
\frac{1}{2}.
\end{align*}

\paragraph*{Gap for $\inr \leq \snr \leq 1+\inr$}
Similarly to the strong interference regime, we can upper bound the difference between the upper and lower bounds as
\begin{align*}
\gap 
&\leq \max \left(
  \mug\left(\snr_{1,t}\right) +\frac{t}{2}\log(2)  - \log(\points\left(\snr_{1,t}\right)), \
  \mug\left(\snr_{2,t}\right) +\frac{1-t}{2}\log(2)- \log(\points\left(\snr_{2,t}\right))
\right) 
\\&\quad
+ \frac{1}{2}\log\left(\frac{2\pi\eu}{12}\right) 
+ \frac{1}{2}\log\left(1+\frac{2}{\kappa_{\gamma,N_1,N_2}^2} \right) 
\\&\leq
  \frac{1}{2}\log\left(\frac{4\pi\eu}{3}\right) 
+ \frac{1}{2}\log\left(1+8 \cdot \frac{\left(1+1/2\ln({1+\snr})\right)^2}{\gamma^2}\right).
\end{align*}

\section{Auxiliary Results for Regime Weak1}
\label{sec:gaps weak1}
We derive here some auxiliary results for the regime in~\eqref{eq: condition for Weak1}, namely
\begin{align*}
&(1+\inr) \leq \snr \leq \inr(1+\inr),
\\
&\frac{1+\snr}{1+\inr+\frac{\snr}{1+\inr}} \leq \frac{1+\inr+\frac{\snr}{1+\inr}}{1+\frac{\snr}{1+\inr}}.
\end{align*}

\subsection{Derivation of~\eqref{eq:sec:gaps:WeakType1:R1+R2 face}}
\label{sec:gaps:WeakType1:R1+R2 face}

The parameters of the mixed inputs are given in~\eqref{eq:achregion for par:moderate choiceALL: R1+R2sumrate}. %sumrate
We aim to derive bounds on $\max(N_1^2,N_2^2)$ and $N_1^2N_2^2$ and used them 
to find the lower bound  on minimum distance in~\eqref{eq:sec:gaps:WeakType1:R1+R2 face}.

The mixed input parameters are given in~\eqref{eq:achregion for par:moderate choiceALL: R1+R2sumrate}.
We have
\begin{align}
\max(N_1^2,N_2^2)-1 
  &\leq \max(\snr_{1,a,t},\snr_{1,b,t}) 
\notag\\&
\leq \max \left(\frac{1+\inr +\frac{\snr}{1+\inr}}{1+\frac{\snr}{1+\inr}},\frac{1+\snr }{1+\inr+\frac{\snr}{1+\inr}}\right) -1
\notag\\&
\stackrel{\text{from~\eqref{eq: condition for Weak1}}}{\leq}
\frac{1+\inr +\frac{\snr}{1+\inr}}{1+\frac{\snr}{1+\inr}}-1
\notag\\&
= \frac{\inr}{1+\frac{\snr}{1+\inr}}
\label{eq: max(N1,N2)bound {sec:gaps:WeakType1:R1+R2 face}}
\\& \leq \inr = \min(\snr,\inr),
\notag
\end{align}
and
\begin{align}
N_1^2N_2^2-1 
&\leq (\snr_{1,a,t}+1)(\snr_{1,b,t}+1)-1
\notag\\&
\stackrel{\text{from~\eqref{eq:achregion for par:moderate choiceALL: R1+R2sumrate}}}{\leq}
\left( \frac{1+\inr +\frac{\snr}{1+\inr}}{1+\frac{\snr}{1+\inr}}\right)\left( \frac{1+\snr}{1+\inr+\frac{\snr}{1+\inr}} \right)-1
\notag\\&
= \frac{1+\snr}{1+\frac{\snr}{1+\inr}} - 1
\notag\\&
=\frac{\inr\frac{\snr}{1+\inr}}{1+\frac{\snr}{1+\inr}}
%\\&= \frac{\snr\inr}{1+\inr+\snr}.
\label{eq: N1*N2bound {sec:gaps:WeakType1:R1+R2 face}}
\\&
\leq \inr. 
\notag
\end{align}

Recall the definition of $\kappa_{\gamma,N_1,N_2}$ in~\eqref{eq:achregion for par:moderate kappakappa definitiontobeusedlater}.
By plugging the bounds in~\eqref{eq: max(N1,N2)bound {sec:gaps:WeakType1:R1+R2 face}}-\eqref{eq: N1*N2bound {sec:gaps:WeakType1:R1+R2 face}}
into~\eqref{eq:achregion for par:moderate kappakappa definitiontobeusedlater tuttituttissimi} we get
\begin{align}
\min_{i \in[1:2]} \frac{d_{\min ({S}_i)}^2}{12 \ \kappa_{\gamma,N_1,N_2}^2 }
&\geq 
\frac{1-\max(\delta_1,\delta_2)}{1+(\snr+\inr)\max(\delta_1,\delta_2)}
\min\left(
\frac{\inr}{\max(N_1^2,N_2^2)-1}, 
\frac{\snr}{N_1^2 N_2^2-1}
\right)
\notag\\
&\geq %
\frac{\inr}{1+\snr+2\inr}
\min\left(
\frac{\inr(1+\snr+\inr)}{\inr(1+\inr)}, 
\frac{\snr(1+\inr+\snr)}{\snr\inr}
\right)
\notag\\
&=
\frac{1+\snr+\inr}{1+\snr+2\inr}\cdot\frac{\inr}{1+\inr}
\notag\\
&\stackrel{1\leq \inr\leq\snr}{\geq}
\frac{1+2\inr}{1+3\inr}\cdot\frac{\inr}{1+\inr}
\geq \frac{3}{8}.
\label{eq: min(dmin)bound {sec:gaps:WeakType1:R1+R2 face}}
\end{align}

Note that the above derivation assumes $1\leq \inr$; this restriction is without loss of generality since for 
$\inr \leq 1$ TIN with Gaussian codebooks is optimal to within $1/2$~bit~\cite{etkin_tse_wang}.
Note also that the minimum distance lower bound holds up to an outage set of measure less than $\gamma$,
where $\gamma$ is a tunable parameter;
the reason why we need an outage set in this regime is the same as in Remark~\ref{rem:why cannot use Prop nooutage}.

\subsection{Derivation of~\eqref{eq:sec:gaps:WeakType1:2R1+R2 face}}
\label{sec:gaps:WeakType1:2R1+R2 face}

We aim to derive different bounds involving $N_1^2$ and $N_2^2$ and used them 
in the minimum distance lower bound in~\eqref{eq:achregion for par:moderate kappakappa definitiontobeusedlater tuttituttissimi}.

From~\eqref{eq:achregion for par:moderate choiceN1: 2R1+R2} we have
\begin{align}
N_1^2 -1 
&\leq \snr_{2,a,t} \leq \max(\snr_{2,a,0},\snr_{2,a,1})
\notag\\&
\leq 
\frac{\max\left(1+\inr+\frac{\snr}{1+\inr}, 1+\snr\right)}{1+\frac{\snr}{1+\inr}}-1
\notag\\&
\stackrel{\text{for $1+\inr \leq \snr$ }}{\leq}
\frac{1+\snr}{1+\frac{\snr}{1+\inr}}-1
\notag\\&
= \inr \cdot \frac{\frac{\snr}{1+\inr}}{1+\frac{\snr}{1+\inr}}
\leq \inr;
\label{eq: N1bound {sec:gaps:WeakType1:2R1+R2 face}}
\end{align}
from~\eqref{eq:achregion for par:moderate choiceN2: 2R1+R2} we have
\begin{align}
N_2^2 -1 
&\leq \snr_{2,b,t} \leq \max(\snr_{2,b,0},\snr_{2,b,1})
\notag\\&
\leq %
\frac{1+\inr+\snr}{1+\inr+\frac{\snr}{1+\inr}}
%,
-1
\notag\\&
= \frac{\inr \cdot \frac{\snr}{1+\inr}}{1+\inr+\frac{\snr}{1+\inr}}
\notag\\&
\leq \min\left( \inr, \frac{\snr}{1+\inr} \right) 
=\frac{\snr}{1+\inr},
\label{eq: N2bound {sec:gaps:WeakType1:2R1+R2 face}}
\end{align}
finally
\begin{align}
\max(N_1^2,N_2^2)-1\leq \max\left( \inr, \frac{\snr}{1+\inr} \right) = \inr = \min(\snr,\inr).
\label{eq: max(N1,N2)bound {sec:gaps:WeakType1:2R1+R2 face}}
\end{align}
We also have
\begin{align}
N_1^2N_2^2-1 
  &\leq (\snr_{2,a,t}+1)(\snr_{2,b,t}+1)-1
\notag\\& = \frac{ \left(1+\snr\right)^{1-t} \left(  1+\inr+\snr\right)^{t} }{1+\frac{\snr}{1+\inr}}
        - 1
\notag\\&\leq \frac{ 1+\inr+\snr }{1+\frac{\snr}{1+\inr}}
        - 1 = \inr.
\label{eq: N1*N2bound {sec:gaps:WeakType1:2R1+R2 face}}
\end{align}

In this regime, as we shall soon see, it also important to bound
\begin{align}
(1+\snr \delta_2)(1+\snr_{2,a,t})
  &= \frac{1+\snr}{1+\frac{\snr}{1+\inr}}
  \left(  \frac{1+\inr+\frac{\snr}{1+\inr}}{1+\inr} \right)^{t}
\notag\\&
\leq \frac{1+\snr}{1+\snr+\inr}
 \left(1+\inr+\frac{\snr}{1+\inr}\right).%
\label{eq: N1*d2bound {sec:gaps:WeakType1:2R1+R2 face}}
\end{align}

We next bound the minimum distances. 
Recall that $\kappa_{\gamma,N_1,N_2}$ is given in~\eqref{eq:achregion for par:moderate kappakappa definitiontobeusedlater}.

With~\eqref{eq:achregion for par:moderate dminS1lowerbound} we have
\begin{align}
\frac{d_{\min ({S}_1)}^2}{12 \  \kappa_{\gamma,N_1,N_2}^2} 
  &\geq \frac{1-\max(\delta_1,\delta_2)} {1+\snr \delta_1+\inr \delta_2} 
  \min\left( \frac{\inr}{N_2^2-1},\frac{\snr}{N_1^2N_2^2-1} \right)
\notag\\&
 \stackrel{\rm (a)}{\ge} \frac{\inr }{1+\snr +2\inr }
  \min\left( \frac{\inr}{N_2^2-1},\frac{\snr}{N_1^2N_2^2-1} \right)
\notag\\&
 \stackrel{\rm (b)}{\ge} \frac{\inr }{1+\snr +2\inr }
\min\left( \frac{\inr\left(1+\inr\right) }{\snr},\frac{\snr}{\inr} \right)
\notag\\&
=
\min\left( \frac{\inr^2\left(1+\inr\right) }{\snr(1+\snr +2\inr)},\frac{\snr}{1+\snr +2\inr} \right)
\notag\\&
 \stackrel{\rm (c)}{\ge} \min\left( \frac{\inr^2\left(1+\inr\right) \left(1+\frac{\snr}{1+\inr} \right) }{(1+\snr +2\inr) \left(1+\inr+\frac{\snr}{1+\inr} \right)^2},\frac{\snr}{1+\snr +2\inr} \right)
\notag\\&
  \stackrel{\rm (d)}{\ge} \min\left( \frac{\inr^2 \left(1+\inr+\snr \right) }{(1+\snr +2\inr) {\left(1+2\inr \right)^2}},
  {\frac{1+\inr}{2+3\inr}} %\
  \right)
\notag\\&
  \stackrel{\rm (e)}{\ge} {\min\left( \frac{2\inr^2 \left(1+\inr\right) }{(2+3\inr) \left(1+2\inr \right)^2},
  \frac{1+\inr}{2+3\inr} %\
  \right)}
\notag\\&
  \stackrel{\rm (f)}{\ge} {\min \left(\frac{4}{45}, \frac{1}{3} \right) = \frac{4}{45}},
\label{eq: dmin(S1)bound {sec:gaps:WeakType1:2R1+R2 face}}
\end{align}
where the inequalities follow since:
(a) $\delta_2\leq \delta_1 = \frac{1}{1+\inr}$,
(b) from~\eqref{eq: N2bound {sec:gaps:WeakType1:2R1+R2 face}} and~\eqref{eq: N1*N2bound {sec:gaps:WeakType1:2R1+R2 face}}, 
(c) from \eqref{eq: condition for Weak2} we have $\snr \le 1+\snr \le  \frac{\left(1+\inr+\frac{\snr}{1+\inr}\right)^2}{1+\frac{\snr}{1+\inr}}$, and
(d) where we have used $1+\inr \le \snr, \ \frac{\snr}{1+\inr} \leq \inr$ 
for 1st term must use largest $\frac{\snr}{1+\inr}$ while for 2nd smallest $\snr$ which $1+\inr$,
(e) since $1+\inr \le \snr$,
(f) comes from using $1 \le \inr$.

With~\eqref{eq:achregion for par:moderate dminS2lowerbound}, and recalling that $1+\snr\delta_2 =  \left(\frac{1+\snr}{1+\inr}\right)^t$ from~\eqref{eq: choice of delta2: moderate: 2R1+R2} and  $\snr_{2,a,t}$ in~\eqref{eq:achregion for par:moderate choiceN1: 2R1+R2}, we have
\begin{align}
\frac{d_{\min ({S}_2)}^2}{12 \  \kappa_{\gamma,N_1,N_2}^2}
  &\geq \frac{1-\max(\delta_1,\delta_2)} {1+\snr \delta_2+\inr \delta_1} 
  \min\left( \frac{\inr}{N_1^2-1},\frac{\snr}{N_1^2N_2^2-1} \right)
\notag\\&
 \stackrel{\rm (a)}{\geq} \frac{ \frac{\inr}{1+\inr} } {1+\snr\delta_2 + \frac{\inr}{1+\inr}}
  \min\left( \frac{\inr}{\snr_{2,a,t}},\frac{\snr}{\inr} \right)
\notag\\&
  \stackrel{\rm (b)}{\geq}
  \min\left( 
 \frac{ \frac{\inr^2}{1+\inr} } {2(1+\snr\delta_2)(1+\snr_{2,a,t})},
 \frac{ \frac{\snr}{1+\inr}   } {1+\frac{\snr}{1+\inr} + \frac{\inr}{1+\inr}}
 \right)
\notag\\&
  \stackrel{\rm (c)}{\geq} %
  \min\left( 
 \frac{ \frac{\inr^2}{1+\inr} (1+\snr+\inr)} {(1+\snr)\left(1+\inr+\frac{\snr}{1+\inr}\right)},
 \frac{\snr}{1+\snr+2\inr}
 \right)
\notag\\&
   \stackrel{\rm (d)}{\geq}
  \min\left( 
 \frac{ \inr^2 (1+\snr+\inr)} {(1+\inr)(1+\snr)\left(1+2\inr\right)},
 \frac{\snr}{1+\snr+2\inr}
 \right) % 
\notag\\&
 \stackrel{\rm (e)}{\geq}
  \min\left( 
 \frac{ \inr^2 (1+\inr)^2} {(1+\inr)(1+\inr+\inr^2)\left(1+2\inr\right)},
 \frac{1+\inr}{2+3\inr}
 \right) %
\notag\\&
 \stackrel{\rm (f)}{\geq}
\min \left(\frac{2}{9}, \frac{1}{3} \right) = \frac{2}{9} 
\label{eq: dmin(S2)bound {sec:gaps:WeakType1:2R1+R2 face}}
\end{align}
where the inequalities follow since:
(a) $\delta_2\leq \delta_1 = \frac{1}{1+\inr}$, $N_1^2-1\leq \snr_{2,a,t}$ and~\eqref{eq: N1*N2bound {sec:gaps:WeakType1:2R1+R2 face}},
(b) $(1+\snr\delta_2) + \frac{\inr}{1+\inr} \leq 2+\snr\delta_2 \leq 2(1+\snr\delta_2)$ and $\delta_2\leq \frac{1}{1+\inr}$, and
the rest of the inequalities from the definition of weak interference $1\leq \inr, \ 1+\inr \leq \snr \leq \inr(1+\inr)$, d) since $ \frac{\snr}{1+\inr}\leq\inr$, e) where we have used $1+\inr \le \snr, \ \frac{\snr}{1+\inr} \leq \inr$  for 1st term must use largest $\frac{\snr}{1+\inr}$ while for 2nd smallest $\snr$ which $1+\inr$, f) comes from using $1 \le \inr$.

By putting together~\eqref{eq: dmin(S1)bound {sec:gaps:WeakType1:2R1+R2 face}} and~\eqref{eq: dmin(S2)bound {sec:gaps:WeakType1:2R1+R2 face}}, we obtain~\eqref{eq:sec:gaps:WeakType1:2R1+R2 face}. %

\section{Auxiliary Results for Regime Weak2}
\label{sec:gaps weak2}
We derive here some auxiliary results for the regime in~\eqref{eq: condition for Weak2}, namely
\begin{align*}
&(1+\inr) \leq \snr \leq \inr(1+\inr),
\\
&\frac{1+\snr}{1+\inr+\frac{\snr}{1+\inr}} \geq \frac{1+\inr+\frac{\snr}{1+\inr}}{1+\frac{\snr}{1+\inr}}.
\end{align*}

\subsection{Derivation of~\eqref{eq: dmin WeakType2:R1+R2 face take1}}
\label{sec:gaps:WeakType2:R1+R2 face}

We aim to derive different bounds on $N_1^2,N_2^2,\delta_1$ and $\delta_1$ so as to obtaine in the minimum distance lower bound in~\eqref{eq: dmin WeakType2:R1+R2 face take1}.

The mixed input parameters are in~\eqref{eq:parameters R1+R2 face weak type 2 take1}.
We have
\begin{align}
N_1^2-1&\leq\snr_{3,a,t}\leq \max(\snr_{3,a,0}, \snr_{3,a,1} )
\notag\\&
=\max \left( 
\frac{ \left(1+\frac{\snr}{1+\inr} \right)(1+\snr)}{1+\inr +\frac{\snr}{1+\inr}}, \
\frac{ \left(1+\inr+\frac{\snr}{1+\inr} \right)^3}{ \left(1+\frac{\snr}{1+\inr} \right)(1+\snr)}
     \right)^{\frac{1}{2}} -1
\notag\\&
\stackrel{\text{from~\eqref{eq: condition for Weak2}}}{=} 
\left(1+\frac{\snr}{1+\inr} \right)
\sqrt{ \frac{ 1+\snr}{\left(1+\inr +\frac{\snr}{1+\inr}\right)\left(1+\frac{\snr}{1+\inr} \right)} } -1
\notag\\&\leq \frac{\snr}{1+\inr}.
\label{eq: N1bound {sec:gaps:WeakType2:R1+R2 face}}
\end{align}
Similarly we have 
\begin{align}
 N_2^2-1
 & \leq
\snr_{3,b,t} \leq \frac{\snr}{1+\inr}.
\label{eq: N2bound {sec:gaps:WeakType2:R1+R2 face}}
\end{align}
The bounds in~\eqref{eq: N1bound {sec:gaps:WeakType2:R1+R2 face}}-\eqref{eq: N2bound {sec:gaps:WeakType2:R1+R2 face}} imply
\begin{align}
\max(\delta_1,\delta_2) \leq \frac{\max(\snr_{3,a,t},\snr_{3,b,t})}{\snr} \leq \frac{1}{1+\inr}.
\label{eq: dsbound {sec:gaps:WeakType2:R1+R2 face}}
\end{align}
Finally we have
\begin{align}
\max(N_1^2,N_2^2)-1  \leq
\frac{\snr}{1+\inr} \leq \inr = \min(\snr,\inr),
\end{align}
by the definition of this regime.

We also have 
\begin{align}
N_1^2N_2^2-1
        & \leq (1+\snr_{3,a,t})(1+\snr_{3,b,t})-1
\notag\\&=     (1+\snr \delta_1)(1+\snr_{3,b,t})-1
\notag\\&=     (1+\snr_{3,a,t})(1+\snr \delta_2)-1
\notag\\&= \inr+\frac{\snr}{1+\inr} \leq 2\inr.
\label{eq: N1*N2 {sec:gaps:WeakType2:R1+R2 face}}
\end{align}

With~\eqref{eq:achregion for par:moderate dminS1lowerbound} we have
\begin{align}
 \frac{d_{\min ({S}_1)}^2}{12 \ \kappa_{\gamma,N_1,N_2}^2 }
&\geq \frac{1-\max(\delta_1,\delta_2)}{1+\snr \delta_1+\inr \delta_2}
\min\left(
\frac{\inr}{N_2^2-1}, 
\frac{\snr}{N_1^2 N_2^2-1}
\right)
\notag\\
&\stackrel{\text{(a)}}{\geq}
\frac{ \frac{\inr}{1+\inr}}{1+\snr_{3,a,t}+\frac{\inr}{1+\inr}}
\min\left(
\frac{\inr}{\snr_{3,b,t}}, 
\frac{\snr}{2 \inr }
\right)
\notag\\
&\stackrel{\text{(b)}}{\geq} 
\min\left(
\frac{\frac{\inr^2}{1+\inr}}{2 \left( 1 +\inr+\frac{\snr}{1+\inr}\right)}, 
\frac{\snr \frac{\inr}{1+\inr}}{(1+\frac{\snr}{1+\inr}+\frac{\inr}{1+\inr})2 \inr }
\right)
\notag\\
&=
\min\left(
\frac{\inr^2}{2 \left[(1 +\inr)^2+\snr\right]}, 
\frac{\snr  }{2 \left[ 1+2\inr+\snr \right]}
\right)
\notag\\
&\stackrel{\text{(c)}}{\geq} 
\min\left(
\frac{\inr^2}{2 \left[(1 +\inr)^2+\inr (1+\inr)\right]}, 
\frac{\inr(1+\inr)}{2 \left[ 1+2\inr+\inr(1+\inr) \right]}
\right)
\notag\\
&\stackrel{\text{(c)}}{\geq}
\min\left(
\frac{\inr^2}{2 \left[(1 +\inr)^2+\inr(1+\inr)\right]}, 
\frac{1+\inr}{2 \left[ 2+3\inr \right]}
\right)
\notag\\
&= 
\min\left(
\frac{1}{12}, 
\frac{1}{6}
\right)
= \frac{1}{12},
\end{align}
where the inequalities follow from: 
(a) using~\eqref{eq: choice delta1: weak 2 R1+R2 take1},~\eqref{eq: dsbound {sec:gaps:WeakType2:R1+R2 face}} and~\eqref{eq: N1*N2 {sec:gaps:WeakType2:R1+R2 face}},
(b) using~\eqref{eq: N1*N2 {sec:gaps:WeakType2:R1+R2 face}}, and
(c) since $1\le \inr$ and $1+\inr \le \snr \leq\inr(1+\inr)$. 

By symmetry an equivalent bound can be derived for $d_{\min ({S}_2)}^2$.

Hence minimum distance in \eqref{eq: dmin WeakType2:R1+R2 face take1} is bounded by
\begin{align}
\min_{i\in[1:2]} \frac{d_{\min ({S}_i)}^2}{12} \geq \kappa_{\gamma,N_1,N_2}^2 \frac{1}{12}.
\end{align}

\subsection{Proof of~\eqref{eq: min  WeakType2:2R1+R2 face with loglog gap}}
\label{sec:gaps:WeakType2:2R1+R2 face}

We first derive some bounds on $N_1^2$ and $N_2^2$ that will be useful in bounding minimum distance of the received constellations.

From~\eqref{eq:parameters 2R1+R2 face weak type 2: first choice N1} we have
\begin{align}
N_1^2-1 &\leq \snr_{4,a,t} \leq \max(\snr_{4,a,0},\snr_{4,a,1})
\notag\\&= \frac{1+\snr}{\min\left( 1+\frac{\snr}{1+\inr},  1+\inr+\frac{\snr}{1+\inr}\right)}-1
\notag\\&=    \frac{1+\snr}{1+\frac{\snr}{1+\inr}}-1
\notag\\&=\inr \cdot \frac{\frac{\snr}{1+\inr}}{1+\frac{\snr}{1+\inr}} \leq \inr.
\label{eq: N1bound {sec:gaps:WeakType2:2R1+R2 face}}
\end{align}
Similarly form~\eqref{eq:parameters 2R1+R2 face weak type 2: first choice N2}
\begin{align}
N_2^2-1 &\leq \snr_{4,b,t} \leq \max(\snr_{4,b,0},\snr_{4,b,1})
\notag\\&= \frac{(1+\inr +\frac{\snr}{1+\inr})^{2} }{1+\snr} - 1 
\notag\\&\leq \frac{(1+\snr)(1+\frac{\snr}{1+\inr}) }{1+\snr} - 1
            = \frac{\snr}{1+\inr},
\label{eq: N2bound {sec:gaps:WeakType2:2R1+R2 face}}
\end{align}
where inequality follow the definition of the regime in \eqref{eq: condition for Weak2}.
We also have
\begin{align}
N_1^2N_2^2-1 
   &\leq (1+\snr_{4,a,t})(1+\snr_{4,b,t})-1
\notag\\&= \left(1+\inr +\frac{\snr}{1+\inr}\right)^{1-t} 
          \left(\frac{1+\snr }{1+\frac{\snr}{1+\inr}} \right)^{t} 
-1
\notag\\&\leq \max\left(\inr +\frac{\snr}{1+\inr}, \inr  \frac{\frac{\snr}{1+\inr}}{1+\frac{\snr}{1+\inr}}  \right)
\notag\\&= \inr +\frac{\snr}{1+\inr}\leq 2\inr,
\label{eq: N1*N2bound {sec:gaps:WeakType2:2R1+R2 face}}
\end{align}
where the last inequality follows from $\frac{\snr}{1+\inr}\le \inr$.

From~\eqref{eq:parameters 2R1+R2 face weak type 2: first choice delta2} we have
\begin{align}
\snr\delta_2 
        &\leq \frac{1+\inr+\frac{\snr}{1+\inr}}{1+\frac{\snr}{1+\inr} }  -1
\notag\\&\stackrel{\rm(a)}{\leq} \frac{1+\snr}{1+\inr+\frac{\snr}{1+\inr} }  -1  
\notag\\&=    \frac{\inr \frac{\snr}{1+\inr} - \inr}{1+\inr+\frac{\snr}{1+\inr} }   
\notag\\&\leq \frac{\inr \frac{\snr}{1+\inr} }{1+\inr+\frac{\snr}{1+\inr} }   
\notag\\&\leq \min\left( \inr, \frac{\snr}{1+\inr} \right) 
\notag\\&\stackrel{\rm(b)}{=} \frac{\snr}{1+\inr},
\label{eq: d2bound {sec:gaps:WeakType2:2R1+R2 face}}
\end{align}
where inequalities follow from: 
(a) using definition of the regime in \eqref{eq: condition for Weak2}, and
(b) using $\frac{\snr}{1+\inr}\le \inr$.

As for the derivation in Section~\ref{sec:gaps:WeakType1:2R1+R2 face}, another key bound is
\begin{align}
(1+\snr\delta_1)(1+\snr_{4,b,t})
   &\leq \left(1+\frac{\snr}{1+\inr}\right) 
         \left( \frac{(1+\inr +\frac{\snr}{1+\inr})^{2} }{1+\snr} \right)^{1-t}
\notag\\&\leq  \left(1+\frac{\snr}{1+\inr}\right)\frac{(1+\inr +\frac{\snr}{1+\inr})^{2} }{1+\snr} 
\notag\\&\stackrel{\rm(a)}{\leq}  \frac{1+\inr+\snr}{1+\snr} \cdot \frac{(1+2\inr)^2}{1+\inr}% 
\end{align}
where inequalities follow from: 
(a) using $\frac{\snr}{1+\inr}\le \inr$, and
(b) using $\frac{1+\inr+\snr}{1+\snr} \le 2$ and $\frac{1+2\inr}{1+\inr} \le 2$.
Similarly, we have
\begin{align}
(1+\snr_{4,a,t})(1+\snr\delta_2) 
   &\leq \frac{1+\snr}{\left(1+\frac{\snr}{1+\inr}\right)^{t} \left( 1+\inr+\frac{\snr}{1+\inr}\right)^{1-t}}
         \left( \frac{1+\inr+\frac{\snr}{1+\inr}}{1+\frac{\snr}{1+\inr} }\right)^{1-t} 
\notag\\&= \frac{1+\snr}{1+\frac{\snr}{1+\inr}} \leq  1+\inr.
\label{eq: d2*N1bound {sec:gaps:WeakType2:2R1+R2 face}}
\end{align}

By using~\eqref{eq:achregion for par:moderate dminS1lowerbound}, the minimum distance for $S_1$ can be bounded as
\begin{align*}
\frac{d_{\min({S}_1)}^2}{12 \ \kappa_{\gamma,N_1,N_2}^2 }& \geq 
\frac{1-\max(\delta_1,\delta_2)} {1+\snr \delta_1+\inr\delta_2}
\min\left( 
\frac{\inr}{N_2^2-1}, 
\frac{\snr}{N_1^2N_2^2-1}
\right)
\\&
\stackrel{\rm (a)}{ \geq}\frac{\frac{\inr}{1+\inr}}{1+\snr \delta_1+\inr \delta_2}
\min\left( 
\frac{\inr}{\snr_{4,b,t}}, 
\frac{\snr}{2\inr}
\right)
\\& 
\stackrel{\rm (b)}{ \geq} \min\left( 
\frac{\frac{\inr^2}{1+\inr}}{2(1+\snr \delta_1)(1+\snr_{4,b,t})}, 
\frac{\snr}{2(1+\snr+2\inr)}
\right)
\\&\stackrel{\rm (c)}{ \geq}   \min\left( 
{
\frac{\inr^2 (1+\snr)}{2 (1+\inr+\snr) (1+2\inr)^2} }, 
{\frac{1+\inr}{2(2+3\inr)}} 
\right)  \\&
\stackrel{ \rm (d)}{\geq}   \min \left(  \frac{3}{8}, \frac{1}{6} \right)=\frac{1}{6} 
\end{align*}
where inequalities follow from: 
(a) $\max(\delta_1, \delta_2) \leq \frac{1}{1+\inr}$  and  from \eqref{eq: N2bound {sec:gaps:WeakType2:2R1+R2 face}} and \eqref{eq: N1*N2bound {sec:gaps:WeakType2:2R1+R2 face}} we have that  $N_2^2-1 \leq \snr_{4,b,t}$ and $N_1^2N_2^2-1 \leq 2\inr$, 
(b) { from \eqref{eq: d2bound {sec:gaps:WeakType2:2R1+R2 face}} $\delta_2 \le \frac{ 1}{1+\inr}$,
(c)  using \eqref{eq: d2bound {sec:gaps:WeakType2:2R1+R2 face}} we have $\snr_{4,b,t}(2+\snr \delta_1) \le 2(1+\snr_{4,b,t})(1+\snr \delta_1)$, 
(d) $\snr \ge (1+\inr)$ and $\inr \ge 1$. 

Similarly,
\begin{align*}
\frac{d_{\min({S}_2)}^2}{12\ \kappa_{\gamma,N_1,N_2}^2 }& \geq 
\frac{1-\max(\delta_1,\delta_2)} {1+\snr \delta_2+\inr\delta_1}
\min\left( 
\frac{\inr}{N_1^2-1}, 
\frac{\snr}{N_1^2N_2^2-1}
\right)
\\&
\stackrel{\rm (a)}{ \geq}  
\min\left( 
\frac{\frac{\inr^2}{1+\inr}}{2(1+\snr \delta_2)(1+\snr_{4,a,t})}, 
\frac{\snr}{2(1+\snr +2\inr)}
\right)\\
& \stackrel{\rm (b)}{\geq}
\min\left( 
\frac{\inr^2}{2(1+\inr)(1+\inr)}, 
{\frac{1+\inr}{2(2+3\inr)}}
\right) 
\\&
\stackrel{\rm (c)}{\geq}  \min \left(\frac{1}{8}, \frac{1}{6} \right) = \frac{1}{8}
\end{align*}
where inequalities follow from: 
(a)   \eqref{eq:parameters 2R1+R2 face weak type 2: first choice N1}  we have $N_1^2-1 \leq \snr_{4,a,t}$,   from \eqref{eq: d2bound {sec:gaps:WeakType2:2R1+R2 face}} and \eqref{eq:parameters 2R1+R2 face weak type 2: first choice delta1} $\max(\delta_1,\delta_2) \leq \frac{1}{1+\inr}$ and from \eqref{eq: N1*N2bound {sec:gaps:WeakType2:2R1+R2 face}} $N_1^2N_2^2-1 \leq \inr$, 
(b)  $\snr_{4,a,t} (2+\snr \delta_2) \leq 2 (1+\snr_{4,a,t}) (1+\snr \delta_2)$ and \eqref{eq: d2*N1bound {sec:gaps:WeakType2:2R1+R2 face}}, and
(c) from $\snr \ge (1+\inr)$ and $\inr \ge 1$.

Hence, the minimum distance in \eqref{eq: min  WeakType2:2R1+R2 face with loglog gap} is bounded by
\begin{align*}
\min_{i \in[1:2] } \frac{d_{\min({S}_i)}^2}{12 \ \kappa_{\gamma,N_1,N_2}^2 }  \ge \frac{1}{8}.
\end{align*}

\section{Constant Gap Derivation for Regime Weak2}
\label{app:sec:CONSTANT GAP FOR WEAK 2}

\subsection{Another Inner Bound for $\mathcal{R}_{R_1+R_2}$}
\label{sec: weak 2 r1+r2 constant gap}
In order to approximately achieve the points in $\mathcal{R}_{R_1+R_2}^{(\text{\ref{sec: weak: weak type 2}})}$ in~\eqref{eq:outer bound in weak2 sumrate} we pick 
\begin{subequations}
\begin{align}
N_1 &= \points\left({ \frac{1}{k}}\snr_{3,a,t}\right), \
\snr_{3,a,t} \ \text{in~\eqref{eq:achregion for par:weak type 2 choiceN1 take1}},
\label{eq:achregion for par:weak type 2 choiceN1}
\\
N_2 &= \points\left({ \frac{1}{k}} \snr_{3,b,t}\right), \
\snr_{3,b,t} \ \text{in~\eqref{eq:achregion for par:weak type 2 choiceN2 take1}}
\label{eq:achregion for par:weak type 2 choiceN2}
\\
\delta_1 &:  \mug \left(\snr \delta_1\right)=\mug\left(\snr_{3,a,t}\right)
\Longleftrightarrow  
\delta_1 =\frac{\snr_{3,a,t}}{\snr}, 
\label{eq: choice delta1: weak 2 R1+R2}
\\
\delta_2& : \mug \left( \snr \delta_2\right)=\mug\left(\snr_{3,b,t}\right)
\Longleftrightarrow
\delta_2 =\frac{\snr_{3,b,t}}{\snr}.
\label{eq: choice delta2: weak 2 R1+R2}
\end{align}
where $k$ is a parameter that we will tune in order to satisfy 
the non-overlap condition in Proposition~\ref{prop:combNOOUTAGE}.
\label{eq:parameters R1+R2 face weak type 2}
\end{subequations} 
Indeed, in order to check whether we can use the bound in~\eqref{eq:achregion for par:moderate all parameters dmin P2 min dminpippo} we must check whether the condition in~\eqref{eq:achregion for par:moderate all parameters dmin P2 min dmin condcondcond} holds.
To simplify the analytical computations we choose to satisfy instead 
\begin{align*}
 \frac{(1-\delta_{i'})N_{i'}^2}{N_{i'}^2-1} \leq k
  \leq \frac{\snr}{\inr} \ \frac{(1-\delta_i)}{N_{i}^2-1} \quad \forall(i,i')\in \{(1,2),(2,1) \},
\end{align*}
for some $k$; since $\frac{(1-\delta_{i'})N_{i'}^2}{N_{i'}^2-1} \leq \frac{N_{i'}^2}{N_{i'}^2-1}  \leq \frac{4}{3}$ for all $N_{i'} \geq 2$, we set $\frac{4}{3}:=k$. In other words, we accept an increase in gap of $\log(k) = \log(4/3)$, due to the reduction of the number of points of the discrete part of the mixed inputs from $\points(x)$ to $\points(3x/4)$ for some `SNR' $x$, for ease of computations.

Therefore, for the rest of this section instead of checking condition in~\eqref{eq:achregion for par:moderate all parameters dmin P2 min dmin condcondcond} we will check the simpler condition
\begin{align}
\frac{4}{3} \inr  \leq \frac{\snr(1-\delta_i)}{N_{i}^2-1} \quad \forall i \in [1:2].
\label{eq: non-overlap  condition with 4/3} 
\end{align}

The gap between the outer bound region in~\eqref{eq:outer bound in weak2 sumrate} and the achievable rate in~\eqref{eq: rates mixed inputs BEFORE Union} with the parameters in~\eqref{eq:parameters R1+R2 face weak type 2} is 
\begin{align*}
\Delta_{R_1} &= 2 \mug \left( \snr_{3,a,t} \right)
  -\log\left(\points \left( 3/4 \ \snr_{3,a,t} \right)\right)
  -\mug(\snr_{3,a,t})+\Delta_{\eqref{eq: rates mixed inputs BEFORE Union}}\\
  &\leq \log \left(\frac{8}{3}\right)+\Delta_{\eqref{eq: rates mixed inputs BEFORE Union}},
\end{align*}
and similarly
\begin{align*}
\Delta_{R_2}  \leq \log \left(\frac{8}{3}\right)+\Delta_{\eqref{eq: rates mixed inputs BEFORE Union}}.
\end{align*}
We are then left with bounding $\Delta_{\eqref{eq: rates mixed inputs BEFORE Union}}$, which depends on minimum distances of the received sum-set constellations.
From \eqref{eq: N1bound {sec:gaps:WeakType2:R1+R2 face}}-\eqref{eq: N2bound {sec:gaps:WeakType2:R1+R2 face}} we have
\begin{align*}
N_1^2-1 \le \frac{3}{4} \snr_{3,a,t}  \le \frac{3}{4} \frac{\snr}{1+\inr}, \ \text{from~\eqref{eq: N1bound {sec:gaps:WeakType2:R1+R2 face}}},
\\
N_2^2-1 \le \frac{3}{4} \snr_{3,b,t}  \le \frac{3}{4} \frac{\snr}{1+\inr}, \ \text{from~\eqref{eq: N2bound {sec:gaps:WeakType2:R1+R2 face}}},
\end{align*}
and thus
\begin{align}
\frac{\snr (1-\delta_{i})}{N_{i}^2-1}
  &\stackrel{\max(\delta_1,\delta_2) \leq \frac{1}{1+\inr}}{\geq} 
\inr \frac{ \frac{\snr}{1+\inr}}{N_{i}^2-1}
%\notag
\geq \frac{4}{3}  \inr,
\label{eq: dmin cond1 {sec:gaps:WeakType2:R1+R2 face}}
\end{align}
as needed in~\eqref{eq: non-overlap  condition with 4/3}.

Therefore, by~\eqref{eq:achregion for par:moderate all parameters dmin P2 min dminpippo}, for $d_{\min ({S}_1)}^2$
we have that 
\begin{align}
\frac{d_{\min ({S}_1)}^2}{12}
  &=\frac{1}{1+\snr \delta_1+\inr \delta_2} \min \left( \frac{(1-\delta_1)\snr}{N_1^2-1}, \frac{(1-\delta_2)\inr}{N_2^2-1} \right)
\notag\\&\stackrel{\text{(a)}}{\geq}
\frac{ \frac{\inr}{1+\inr}}{1+\snr \delta_1+ \frac{\inr}{1+\inr}} \min \left( \frac{\snr}{N_1^2-1}, \frac{\inr}{N_2^2-1} \right)
\notag\\&\stackrel{\text{(b)}}{\geq}
\frac{ \frac{\inr}{1+\inr}}{1+\snr \delta_1+ \frac{\inr}{1+\inr}} { \frac{4}{3}} \min \left(1+\inr, \frac{\inr}{\snr_{3,b,t}} \right)
\notag\\&\stackrel{\text{(c)}}{\geq}
{ \frac{4}{3}}\min \left(
\frac{\inr(1+\inr)}{1+\snr + 2\inr}, 
\frac{ \frac{\inr^2}{1+\inr} }{2(1+\snr \delta_1)(1+\snr_{3,b,t})} \right)
\notag\\&\stackrel{\text{(d)}}{\geq} { \frac{4}{3}}
\min \left(
\frac{\inr(1+\inr)}{1+\snr + 2\inr}, 
\frac{\inr^2}{2(1+\inr)(1+\inr+\frac{\snr}{1+\inr})} \right)
\notag\\&\stackrel{\text{(e)}}{\geq} { \frac{4}{3}}
\min \left(
\frac{\inr(1+\inr)}{1+ 3\inr+\inr^2}, 
\frac{\inr^2}{2(1+\inr)(1+2\inr)} \right) \notag
\\&\geq { \frac{4}{3}}
\min \left(
\frac{2}{5}, 
\frac{1}{12} \right) 
= \frac{1}{9},
\label{eq: dmin(S1)bound {sec:gaps:WeakType2:R1+R2 face}}
\end{align}
where the inequalities follows from:
(a) $\max(\delta_1,\delta_2) \leq  \frac{1}{1+\inr}$,
(b) from~\eqref{eq: N1bound {sec:gaps:WeakType2:R1+R2 face}} and~\eqref{eq: N2bound {sec:gaps:WeakType2:R1+R2 face}}
(c) $\max(\delta_1,\delta_2) \leq  \frac{1}{1+\inr}$,
(d) from~\eqref{eq: N1*N2 {sec:gaps:WeakType2:R1+R2 face}}, and
(e) from $1\leq \inr \leq \snr \leq \inr(1+\inr)$.

By symmetry, $\frac{d_{\min ({S}_2)}^2}{12}$ is bounded in the same way,
thus 
\begin{align}
\min_{i\in[1:2]} \frac{d_{\min ({S}_i)}^2}{12} \geq \frac{1}{9}.
\label{eq: dmin WeakType2:R1+R2 face}
\end{align}

Finally the gap for this face is
\begin{align}
{\gap}_{\eqref{eq:sec:gaps:WeakType2:R1+R2 face}}
        &\leq \max(\Delta_{R_1},\Delta_{R_2})
        = \log \left(\frac{8}{3}\right)+\Delta_{\eqref{eq: rates mixed inputs BEFORE Union}}
\notag\\& \leq \log \left(\frac{8}{3}\right)+ \frac{1}{2} \log \left( \frac{\pi \eu}{3}\right)+\frac{1}{2} \log \left(1+ 9\right)
\notag\\&=\frac{1}{2} \log \left( \frac{640 \pi \eu}{27}\right) 
  \approx 3.83~\text{bits}.
\label{eq:sec:gaps:WeakType2:R1+R2 face}
\end{align}

\subsection{Another Inner Bound for $\mathcal{R}_{2R_1+R_2}$}

We choose the mixed input parameters as
\begin{subequations}  
\begin{align}
N_1 &= \points\left( \frac{3}{4}\ \frac{\snr-\inr}{1+\inr} \right), \
\label{eq:achregion for par:weak type 2 choiceN1:2R_1R+2:take 2}
\\
N_2 &= \points\left( \frac{3}{4}\snr_{4,b,t}\right), \
\snr_{4,b,t} := \left( \frac{(1+\inr +\frac{\snr}{1+\inr})^{2} }{1+\snr} \right)^{1-t}-1 
\stackrel{\text{by eq.\eqref{eq: N2bound {sec:gaps:WeakType2:2R1+R2 face}}}}
\leq \frac{\snr}{1+\inr},
\label{eq:achregion for par:weak type 2 choiceN2:2R_1R+2:take 2}
\\
\delta_1&= \frac{\snr_{4,a,t} }{\snr}
\stackrel{\text{by  eq.\eqref{eq: N1bound {sec:gaps:WeakType2:2R1+R2 face}}}}{\leq}
\frac{ \inr }{ \snr } ,
\label{eq: choice of delta1: for 2R1+R2: take 2}
\\
\delta_2&=\frac{ 1+\inr+\frac{\snr}{1+\inr} }{ \left(1+\frac{\snr}{1+\inr}\right)(1+\snr)} 
\stackrel{\text{by  eq.\eqref{eq: condition for Weak2}}}\le\frac{1}{1+\inr+\frac{\snr}{1+\inr}} \le \frac{1}{1+\inr} 
\label{eq: choice of delta2: for 2R1+R2:take 2}, 
\end{align}
\label{eq:parameters 2R1+R2 face weak type 2: second choice}
\end{subequations}
where the factor $\frac{3}{4}$ in the number of points  appears for the same reason as in Section~\ref{sec: weak 2 r1+r2 constant gap}.

An inequality we will need is
\begin{align}
\frac{\snr \delta_2}{1+\inr \delta_1}& \stackrel{\rm (a)}{\ge} \frac{\snr \frac{\left(1+\inr+\frac{\snr}{1+\inr} \right) (1+\inr)}{(1+\inr+\snr)(1+\snr)}} {\frac{\inr}{\snr} \frac{(1+\snr) }{\left(1+\frac{\snr}{1+\inr}\right)^{t} \left( 1+\inr+\frac{\snr}{1+\inr}\right)^{1-t}} }\notag \\
&= \frac{\snr^2}{(1+\snr)^2}  \frac{\left(1+\inr +\frac{\snr}{1+\inr} \right) \left(1+\frac{\snr}{1+\inr}\right)^t \left(1+\inr +\frac{\snr}{1+\inr} \right)^{1-t} }{\inr  \left( 1+\frac{\snr}{1+\inr}\right)}\notag\\
&= \frac{\snr^2}{(1+\snr)^2}  \frac{\left(1+\inr +\frac{\snr}{1+\inr} \right) (1+\inr)^{1-t}  \left(1+\inr +\frac{\snr}{1+\inr} \right)^{1-t} }{\inr  \left( 1+\inr+\snr\right)^{1-t}}
\notag\\
& \stackrel{\rm (b)}{\ge} \frac{3}{4}   \left( \frac{1+\inr +\frac{\snr}{1+\inr} }{ 1+\frac{\snr}{1+\inr}} \right)^{1-t} \label{eq: bound on snr delta 1 div inr delta2}
\end{align}
where the inequalities follow from: 
(a) plugin in values of $\delta_1$ and $\delta_2$ and lower bounding the denominator, and
(b) using $\snr\ge 1$  we have that $\frac{\snr^2}{(1+\snr)^2} \ge \frac{1}{4}$ and using $\snr \ge (1+\inr) $ we have$\frac{1+\inr+\frac{\snr}{1+\inr} }{\inr} \ge \frac{2+\inr}{\inr} \ge 3$.

Another inequality we will need is
\begin{align}
\snr \delta_2 &= \snr \frac{\left(1+\inr+\frac{\snr}{1+\inr} \right) (1+\inr)}{(1+\inr+\snr)(1+\snr)} \notag\\
&  \le \frac{\left(1+\inr+\frac{\snr}{1+\inr} \right) (1+\inr)}{(1+\inr+\snr)}\notag\\
& =\frac{\left(1+\inr+\frac{\snr}{1+\inr} \right) }{(1+\frac{\snr}{1+\inr})} \notag\\\
&\le \frac{(1+2\inr)(1+\inr)}{\snr} \label{eq: bound on snr delta2}
\end{align}
where the last inequality comes from using $\frac{\snr}{1+\inr} \le \inr$ and dropping one in the denominator.

\paragraph*{Gap for $\mathcal{R}_{2R_1+R_2}$}

The gap between the outer bound in~\eqref{eq:outer bound in weak2 otherrate} and the achievable rate in Proposition~\ref{prop:ach-with-mixedinput} with the choice of parameters in~\eqref{eq:parameters 2R1+R2 face weak type 2: second choice} 
is 
\begin{align*}
\Delta_{R_1} &=  \mug\left( \snr_{4,a,t}\right)+ \mug \left( \frac{\snr}{1+\inr}\right)
- \log\left(\points\left( \frac{3}{4}\ \frac{\snr-\inr}{1+\inr} \right) \right) - \mug \left(\snr \delta_1 \right)
+\Delta_{\eqref{eq: rates mixed inputs BEFORE Union}}
\\
& \le \log(2) +\frac{1}{2}\log(2)%
+\Delta_{\eqref{eq: rates mixed inputs BEFORE Union}},
\end{align*} 

\[
\mug \left( \frac{\snr}{1+\inr}\right)-\mug\left(  \frac{3}{4}\ \frac{\snr-\inr}{1+\inr} \right)
=\frac{1}{2} \log\frac{1+\inr+\snr}{1+\inr}\frac{1+\inr}{1+\inr/4+3\snr/4}
\leq \frac{1}{2} \log\frac{1+2\snr}{1+\snr} \leq \frac{1}{2} \log(2),
\]

and similarly
\begin{align*}
\Delta_{R_2}&=\mug(\snr_{4,b,t}) + 
 \frac{1-t}{2} \log \left( \frac{ \left (1+\inr+\frac{\snr}{1+\inr} \right)(1+\inr) }{1+\inr+\snr} \right)+t c
\\&
\quad -\log \left( \points(\snr_{4,b,t}) \right) -\mug\log\left(\frac{\snr \delta_2}{1+\inr \delta_1}\right)
-\frac{1}{2}\log(2)
 +\Delta_{\eqref{eq: rates mixed inputs BEFORE Union}}
\\
&\le
  \log(2)+\frac{1}{2}\log\left(\frac{4}{3}\right)
 +\frac{1}{2}\log \left(\frac{4}{3}\right)
 +\log(2) 
 -\frac{1}{2}\log(2)
 +\Delta_{\eqref{eq: rates mixed inputs BEFORE Union}}
=\frac{1}{2}\log \left(\frac{2^7}{3^2}\right)
+\Delta_{\eqref{eq: rates mixed inputs BEFORE Union}}
\end{align*}
where we have used $tc \le \log(2)$ and the bound in~\eqref{eq: bound on snr delta 1 div inr delta2};
the term `$-\frac{1}{2}\log(2)$' is because of the definition of~$\Delta_{\eqref{eq: rates mixed inputs BEFORE Union}}$ that assumed $\max(\delta_1,\delta_2) \leq \frac{1}{1+\in}$, which is not the case here.
% . 

So, we are left with bounding $\Delta_{\eqref{eq: rates mixed inputs BEFORE Union}}$, which depends on the minimum distances of the received constellations. 
We must verify the condition in~\eqref{eq: non-overlap  condition with 4/3} at each receiver.% 

For receiver~1  we have
\begin{align}
\frac{\snr(1-\delta_1)}{N_1^2-1} 
        &\stackrel{\text{from eq.\eqref{eq: choice of delta1: for 2R1+R2: take 2}}}{\ge}  
        \frac{\snr-\inr}{N_1^2-1} %
\notag\\&\stackrel{\text{from eq.\eqref{eq:achregion for par:weak type 2 choiceN1:2R_1R+2:take 2}}}{\ge}  
        \frac{\snr-\inr}{ \frac{3}{4}\ \frac{\snr-\inr}{1+\inr}}
\notag\\&=\frac{4}{3}(1+\inr) \geq \frac{4}{3}\inr,
\label{eq: non-overlap at S1: Weak 2 face 2R1+R2}
\end{align}
and therefore
\begin{align*}
\frac{d_{\min ({S}_1)}^2}{12}&=\frac{1}{1+\snr \delta_1+\inr \delta_2} \min \left( \frac{(1-\delta_1)\snr}{N_1^2-1}, \frac{(1-\delta_2)\inr}{N_2^2-1} \right)\\
&\stackrel{\rm (a)}{\ge}  \frac{1}{1+\snr \delta_1+\inr \frac{1}{1+\inr}} \min \left( \frac{4}{3} (1+\inr), \frac{ \frac{\inr^2}{1+\inr}}{\frac{3}{4} \snr_{4,b,t} } \right)\\
&= \frac{4}{3}\min \left( \frac{ (1+\inr)}{1+\snr \delta_1+\inr \frac{1}{1+\inr}}, \frac{ \frac{\inr^2}{1+\inr}}{ \left(1+\snr \delta_1+\inr \frac{1}{1+\inr}\right) \snr_{4,b,t} } \right)\\
&\stackrel{\rm (b)}{\ge}   \frac{4}{3}\min \left(\frac{1+\inr }{2+\snr \frac{\inr}{\snr}} , \frac{\frac{\inr^2}{1+\inr}}{2(1+\snr \delta_1)(1+\snr_{4,b,t})}  \right)\\
& \stackrel{\rm (c)}{\ge}  \frac{4}{3} \min \left(\frac{1+\inr }{2+\inr } , \frac{\frac{\inr^2}{1+\inr}}{4 (1+2\inr)}  \right)\\
& \stackrel{\rm (d)}{\ge}  \frac{4}{3} \min \left(\frac{2}{3}, \frac{1}{24} \right)=\frac{1}{18},
\end{align*}
where the bounds are obtained by: 
(a) using \eqref{eq: non-overlap at S1: Weak 2 face 2R1+R2} and $\delta_2 \leq \frac{1}{1+\inr}$ and \eqref{eq: N2bound {sec:gaps:WeakType2:2R1+R2 face}}, 
(b) using $\delta_2 \le \frac{\inr}{\snr}$ form \eqref{eq: N1bound {sec:gaps:WeakType2:2R1+R2 face}} and $(1+\snr \delta_2+\inr \frac{1}{1+\inr}) \le 2 (1+\snr \delta_2) (1+\snr_{4,b,t})$, 
%c) 
(c) using \eqref{eq: choice of delta1: for 2R1+R2: take 2} we have $(1+ \snr \delta_2)(1+\snr_{4,b,t})= (1+ \snr_{4,a,t})(1+\snr_{4,b,t})$  and then using \eqref{eq: N1*N2bound {sec:gaps:WeakType2:2R1+R2 face}}, and
(d) come from minimizing over  $\inr \geq 1$. 

For receiver~2  we have
\begin{align*}
\frac{\snr(1-\delta_2)}{N_2^2-1} \ge \frac{\snr  \frac{\inr}{1+\inr} }{N_2^2-1}  \ge \frac{\snr  \frac{\inr}{1+\inr} }{ \frac{3}{4} \frac{\snr}{1+\inr}}=\frac{4}{3} \inr,
\end{align*}
and therefore %
\begin{align*}
\frac{d_{\min ({S}_2)}^2}{12}&=\frac{1}{1+\snr \delta_2+\inr \delta_1} \min \left( \frac{(1-\delta_2)\snr}{N_2^2-1}, \frac{(1-\delta_1)\inr}{N_1^2-1} \right)\\
& \stackrel{\rm (a)}{\ge}    \frac{1}{1+\snr \delta_2+\inr  \frac{\inr}{\snr}} \min \left( \frac{(1-\frac{1}{1+\inr})\snr}{N_2^2-1}, \frac{(1-\frac{\inr}{\snr})\inr}{N_1^2-1} \right)\\
& \stackrel{\rm (b)}{\ge} \frac{4}{3}  \frac{1}{1+\snr \delta_2+\inr  \frac{\inr}{\snr}} \min \left( \frac{(1-\frac{1}{1+\inr})\snr}{ \frac{\snr}{1+\inr} }, \frac{(1-\frac{\inr}{\snr})\inr (1+\inr)}{\snr-\inr} \right)\\
& = \frac{4}{3} \frac{1}{1+\snr \delta_2+\inr  \frac{\inr}{\snr}} \min \left( \inr, \frac{\inr (1+\inr)}{\snr} \right)\\
& =  \frac{4}{3} \min \left( \frac{\inr}{1+\snr \delta_2+\inr  \frac{\inr}{\snr}}, \frac{\inr (1+\inr)}{(1+\snr \delta_2+\inr  \frac{\inr}{\snr}) \snr} \right)\\
& \stackrel{\rm (c)}{\ge}  \frac{4}{3} \min \left( \frac{\inr}{1+\snr \frac{1}{1+\inr}+\inr  \frac{\inr}{\snr}}, \frac{\inr (1+\inr)}{(1+\snr \delta_2+\inr  \frac{\inr}{\snr}) \snr} \right)\\
& \stackrel{\rm (d)}{\ge} \frac{4}{3}   \min \left( \frac{\inr}{1+2\inr  }, \frac{\inr (1+\inr)}{(1+\frac{(1+2\inr)(1+\inr)}{\snr}+\inr  \frac{\inr}{\snr}) \snr} \right)\\
&= \min \left( \frac{\inr}{1+2\inr  }, \frac{\inr (1+\inr)}{(\snr+(1+2\inr)(1+\inr)+\inr^2  ) } \right)\\
&\stackrel{\rm (e)}{\ge} \frac{4}{3}\min \left( \frac{\inr}{1+2\inr  }, \frac{\inr (1+\inr)}{(\inr(1+\inr)+(1+2\inr)(1+\inr)+\inr^2  ) } \right)\\
&\ge \min \left( \frac{1}{3}, \frac{2}{10} \right)=\frac{4}{15}
\end{align*}
where the bounds are obtained by: 
(a) using $\delta_2 \le \frac{1}{1+\inr}$ and $\delta_1 \le \frac{\inr}{\snr}$, 
(b) from \eqref{eq:achregion for par:weak type 2 choiceN1:2R_1R+2:take 2} we have that $N_1^2-1 \le \frac{3}{4}\frac{\snr-\inr}{1+\inr}$, (c) using $\delta_2\le \frac{1}{1+\inr}$, 
(d) used bound in \eqref{eq: bound on snr delta2}, and 
(e) used bound $\snr \le \inr (1+\inr)$.

So, finally the gap is 
\begin{align}
{\gap}_{\eqref{eq:gap weak2 otherrate take2}}
        &\leq \max( \Delta_{R_1},\Delta_{R_2}) 
\notag\\&=
\frac{1}{2}\log \left(\frac{2^7}{3^2}\right)
+\frac{1}{2} \log \left(\frac{\pi \eu}{3 }\right)+\frac{1}{2} \log \left(1+\frac{15}{4}\right)
\notag\\&=\frac{1}{2}\log \left(\frac{608 \ \pi \eu}{27}\right) \approx 3.79~\text{bits} 
\label{eq:gap weak2 otherrate take2}
\end{align}

\paragraph*{Overall Constant Gap for Weak 1}
Therefore, the overall  gap for Weak 1 is 
\begin{align*}
{\gap} \le \max({\gap}_{\eqref{eq:sec:gaps:WeakType2:R1+R2 face}},{\gap}_{\eqref{eq:gap weak2 otherrate take2}})={\gap}_{\eqref{eq:sec:gaps:WeakType2:R1+R2 face}}.
\end{align*}

\bibliography{refs}
\bibliographystyle{IEEEtran}

\end{document}